%% file: mainArxivFinalNoTypo.tex
\documentclass[11pt]{article}

\usepackage[english]{babel}
\usepackage{a4wide}
\usepackage[latin1]{inputenc}
\usepackage[T1]{fontenc}
\usepackage{color}
\usepackage{eurosym}
\usepackage[top=4cm,bottom=4cm,left=3cm,right=3cm]{geometry}
\usepackage{fancyhdr}
\usepackage{float}
\usepackage{amssymb,amsmath, amsthm}
\usepackage{url,graphicx}

\usepackage{booktabs}
\usepackage{stmaryrd} 
\usepackage{mathenv}
\usepackage{dsfont}
\usepackage{amssymb}
\usepackage{latexsym}
\usepackage{amsmath}
\usepackage{color}
\usepackage{comment}
\usepackage{amsthm}
\usepackage{bm}
\usepackage{bbm} 
\usepackage[bbgreekl]{mathbbol} 
\usepackage{authblk}
\usepackage{multirow}
\newif\ifrs
\rstrue
\ifrs \usepackage{mathrsfs} \fi  
\newif\ifcol
\coltrue 

\newtheorem{theorem*}{Theorem}[section]

\newtheorem{note*}[theorem*]{Note}
\newtheorem{lemma*}[theorem*]{Lemma}
\newtheorem{definition*}[theorem*]{Definition}
\newtheorem{proposition*}[theorem*]{Proposition}
\newtheorem{corollary*}[theorem*]{Corollary}
\newtheorem{remark*}[theorem*]{Remark}
\newtheorem{example*}[theorem*]{Example}
\newtheorem*{conditionE'3}{Condition (E3')}

\numberwithin{equation}{section}
\newif\ifcol
\coltrue
\ifcol
\newcommand{\colorr}{\color[rgb]{0.8,0,0}}

\newcommand{\colorn}{\color[rgb]{1,1,1}}

\else

\newcommand{\colorr}{\color{black}}

\newcommand{\colorn}{\color{black}}

\fi
\excludecomment{en-text}
\includecomment{jp-text}
\includecomment{comment}
\input nakamacro27.tex

\newcommand{\Tau}{\mathrm{T}}
\newcommand{\var}{\operatorname{Var}}

\DeclareMathOperator{\ind}{\perp \!\!\! \perp}

\newcommand{\bop}[1]{O_\proba\l(#1\r)}

\newcommand{\ti}[1]{t_{#1}^n}

\newcommand{\qtermlocal}[2]{\Delta_B^{-1}\int_{#1}^{#2}{\alpha_s^{-1}ds}\Big\{\int_{#1}^{#2}{\sigma_s^4 \alpha_sds} + \sum_{#1 < s \leq #2} \Delta J_s^2 (\sigma_s^2\alpha_s + \sigma_{s-}^2\alpha_{s-}) \Big\}}

\theoremstyle{plain}
\newtheorem{RK0}{Theorem}
\newtheorem{RK}[RK0]{Theorem}
\newtheorem{RKcor}[RK0]{Corollary}
\newtheorem{RKth3}[RK0]{Theorem}
\newtheorem{propAVARRK}[RK0]{Proposition}
\newtheorem{RKjumps}[RK0]{Theorem}

\newtheorem{QMLE}[RK0]{Theorem}
\newtheorem{QMLEcor}[RK0]{Corollary}
\newtheorem{QMLEjumps}[RK0]{Theorem}
\newtheorem{propAVARQMLE}[RK0]{Proposition}
\newtheorem{QMLEcorRobust}[RK0]{Proposition}
\newtheorem{QMLEcorRobust2}[RK0]{Proposition}
\newtheorem{RKcorRobust}[RK0]{Proposition}
\newtheorem{RKcorRobust2}[RK0]{Proposition}

\theoremstyle{remark}

\begin{document}
 \title{Efficient asymptotic variance reduction when estimating volatility in high frequency data\footnote{We would like to thank Markus Bibinger, Jia Li, Dacheng Xiu, Jean Jacod, Yacine A\"{i}t-Sahalia (the Editor), two anonymous referees and an anonymous Associate Editor, the participants of Keio Econometrics Workshop, the Workshop on Portfolio dynamics and limit order books in Ecole Centrale Paris, The Quantitative Methods in Finance 2016 Conference in Sydney for helpful discussions and advice. The research of Yoann Potiron is supported by a special private grant from Keio University and Japanese Society for the Promotion of Science Grant-in-Aid for Young Scientists No. 60781119. All financial data is provided by the Chair of Quantitative Finance of the Ecole Centrale Paris. The research of Simon Clinet is supported by CREST Japan
Science and Technology Agency and a special grant from Keio University.}}
\author{Simon Clinet\footnote{Faculty of Economics, Keio University. 2-15-45 Mita, Minato-ku, Tokyo, 108-8345, Japan. Phone:  +81-3-5427-1506. E-mail: clinet@keio.jp website: http://user.keio.ac.jp/\char`\~ clinet} \footnote{Japan Science and Technology Agency, CREST.}  and Yoann Potiron\footnote{Faculty of Business and Commerce, Keio University. 2-15-45 Mita, Minato-ku, Tokyo, 108-8345, Japan. Phone:  +81-3-5418-6571. E-mail: potiron@fbc.keio.ac.jp website: http://www.fbc.keio.ac.jp/\char`\~ potiron}}
\date{This version: \today} 

\maketitle

\begin{abstract}
This paper shows how to carry out efficient asymptotic variance reduction when estimating volatility in the presence of stochastic volatility and microstructure noise with the realized kernels (RK) from \cite{barndorff2008designing} and the quasi-maximum likelihood estimator (QMLE) studied in \cite{xiu2010quasi}. To obtain such a reduction, we chop the data into $B$ blocks, compute the RK (or QMLE) on each block, and aggregate the block estimates. The ratio of asymptotic variance over the bound of asymptotic efficiency converges as $B$ increases to the ratio in the parametric version of the problem, i.e. 1.0025 in the case of the fastest RK Tukey-Hanning 16 and 1 for the QMLE. The impact of stochastic sampling times and jump in the price process is examined carefully. The finite sample performance of both estimators is investigated in simulations, while empirical work illustrates the gain in practice.

\end{abstract}
\textbf{Keywords}: high frequency data ; jumps ; market microstructure noise ; integrated volatility ; quasi-maximum likelihood estimator ; realized kernels ; stochastic sampling times
\\

\section{Introduction} \label{introduction}
Over the past decades, the availability of high frequency data has led to a better understanding of asset prices. The main object of interest, the quadratic variation, can be used for example as a proxy for the spot volatility or the volatility parameter of a time-varying model. Moreover, forecasts of future volatility can be improved with it. Without microstructure noise, the realized variance (RV) estimator (e.g. \cite{andersen2001distribution}, \cite{meddahi2002theoretical}, \cite{barndorff2002estimating}) is both consistent and efficient. The convergence rate $n^{1/2}$ and the asymptotic variance (AVAR) were established in \cite{genon1993estimation}, \cite{jacod1994limit} and \cite{jacod1998asymptotic} (see also \cite{zhang2001martingales} and \cite{mykland2006anova}).

\smallskip 
Under market frictions, the RV is no longer consistent. \cite{zhang2005tale} bring forward the Two-Scale Realized Volatility nonparametric estimator, the first consistent estimator in the presence of noise and with a relatively slow convergence rate of $n^{1/6}$. \cite{zhang2006efficient} modifies it to provide the Multi-Scale Realized Volatility (MSRV) which features the optimal rate of convergence $n^{1/4}$ as documented in \cite{gloter2001diffusions}. Other approaches consist in and are not limited to: pre-averaging (PAE) the observations (\cite{jacod2009microstructure}), \cite{barndorff2008designing} advocates for the realized kernels (RK) and \cite{xiu2010quasi} studies the quasi-maximum likelihood estimator (QMLE) which was originally considered in \cite{ait2005often} when volatility is constant. Those three approaches share the optimal rate property and only differ through edge-effects which impact their respective AVAR. 

\smallskip 
The nonparametric AVAR bound of efficiency is equal to $8 a_0 T^{\frac{1}{2}} \int_0^T \sigma_u^3 du$, where $T$ stands for the time horizon and $a_0^2$ corresponds to the noise variance. This was shown in \cite{reiss2011asymptotic} under the deterministic volatility and Gaussian noise setting, but it is commonly assumed that it stays true under stochastic volatility.  Subsequently, in a recent breakthrough
paper, \cite{altmeyer2015functional} found an estimator based on the spectral approach introduced in \cite{reiss2011asymptotic} which reaches the bound in a very general situation.  More recently, \cite{jacod2015microstructure} proposed an adapted version of the pre-averaging estimator using local estimates as in \cite{reiss2011asymptotic} which gave rise to estimators that are within 7\% of the bound.

\smallskip
To be fair when comparing several estimators, we need the candidates to be equipped with the same technology. Following closely the local technique used in \cite{reiss2011asymptotic} and more recently in \cite{jacod2015microstructure}, we aim to adapt accordingly the RK and the QMLE. Indeed, although both estimators behave remarkably well when volatility is constant, i.e. in the parametric case the ratio of AVAR over the bound of asymptotic efficiency is 1.0025 when considering the most efficient Tukey-Hanning 16 RK and 1 for the QMLE, they can actually be highly inefficient in the non-parametric setting as documented in the following of this introduction and in Section 2. Under time-varying volatility, we aim to reduce significantly their AVAR and make them efficient. Although it would reduce the AVAR the same way, we did not implement the local version of the MSRV. In fact, MSRV and RK are asymptotically equivalent in the sense that they share the same asymptotic variance when considering the same kernel (see Section 2.2 in \cite{bibinger2016inference}).

\smallskip
To reduce the variance, we divide the interval $\big[0,T \big]$ into $B$ non-overlapping regular blocks $\big[ 0, T/B \big]$, $\big[ T/B, 2 T/B \big]$, $\cdots$, $\big[ (B-1)T/B, T \big]$. We then compute the RK (QMLE) on each block, and take the sum of the $B$ estimates. We show that the nonparametric ratio of AVAR over the bound of efficiency converges to the parametric ratio as $B$ increases. More importantly for practical applications, the convergence is very fast, and the gain is already important in the case $B=2$ blocks. 

\smallskip
As an example, we focus on the RK Tukey-Hanning 16 and consider the (apparently innocuous) block constant model $\sigma_t = 1$ for $t \in [0,\frac{1}{2})$ and $\sigma_t = 2$ for $t \in [\frac{1}{2},1]$. When choosing the optimal bandwidth, \cite{barndorff2008designing}\footnote{see pp. 1494-1495 for more details.} showed that the AVAR is equal to
\begin{eqnarray}
\label{AVAR}
AVAR_{[0,T]}^{(RK)} = a_0 \l( T \int_0^T \sigma_u^4 du \r)^{3/4} g,
\end{eqnarray}
where $g$ is defined as
$$g = \frac{16}{3} \sqrt{\rho k_{\bullet}^{0,0} k_{\bullet}^{1,1}} \bigg( \frac{1}{\sqrt{1 + \sqrt{1 + 3 d / \rho^2}}} + \sqrt{1 + \sqrt{1 + 3 d / \rho^2}} \bigg),$$
with
$$\rho = \frac{\int_0^T \sigma_u^2 du}{\sqrt{T\int_0^T \sigma_u^4 du}} \text{   ,   } d = \frac{k_{\bullet}^{0,0} k_{\bullet}^{2,2}}{ (k_{\bullet}^{1,1})^2},$$
and where $k_{\bullet}^{i,i}$ are constant functions of the kernel. We fix $T=1$ and we compute in that case $\int_0^1 \sigma_u^2 du = 5/2$, $\int_0^1 \sigma_u^3 du = 9/2$ and $\int_0^1 \sigma_u^4 du = 17/2$. Thus, the bound of efficiency is equal to $36 a_0$, whereas $AVAR_{[0,1]}^{(RK)} = 37.89 a_0$. This can be expressed as a loss of $\frac{37.89 - 36}{36} \approx 5 \%$, which is to be compared to the loss in the parametric case\footnote{Details can be found on Table II (p. 1495, \cite{barndorff2008designing}).} $\frac{8.02 - 8}{8} \approx .25 \%$. When fixing $B=2$, the volatility on each block is constant and thus yields $AVAR_{[0,1/2]}^{(RK)} = 2^{-3/2} \times 8.02 a_0$ on the first block and $AVAR_{[1/2,1]}^{(RK)} = 2^{3/2} \times 8.02 a_0$ on the second block. As both estimates are uncorrelated\footnote{if we remove end-effects.}, we obtain that the global AVAR is equal to $AVAR_{[0,1]}^{(RK)} = \sqrt{2} (AVAR_{[0,1/2]}^{(RK)} + AVAR_{[1/2,1]}^{(RK)}) = 8.02 a_0 \int_0^1 \sigma_u^3 du$, i.e. .25 \% loss which corresponds exactly to the parametric loss.

\smallskip
From (\ref{AVAR}), we can see that the theoretical loss can be expressed as a deterministic function of the already well-known measure of volatility constancy $\rho$ and another connected quantity which we denote $$\kappa = \frac{\int_0^T \sigma_u^3 du}{
T^{1/4}(\int_0^T \sigma_u^4 du)^{3/4}}.$$
Details can be found in Section 2, along with an expression for the QMLE loss as well. In the previous example where the loss was about $5 \%$, the corresponding setting can be computed as $\rho = 5/2 \times \sqrt{2/17} \approx .86$ and $\kappa = 9/2 \times (2/17)^{3/4} \approx .90$. Volatility on real data is moving more than on this toy example, corresponding to lower $\rho$ and $\kappa$. In their empirical study, \cite{andersen2014robust} daily estimate $\rho^{-1}$ and find that the typical value is around 1.3, and about 1.6 when restricting to the top 10\% days in terms of intraday variation of volatility. This corresponds respectively to estimates of $\rho$ as $1/1.3 \approx .77$ and $1/1.6 \approx .62$. When taking respectively those two realistic values, the corresponding RK and QMLE losses are expected to be around 20\% (can go up to 100 \%), depending on the other parameter value $\kappa$. With such highly inefficient estimators, we believe that there is a practical need for variance reduction. This is especially the case on days when the volatility is moving a lot.

\smallskip
Clearly this estimator is related to local parametric methods in high-frequency data, i.e. aggregating local parametric estimates. For example, \cite{mykland2009inference} investigated the ex post adjustment involving asymptotic likelihood
ratios to make when assuming constant local volatility. \cite{reiss2011asymptotic} showed the asymptotic equivalence in Le Cam's sense between the non-parametric and locally constant volatility experiment. To estimate quarticity and other functionals of volatility, \cite{jacod2013quarticity} estimated the volatility locally and plugged the value into the sum. Our work includes \cite{potiron2017estimation}, \cite{potiron2016local}, \cite{clinet2018statistical}.

\smallskip
The remainder of the paper is structured as follows. Section 2 stretches the limitations of the global approach by expressing the loss as a function of $\rho$ and $\kappa$. In Section 3, we provide the model, investigate the RK and the QMLE and their corresponding limit theory. Section 4 investigates what happens to both methods when considering stochastic arrival times and adding jump in the price process. Section 5 performs a Monte Carlo experiment to assess finite sample performance and AVAR reduction. Section 6 provides an empirical illustration where we quantify the expected gain in practice. Theoretical details and proofs can be found in the Appendix.
\section{Limitations of the global approach} \label{limitations}
This section documents the performance of the global RK and QMLE. In particular, we show how it deteriorates as a function of heteroskedasticity. Finally, we diagnose the reasons and provide the solution to this relative failure. 

\smallskip
One crucial feature common to both estimators is that they behave remarkably well when volatility is constant. Indeed, the QMLE is efficient and the RK Tukey-Hanning 16 almost efficient in that case. Even the RK Tukey-Hanning 2, with an AVAR over the bound of efficiency ratio of less than 1.04, can be considered as "practically efficient". To study what happens when volatility is time-varying, it is useful for $0 \leq r < s \leq T$  to define
$$\rho_{r,s} = \frac{\int_r^s \sigma_u^2 du}{\sqrt{(s-r) \int_r^s \sigma_u^4 du}} \text{ and  } \kappa_{r,s} = \frac{\int_r^s \sigma_u^3 du}{
(s-r)^{1/4}(\int_r^s \sigma_u^4 du)^{3/4}}$$
to be measures of heteroskedasticity. In the following, we will be using $\rho$ and $\kappa$ in place of $\rho_{0,T}$ and $\kappa_{0,T}$. The quantity $\rho$ was already introduced in \cite{barndorff2008designing} and plays an important role in the AVAR of both RK and the QMLE. \cite{xiu2010quasi} (Figure 1, p. 241) expresses the quotient of both AVARs as a function of $\rho$, but does not assess their respective performance when compared to the (conjectured) bound of efficiency defined as $$AVAR_{[0,T]}^{(Bound)} = 8 a_0 T^{\frac{1}{2}} \int_0^T \sigma_u^3 du.$$ 
In contrast, the other quantity $\kappa$ is introduced to investigate that relative performance. More precisely, $\kappa$ is needed to express the AVAR over the bound of efficiency ratio for both approaches since the AVAR does not feature the tricity, i.e. the integrated third moment of volatility, which is key in the bound of efficiency. Evidently, both measures $\rho$ and $\kappa$ are very much connected and we can actually show that we have that 
\begin{eqnarray}
\label{RhoKappa}
0 < \rho_{r,s}^{3/2} \leq \kappa_{r,s} \leq \rho_{r,s}^{1/2} \leq 1.
\end{eqnarray}

Note that the equality $\rho_{r,s} = \kappa_{r,s} = 1$ for all $r,s \in [0,T]$ corresponds to the parametric case. In particular, Eq. (\ref{RhoKappa}) implies that for any given $\rho$, the value $\kappa$ is a.s. in a small boundary around $\rho$. This is of particular interest because as far as the authors know under noisy observations the literature on quarticity estimation\footnote{see, e.g., \cite{jacod2009microstructure}, \cite{andersen2014robust}, \cite{mancino2012estimation}, \cite{potiron2016local} and \cite{clinet2017estimation}.}  is far more abundant than the corresponding work on estimating tricity \footnote{see the spectral approach AVAR estimator in \cite{altmeyer2015functional}, \cite{potiron2016local} and \cite{clinet2017estimation}.}, which implies that in practice $\rho$ can be estimated relatively easily, whereas $\kappa$ would require more effort. From \cite{andersen2014robust} (Figure 7, p. 41), when taking a pre-averaging window equal to one minute (chosen consistently with their recommendation in Section 5.2.4 on p. 34 where the authors argue that a reasonable choice of window should lie between 30 seconds and 2 minutes) we infer that the estimates of $\rho$ are about $1/1.2 \approx .83$, $1/1.3 \approx .77$ and $1/1.6 \approx .62$ when considering respectively the bottom 10\% days in terms of intraday variation of volatility, all days and the top 10\% days in terms of intraday variation of volatility. Correspondingly, we will be using $\rho_{high} = .83$, $\rho_{regular} = .77$, $\rho_{low} = .62$ to refer respectively to high, regular and low values of $\rho$ throughout the rest of the paper. It is not surprising to find such low values on stocks data as it has been understood for several decades now that many stylized facts describe volatility as time-varying (see, e.g., \cite{ghysels19965}, \cite{engle2001good}).

\smallskip
When using the optimal bandwidth, $AVAR_{[0,T]}^{(RK)}$ is defined as
\begin{eqnarray*}
AVAR_{[0,T]}^{(RK)} = a_0 \l( T \int_0^T \sigma_u^4 du \r)^{3/4} g,
\end{eqnarray*}
where we have 
$$g = \frac{16}{3} \sqrt{\rho k_{\bullet}^{0,0} k_{\bullet}^{1,1}} \bigg( \frac{1}{\sqrt{1 + \sqrt{1 + 3 d / \rho^2}}} + \sqrt{1 + \sqrt{1 + 3 d / \rho^2}} \bigg) \text{ and } d = \frac{k_{\bullet}^{0,0} k_{\bullet}^{2,2}}{ (k_{\bullet}^{1,1})^2},$$
with $k_{\bullet}^{i,i}$ constant functions of the kernel. Correspondingly, we give the formal definition of the RK loss as
\begin{eqnarray}
\label{loss}
L^{(RK)} = \frac{AVAR_{[0,T]}^{(RK)}}{AVAR_{[0,T]}^{(Bound)}} - 1.
\end{eqnarray}
Obvious computations lead to $L^{(RK)} = g \kappa^{-1} /8 -1$. If we see $g$ as a function of $\rho$, $L^{(RK)}$ is equal to $g(1)/8 - 1$  in the parametric case. The parametric values for several kernels can be directly inferred from \cite{barndorff2008designing} (Table II, p. 1495) and the loss is equal to .25 \% when considering the Tukey-Hanning 16, 3.625 \% for the Tukey-Hanning 2, 6.75 \% for the Parzen and 13 \% for the Cubic kernel. We have that $g$ is an increasing function of $\rho$, and thus the effect of $\rho$ and $\kappa$ are reverse. Next we consider the AVAR of the QMLE expressed via 
$$AVAR_{[0,T]}^{(QMLE)} = \frac{5 T a_0 \int_0^T \sigma_u^4 du}{( \int_0^T \sigma_u^2 du)^{1/2}} + 3 a_0 \l(\int_0^T \sigma_u^2 du \r)^{3/2}.$$
The formula can actually be found in Box V (p. 240, \cite{xiu2010quasi}). The corresponding QMLE loss is defined in analogy with (\ref{loss}) and can be expressed as
\begin{eqnarray*}
L^{(QMLE)} & = & \frac{AVAR_{[0,T]}^{(QMLE)}}{AVAR_{[0,T]}^{(Bound)}} - 1\\
& = & \frac{5+3 \rho^2}{8 \kappa \rho^{1/2}} - 1.
\end{eqnarray*}
Figure \ref{efficiencyGlob} plots the feasible loss region for three typical RK, the QMLE and the PAE with triangle kernel. It is clear that they highly lose efficiency when $\rho$ is decreasing. The QMLE is dominated by the RK approach when $\rho$ is low, which was observed on Figure 1 (p. 241, \cite{xiu2010quasi}).

\smallskip
The problem behind this potentially high loss can be intuitively explained as follows. For the RK, although the optimal tuning parameter is robust to time-varying volatility, it suffers from the fact that one day\footnote{or one week, one month, etc.} is too long to "stay optimal". This is a very similar situation to the PAE, which also features a tuning parameter. Subsequently, \cite{jacod2015microstructure} used block estimations to heavily reduce variance. As for the QMLE, which in contrast is designed in a parametric way yielding no choice of tuning parameter, the smaller $\rho$ and $\kappa$ are, the further the misspecified model deviates from the truth. It is by nature a different estimator, but local methods are expected to reduce the misspecification as in \cite{reiss2011asymptotic}. Thus, we aim to reduce the non-parametric loss into the parametric loss using adapted local methods. As we can see on Figure \ref{efficiencyGlob}, the QMLE will benefit the most as it is efficient in the parametric case and deteriorates more than the RK in the non-parametric case.

\section{Local estimation}
\subsection{Model for the observations}
\label{model}
We assume that the latent log-price process and the volatility follow
\begin{eqnarray}
dX_t & = & b_t dt + \sigma_t dW_t,\\
d \sigma_t & = & \tilde{b}_t dt + \tilde{\sigma}_t^{(1)} dW_t + \tilde{\sigma}_t^{(2)} d\tilde{W}_t + d\tilde{J}_t,
\end{eqnarray}
where $(W_t, \tilde{W}_t)$ is a 2 dimensional standard Brownian motion, the drift $(b_t,\tilde{b}_t)$ is componentwise locally bounded, the volatility matrix 
$$\left(\begin{matrix}  \sigma_t & 0 \\ 
\tilde{\sigma}_t^{(1)} & \tilde{\sigma}_t^{(2)} 
\end{matrix}\right)$$ 
is componentwise locally bounded, itself an It\^{o} process and $\inf_t (\min(\sigma_t, \tilde{\sigma}_t^{(2)})) > 0$ a.s. We also assume that $\tilde{J}_t$ is a pure jump process of finite activity. This rules out jumps in $X_t$, an issue addressed in Section \ref{robustness}. In contrast the volatility process can include jumps (see, e.g., \cite{todorov2011volatility} for empirical evidence). The observations are contaminated by the microstructure noise so that we observe 
$$Z_{t_i} = X_{t_i} + \epsilon_{t_i},$$ 
where $t_i$ correspond to the observation times\footnote{Note that $t_i$, $\Delta$, etc. are implicitly assumed to depend on the index $n$. We sometimes refer to $t_i^n$, $\Delta_n$ when necessary.} which are assumed to be regularly spaced, i.e. satisfying $t_{i} - t_{i-1} = \Delta$. Stochastic arrival times are also considered in Section \ref{robustness}. Furthermore, we assume that the noise is independent and identically distributed (i.i.d), and independent of the other quantities, with null-mean, variance $a_0^2$ and finite fourth moment. Next the horizon time is defined as $T >0$. Finally, we consider the high frequency asymptotics and assume that $n$ goes to infinity, where $T = n \Delta$. In particular, the time gap $\Delta$ goes to $0$.  

\subsection{Realized Kernels}
\label{RK}

\subsubsection{Local RK definition}
We consider first the framework $B=1$ where the local RK coincides with the RK. The flat-top RK takes on the form
$$K = \gamma_0 + \sum_{h=1}^{H} k \l( \frac{h-1}{H} \r) \l( \gamma_h + \gamma_{-h} \r),$$ 
where $H > 0$ and the deterministic kernel $k(x)$ is defined for $x \in [0,1]$. The realized autocovariance is defined as
$$\gamma_h = \sum_{j=1}^{n} (Z_{\Delta j} - Z_{\Delta (j-1)}) (Z_{\Delta (j - h)} - Z_{\Delta (j-h-1)}),$$
where $h = -H, \cdots, -1, 0, 1 , \cdots, H$. 

\smallskip
In the general case $B >1$, for each $i =1,\cdots,B$ we choose a bandwidth $H_i >0$ and define $K_i$ the estimate on the $i$th block $\big[ \Tau_{i-1}, \Tau_i \big]$, where $\Tau_{i} = iT/B$. On each block, we also assume that the number of observations $n/B$ is an integer for simplicity of exposition. Formally, all the considered quantities could be written with floor brackets, and all the results would still hold. We aggregate the local estimates to obtain the adapted version of the RK defined as 
$$\tilde{K} = \sum_{i=1}^{B} K_i.$$ 
The corresponding $H = (H_1,\cdots,H_B)$ is now $B$-dimensional in this case. We also adapt the jittering introduced in Section 2.6 (\cite{barndorff2008designing}, p. 1487), i.e. for $i=0, \cdots, B$ we assume that $X_{\Tau_{i}}$ is an average of $m$ distinct observations on the interval $(\Tau_i - \Delta, \Tau_i + \Delta)$. 
\subsubsection{Asymptotic theory}
We define $L_X$ for $\sigma(X)$-stable convergence. We further define 
$$\xi_{r,s}^2 = \frac{a_0^2}{\sqrt{ (s-r) \int_{r}^{s} \sigma_u^4 du}}$$
as the noise-to-signal ratio, and refer to $\xi^2 = \xi_{0,T}^2$ in the following. Finally, we define kernel weight functions $k(x)$ that are two times continuously differentiable on $[0,1]$ and 
$$k_{\bullet}^{0,0} = \int_0^1 k(x)^2 dx, \text{ } k_{\bullet}^{1,1} = \int_0^1 k'(x)^2 dx, \text{ } k_{\bullet}^{2,2} = \int_0^1 k''(x)^2 dx.$$
We recall the main asymptotic result with fastest rate of convergence about the RK which can be found in Theorem 4 (p. 1493) in \cite{barndorff2008designing}. When $k'(0)^2 + k'(1)^2 =0$, $m \rightarrow \infty$, and $H = c n^{1/2}$
, we have
\begin{eqnarray}
\label{RK00}
n^{1/4} \Big( K - \int_0^T \sigma_u^2 du \Big) \overset{L_X}{\rightarrow} \calm\caln \Big(0, \underbrace{4T \int_0^T \sigma_u^4 du \big\{ c k_{\bullet}^{0,0} + c^{-1} 2 k_{\bullet}^{1,1} \rho \xi^2 + c^{-3} k_{\bullet}^{2,2} \xi^4 \big\}}_{AVAR_{[0,T]}^{(RK,c)}} \Big),
\end{eqnarray}
where $\calm\caln$ denotes a mixed normal distribution. A straightforward application of (\ref{RK00}) on each block $i=1, \cdots, B$ yields
\begin{eqnarray}
\label{fastlocalAVAR}
n^{1/4} \l(K_i - \int_{\Tau_{i-1}}^{\Tau_i} \sigma_u^2 du \r) \overset{L_X}{\rightarrow} \calm\caln \l(0, B^{1/2} AVAR_{[\Tau_{i-1}, \Tau_i]}^{(RK,c_i)} \r),
\end{eqnarray}
where $c_i$ is the tuning parameter used on the $i$th block. Next we show that the AVAR associated to $\tilde{K}$ is equal to the sum of variance terms in (\ref{fastlocalAVAR}).
\begin{RK}
\label{RKth} (CLT for local RK) When $k'(0)^2 + k'(1)^2 =0$, $m \rightarrow \infty$, and $H = c n^{1/2}$
, we have
\begin{eqnarray}
\label{RKth2}
n^{1/4} \l( \tilde{K} - \int_0^T \sigma_u^2 du \r) \overset{L_X}{\rightarrow}  \calm\caln \l(0, B^{1/2} \sum_{i=1}^B AVAR_{[\Tau_{i-1}, \Tau_i]}^{(RK,c_i)}\r).
\end{eqnarray}
\end{RK}
\begin{remark}
The requirement that $m \rightarrow \infty$ in (\ref{RKth2}) is due to end-effects. The reader should refer to the discussion in \cite{barndorff2008designing} (p. 1493) in the case $B=1$. When $m$ is fixed, the relative contribution\footnote{The corresponding expression can be found in the second term in (\ref{th34}).} to the AVAR is proportional to $\xi^2/m$, as it was already the case for the RK. \cite{barndorff2009realized} documented that this magnitude can reasonably be ignored in practice.
\end{remark}
To determine the $B$ tuning parameters that minimize the AVAR in (\ref{RKth2}), we can consider each local AVAR independently as they depend on one distinct tuning parameter. For that purpose, we follow Section 4.3 in \cite{barndorff2008designing} (p. 1494-1496) and consider that 
$$(H^{(1)}, \cdots, H^{(B)}) = (c_1 \xi_{0, \Tau_{1}}, \cdots, c_B \xi_{\Tau_{B-1}, T}) \sqrt{n/B}.$$ 
The optimal values are then shown to be equal to 
$$c_i^* = \sqrt{\rho_{\Tau_{i-1}, \Tau_i}  \frac{k_{\bullet}^{1,1}}{k_{\bullet}^{0,0}} \bigg( 1 + \sqrt{1 + 3 d / \rho_{\Tau_{i-1}, \Tau_i}^2} \bigg)}.$$
The corresponding AVAR is equal to 
$$AVAR_{[\Tau_{i-1}, \Tau_i]}^{(RK,c_i^*)} =  a_0 \l( \Delta_B \int_{T_{i-1}}^{T_i} \sigma_u^4 du \r)^{3/4} g (\rho_{T_{i-1},T_i}),$$
where $g$ is considered here as a function of $\rho$ and $\Delta_B = T/B$ corresponds to the block length. 

\smallskip
We provide in what follows a consistent estimator for each tuning parameter. To pre-estimate on each block the integrated volatility and quarticity, we consider the pre-averaging estimators from \cite{jacod2009microstructure}. For each block $i=1, \cdots, B$ we choose an integer $k_i$ and a real parameter $\theta_i > 0$ which satisfy $k_i \sqrt{\Delta} = \theta_i + o (\Delta^{1/4})$. We also consider a continuous function $f$ on $[0,1]$, piecewise $C^1$ with a piecewise Lipschitz derivative $f'$ such that $f(0) = f(1) = 0$, $\int_0^1 f(s)^2 ds > 0$. We define
\begin{eqnarray}
\phi_1 (s) & = & \int_s^1 f'(u) f'(u-s)du,\\
\phi_2 (s) & = & \int_s^1 f(u) f(u-s)du,\\
\Phi_{jl} & = & \int_0^1 \phi_j (s) \phi_l (s) ds \text{ for } j,l =1,2,
\end{eqnarray}
$\psi_1 = \phi_1 (0)$ and $\psi_2 = \phi_2 (0)$. We further define 
$$\bar{Z}_{l,i} = \sum_{j=1}^{k_i-1} f(j/k_i) (Z_{(l+j)\Delta} - Z_{(l+j-1)\Delta}).$$
The pre-averaging estimators of integrated volatility and quarticity on each block take on the form
\begin{eqnarray}
\widehat{\int_{\Tau_{i-1}}^{\Tau_i} \sigma_u^2 du} & = &\frac{\sqrt{\Delta}}{\theta_i \psi_2} \sum_{j=0}^{n/B - k_i +1} \bar{Z}_{j,i}^2 - \frac{\psi_1 \Delta}{2 \theta_i^2 \psi_2} \sum_{j=1}^{n/B} (Z_{(n(i-1)/B + j) \Delta} - Z_{(n(i-1)/B + j-1) \Delta})^2,\\
\widehat{\int_{\Tau_{i-1}}^{\Tau_i} \sigma_u^4 du} & = & \frac{1}{3 \theta_i^2 \psi_2^2} \sum_{j=0}^{n/B - k_i +1} \bar{Z}_{j+(i-1)n/B,i}^4\\ \nonumber
& &- \frac{\Delta \psi_1}{\theta_i^4 \psi_2^2} \sum_{j=0}^{n/B - 2 k_i +1} \bar{Z}_{j+(i-1)n/B,i}^2 \sum_{l=j+k_i}^{j + 2k_i-1} (Z_{l \Delta} - Z_{(l-1+(i-1)n/B) \Delta})^2\\ \nonumber
& & + \frac{\Delta \psi_1^2}{4 \theta_i^4 \psi_2^2} \sum_{j=1}^{n/B - 2} (Z_{(j+(i-1)n/B) \Delta} - Z_{(j-1+(i-1)n/B) \Delta})^2 (Z_{(j+2+(i-1)n/B) \Delta} - Z_{(j+1+(i-1)n/B) \Delta})^2.
\end{eqnarray}
We then estimate
\begin{eqnarray}
\widehat{\rho}_{\Tau_{i-1}, \Tau_i} & = & \frac{\widehat{\int_{\Tau_{i-1}}^{\Tau_i} \sigma_u^2 du} }{\sqrt{\Delta_B \widehat{\int_{\Tau_{i-1}}^{\Tau_i} \sigma_u^4 du}}},\\
\widehat{c}_i^* & = & \sqrt{\widehat{\rho}_{\Tau_{i-1}, \Tau_i}  \frac{k_{\bullet}^{1,1}}{k_{\bullet}^{0,0}} \bigg( 1 + \sqrt{1 + 3 d / \widehat{\rho}_{\Tau_{i-1}, \Tau_i}^2} \bigg)}.
\end{eqnarray}
We provide now a consistent estimator of $AVAR_{B}^{(RK)} = B^{1/2} \sum_{i=1}^B  AVAR_{[\Tau_{i-1}, \Tau_i]}^{(RK,c_i^*)}$. We estimate the noise as $\widehat{a}^2 = (2n)^{-1} \sum_{j=1}^{n}{\l(Z_{\Delta(j+1)} - Z_{\Delta j} \r)^2}$ and the asymptotic variance as
$$\widehat{AVAR}_B^{(RK)} =  \widehat{a}  B^{1/2} \sum_{i=1}^B \l( \Delta_B \widehat{\int_{T_{i-1}}^{T_i} \sigma_u^4 du} \r)^{3/4} g (\widehat{\rho}_{T_{i-1},T_i}).$$
The feasible CLT is given in the following theorem.
\begin{RKcor}
\label{RKcor} (feasible CLT for local RK) When $k'(0)^2 + k'(1)^2 =0$, $m \rightarrow \infty$, and $H = \widehat{c} n^{1/2}$ with $\widehat{c}=(\widehat{c}_1^*, \cdots, \widehat{c}_B^*)$, we have $\widehat{AVAR}_B^{(RK)} \overset{\proba}{\rightarrow} AVAR_{B}^{(RK)}$ and
\begin{eqnarray}
\label{RKthcor2}
n^{1/4} \frac{ \tilde{K} - \int_0^T \sigma_u^2 du}{\sqrt{\widehat{AVAR}_B^{(RK)}}}\overset{\call}{\rightarrow}  \caln \l(0, 1\r).
\end{eqnarray}
\end{RKcor}

\smallskip
Finally, we show that when choosing the optimal values, the AVAR associated to $\tilde{K}$ goes to $\frac{g(1)}{8} AVAR_{[0,T]}^{(Bound)}$ when $B \rightarrow \infty$. The constant $g(1)/8$, when normalized to $g(1)/8 - 1$, corresponds to the parametric loss and depends solely on the shape of the kernel. The rationale of such result is that when $B$ increases we have the volatility roughly constant on each block and thus
\begin{eqnarray*}
\sum_{i=1}^B B^{1/2} AVAR_{[\Tau_{i-1}, \Tau_i]}^{(RK,c_i^*)} & = & \sum_{i=1}^B a_0 B^{1/2} \l( \Delta_B \int_{T_{i-1}}^{T_i} \sigma_u^4 du \r)^{3/4} g (\rho_{T_{i-1},T_i}),\\
& \approx & \sum_{i=1}^B a_0 B^{1/2} \Delta_B^{3/2} \sigma_{\Tau_{i-1}}^3 g (1).
\end{eqnarray*}
Next we obtain by a Riemann sum argument that
\begin{eqnarray*}
\sum_{i=1}^B a_0 B^{1/2} \Delta_B^{3/2} \sigma_{\Tau_{i-1}}^3 g (1) & = & a_0 T^{1/2} \sum_{i=1}^B \Delta_B \sigma_{\Tau_{i-1}}^3 g (1),\\
& \approx & a_0 T^{1/2} \int_0^T \sigma_u^3 du g (1),
\end{eqnarray*}
which can be expressed as $\frac{g(1)}{8} AVAR_{[0,T]}^{(Bound)}$. The formal result is given in the following proposition.
\begin{propAVARRK} \label{propAVARRK}
(Convergence of local RK AVAR)  When $B \to +\infty$, we have
\begin{eqnarray}
\label{eqAVARRK}
AVAR_{B}^{(RK)} \overset{a.s.}{\rightarrow} \frac{g(1)}{8} AVAR_{[0,T]}^{(Bound)}.
\end{eqnarray}
\end{propAVARRK}

\begin{remark} \label{rklossprop}
In particular, the asymptotic loss for $B \to +\infty$ is $g(1)/8-1$, which is always smaller than $L^{(RK)} = g \kappa^{-1}/8-1$ when using the RK with $B=1$. The proof of this statement can be found in Appendix (Section \ref{proofrk}).
\end{remark}

\subsection{QMLE}
\label{QMLE}
In analogy with Section \ref{RK}, we provide in this section a definition of the local estimator and equivalent asymptotic results in the case of the QMLE.

\subsubsection{Local QMLE definition}
\label{localQMLEsub}
We consider first the setting $B=1$ where the local QMLE is equal to the global QMLE. We recapitulate the parametric approach, and introduce the quasi-estimator. \cite{ait2005often} studied the parametric case assuming that the latent efficient log price process satisfies 
\begin{eqnarray}
\label{param}
dX_t = \sigma dW_t.
\end{eqnarray}
The observed log returns $Y_i = Z_{t_i} - Z_{t_{i-1}}$ are following a MA(1) process in that situation. If we postulate that the noise distribution is Gaussian, then the log likelihood function for $Y = (Y_1, \cdots, Y_n)^T$ can be expressed as
\begin{eqnarray}
\label{loglik}
l(\sigma^2, a^2) = - \frac{1}{2} \log \det(\Omega) - \frac{n}{2} \log (2\pi) -\frac{1}{2} Y^T \Omega^{-1} Y,
\end{eqnarray}
where 
$$\Omega = \left(\begin{matrix}
                    \sigma^2 \Delta + 2 a^2 & - a^2 & 0 & \cdots & 0 \\
                    - a^2 & \sigma^2 \Delta + 2 a^2 & - a^2 & \ddots & \vdots \\
                    0 & - a^2 & \sigma^2 \Delta + 2 a^2 & \ddots & 0\\
                    \vdots & \ddots & \ddots & \ddots & - a^2\\
                    0 & \cdots & 0 & - a^2 & \sigma^2 \Delta + 2 a^2
                  \end{matrix}\right) \in \reels^{n\times n}.\\[12pt] $$
We define the corresponding MLE which maximizes (\ref{loglik}) as $(\widehat{\sigma}^2, \widehat{a}^2)$ and the estimator of integrated volatility as $Q = T \widehat{\sigma}^2$. When the log price $X_t$ features stochastic volatility and drift as in Section \ref{model} and/or when the noise is not normally distributed, $(\widehat{\sigma}^2, \widehat{a}^2)$ is seen as the QMLE. 

\smallskip
When $B > 1$, we define for each block $i \in \{1,\cdots,B\}$ a local QMLE estimator $(\widehat{\sigma}_i^2, \widehat{a}_i^2)$ which maximizes the expression $l(\sigma^2,a^2)$ applied to the observations on $(\Tau_{i-1},\Tau_i]$ only, along with the local integrated volatility estimator $Q_i = \Delta_B \widehat{\sigma}_i^2$. We then construct the aggregate version of the QMLE as 
$$\tilde{Q} = \sum_{i=1}^B Q_i.$$
 
\subsubsection{Asymptotic theory}
\label{QMLEasymptotictheory}
We state the main result in \cite{xiu2010quasi} (Box V, p. 240). If we assume that $\int_0^T \sigma_u^2 du \in [\underline{\Sigma}, \overline{\Sigma}]$ with $0 < \underline{\Sigma} < \overline{\Sigma}$, we have
\begin{eqnarray}
\label{QMLE00}
\left(\begin{matrix} n^{1/4} \l(Q - \int_0^T \sigma_u^2 du \r) \\ 
n^{1/2} \l( \widehat{a}^2 - a_0^2 \r) 
\end{matrix}\right) \overset{L_X}{\rightarrow} \calm\caln \left( \left(\begin{matrix} 0 \\ 
0 
\end{matrix}\right), \left(\begin{matrix} AVAR_{[0,T]}^{(QMLE)} & 0 \\ 
0 & 2 a_0^4 + \textnormal{cum}_4[\epsilon]
\end{matrix}\right)  \right),
\end{eqnarray}
where we recall that 
$$ AVAR_{[0,T]}^{(QMLE)} =  \frac{ 5 T  a_0 \int_{0}^{T} \sigma_u^4 du}{ \l( \int_{0}^{T} \sigma_u^2 du\r)^{1/2} } + 3 a_0 \l(\int_{0}^{T} \sigma_u^2 du\r)^{3/2}.$$
Also, $\textnormal{cum}_4[\epsilon]$ refers to the fourth cumulant of $\epsilon_0$. An obvious application of (\ref{QMLE00}) for each block $i =1, \cdots, B$ gives us that
$$n^{1/4} \l(Q_i - \int_{\Tau_{i-1}}^{\Tau_i} \sigma_u^2 du \r) \overset{L_X}{\rightarrow} \calm\caln \l(0, B^{1/2} AVAR_{[\Tau_{i-1},\Tau_i]}^{(QMLE)} \r).$$
We show in the following theorem that the AVAR associated to $\tilde{Q}$ can be decomposed as a sum of local AVARs scaled by $B^{1/2}$.
\begin{QMLE} \label{CLTQMLE} (CLT for local QMLE) We have 
$$\left(\begin{matrix} n^{1/4} \l(\tilde{Q} - \int_0^T \sigma_u^2 du \r) \\ 
n^{1/2} \l( B^{-1} \sum_{i=1}^B \widehat{a}_i^2 - a_0^2 \r) 
\end{matrix}\right) \overset{L_X}{\rightarrow} \calm\caln \left( \left(\begin{matrix} 0 \\ 
0 
\end{matrix}\right), \left(\begin{matrix} B^{1/2} \sum_{i=1}^B  AVAR_{[\Tau_{i-1},\Tau_i]}^{(QMLE)} & 0\\ 
0 & 2 a_0^4 + \textnormal{cum}_4[\epsilon]
\end{matrix}\right)  \right).$$
\end{QMLE}
We define $AVAR_{B}^{(QMLE)} = B^{1/2} \sum_{i=1}^B AVAR_{[\Tau_{i-1}, \Tau_i]}^{(QMLE)}$ which is estimated via 
$$ \widehat{AVAR}_B^{(QMLE)} = B^{1/2} \sum_{i=1}^B \l\{\frac{ 5 \Delta_B \widehat{a} \widehat{\int_{\Tau_{i-1}}^{\Tau_i} \sigma_u^4 du}}{ \l( \widehat{\int_{\Tau_{i-1}}^{\Tau_i} \sigma_u^2 du}\r)^{1/2} } + 3 \widehat{a} \l(\widehat{\int_{\Tau_{i-1}}^{\Tau_i} \sigma_u^2 du}\r)^{3/2} \r\}.$$
The feasible theorem follows.
\begin{QMLEcor}
\label{QMLEcor} (feasible CLT for local QMLE) We have $\widehat{AVAR}_B^{(QMLE)} \overset{\proba}{\rightarrow} AVAR_{B}^{(QMLE)}$ and
\begin{eqnarray}
\label{QMLEthcor2}
n^{1/4} \frac{ \tilde{Q} - \int_0^T \sigma_u^2 du}{\sqrt{\widehat{AVAR}_B^{(QMLE)}}}\overset{\call}{\rightarrow}  \caln \l(0, 1\r).
\end{eqnarray}
\end{QMLEcor}
We show now that the AVAR associated to $\tilde{Q}$ goes to $AVAR_{[0,T]}^{(Bound)}$ when $B$ increases.
\begin{propAVARQMLE} \label{propAVARQMLE}
(Convergence of local QMLE AVAR) When $B \to +\infty$, we have
\begin{eqnarray}
\label{eqAVARQMLE}
AVAR_{B}^{(QMLE)} \overset{a.s.}{\rightarrow} AVAR_{[0,T]}^{(Bound)}.
\end{eqnarray}
\end{propAVARQMLE}

\section{Is the local method robust to stochastic sampling times and jump in the price process?}
\label{robustness}
We discuss in this section what happens to both approaches when considering stochastic arrival times and adding jump to the price process. Related work in the global case include \cite{da2017moving} for the QMLE and \cite{varneskov2016flat} for the RK. We further inspect the AVAR behavior when $B \rightarrow +\infty$ in this situation. The results are mitigated. Reduction (conjectured to be efficient) is obtained in the case of stochastic arrival times on the one hand, but there are additional terms in the AVAR as $B \rightarrow +\infty$ when adding jumps on the other hand. 

\subsection{Central limit theory when $B$ is fixed}
We assume that the latent log-price process is now an It\^{o} semimartingale defined by $$dX_t = b_t dt + \sigma_t dW_t + dJ_t,$$
where $b_t$ and $\sigma_t$ satisfy the same conditions as in Section \ref{model}, and $J_t$ is a pure jump process of finite activity.

\smallskip
For the observation times, we adopt the random discretization scheme of \cite{jacod2011discretization} (see Section 14.1) and we assume that there exists an It\^{o} semimartingale $\alpha_t > 0$ which satisfies Assumption 4.4.2 (p. 115) in \cite{jacod2011discretization} and is locally bounded away from $0$, and i.i.d $U_i > 0$ that are both independent of the other quantities, $\alpha_t \ind U_i$, such that 
\bea 
t_0 & = & 0,\\
t_i & = & t_{i-1} + \Delta \alpha_{t_{i-1}} U_i,
\eea
where we recall that $\Delta = T/n$. Finally, we assume that $\esp U_i = 1$, and that for any $q >0$, $m_{q} := \esp U_i^q \to m_{q,\infty} < \infty $ as $n \to +\infty$. Note that, defining $\pi_t := \sup_{i \geq 1} t_{i}-t_{i-1}$, the number of observations before $t$ as $N_n(t) = \sup \{ i\in \naturels - \{0\} | t_{i}  \leq t \}$, we have $\pi_t \overset{\proba}{\rightarrow} 0$ as $n \to +\infty$ and 
\bea 
 \frac{N_n(t)}{n} \to^{u.c.p} \inv{T}\int_0^t{\alpha_s^{-1}ds},
 \label{eqNumberJumps}
 \eea 
 where the convergence $u.c.p$ means uniformly in probability on $[0,t]$ for any $t \in [0,T]$.\footnote{We can prove (\ref{eqNumberJumps}) using Lemma 14.1.5 in \cite{jacod2011discretization}. The uniformity is a consequence of the fact that $N_n$ and $\int_0^.{\inv{\alpha_s}ds}$ are increasing processes and Property (2.2.16) in \cite{jacod2011discretization}.} We further define $N_n = N_n(T)$.

\smallskip
As pointed out in \cite{jacod2011discretization} (p. 431), any deterministic grid satisfies the above conditions. Actually, this model can be considered as more general than the time deformation proposed by \cite{barndorff2008designing} (Section 5.3, pp. 1505-1507)
in the sense that more complex arrival times, such as a Poisson process independent of the other quantities fall under the model. On the contrary, assuming the existence of the quadratic variation of time (see, e.g., Assumption A on p. 1939 in \cite{mykland2006anova}) is too general as our proofs require the existence of the quadratic covariation of time lags for all lags\footnote{To see a condition on the first lag, one can look at Assumption B.vii on p. 37 in \cite{li2016efficient}. This does not include other lags.}. 

\smallskip
Since the price process features possible jumps, the two estimators are no longer consistent to the integrated volatility, but they converge to the quadratic variation 
$$T\bar{\sigma}_0^2 :=\int_0^T{\sigma_s^2ds} + \sum_{0 < s \leq T}\Delta J_s^2,$$
where $\Delta J_s = J_s - J_{s-}$ corresponds to the size of the jump if there is a jump at time $s$ and 0 otherwise. Correspondingly we define on each block $i=1, \cdots, B$ the new local target as 
$$\Delta_B \bar{\sigma}_i^2 :=\int_{\Tau_{i-1}}^{\Tau_i} {\sigma_s^2ds} + \sum_{\Tau_{i-1} < s \leq \Tau_i}\Delta J_s^2.$$
The AVARs obtained in the robust theorems feature $\Delta_B \bar{\sigma}_i^2$ in place of integrated volatility, and the following quantity as an alternative for quarticity:
$$\calq_{(i)}  =  \qtermlocal{\Tau_{i-1}}{\Tau_i}.$$
Correspondingly, we define substitutes for the measure of heteroskedasticity and the noise-to-ratio measure as
\begin{eqnarray*}
\widetilde{\rho}_{\Tau_{i-1},\Tau_i}  =  \frac{\Delta_B \bar{\sigma}_i^2}{\sqrt{\Delta_B \calq_{(i)}}} \text{ and } \widetilde{\xi}_{\Tau_{i-1},\Tau_i}^2 = \frac{a_0^2}{\sqrt{ \Delta_B \calq_{(i)}}}.
\end{eqnarray*}
Moreover, we also introduce 
\bea 
R_{(i)} := \frac{\int_0^T{\alpha_s^{-1}ds}}{\int_{\Tau_{i-1}}^{\Tau_i} \alpha_s^{-1}ds},
\eea 
which corresponds to the asymptotic ratio of the total number of observations over the number of observations on the block $i$ as we have $N_n(T)/(N_n(\Tau_{i})-N_n(\Tau_{i-1})) \overset{\proba}{\rightarrow} R_{(i)}$. Finally, we define $\calg_T := \sigma \l\{U_i^n, \alpha_s, X_s | (i,n) \in \naturels^2 ,  0 \leq s \leq T \r\}$  and refer to $L_{\calg}$ for stable convergence with respect to $\calg_T$. We provide the CLT for the two approaches in what follows.
\begin{RKjumps}
\label{RKthjumps} (robust CLT for local RK) When $k'(0)^2 + k'(1)^2 =0$, $m \rightarrow \infty$, $c = (c_1,\cdots,c_B)$, and $H = c N_n^{1/2}$, we have
\begin{eqnarray}
\label{RKth2Robust}
N_n^{1/4} \l( \tilde{K} -T\bar{\sigma}_0^2 \r) \overset{L_\mathcal{G}}{\rightarrow}  \calm\caln \Big(0, AVAR^{(RK,rob,c)}\Big),
\end{eqnarray}
where
\begin{eqnarray*}
AVAR^{(RK,rob,c)} &=& \sum_{i=1}^B R_{(i)}^{1/2}AVAR_{[\Tau_{i-1}, \Tau_i]}^{(RK,rob,c_i)}, \\
AVAR_{[\Tau_{i-1}, \Tau_i]}^{(RK,rob,c_i)} &=& 4 \Delta_B \calq_{(i)} \big\{ c_i k_{\bullet}^{0,0} + c_i^{-1} 2 k_{\bullet}^{1,1} \widetilde{\rho}_{\Tau_{i-1},\Tau_i} \widetilde{\xi}_{\Tau_{i-1},\Tau_i}^2 + c_i^{-3} k_{\bullet}^{2,2} \widetilde{\xi}_{\Tau_{i-1},\Tau_i}^4 \big\}.
\end{eqnarray*} 
\end{RKjumps}
The new optimal bandwidth is given by 
$$\widetilde{c}_i^* = \sqrt{\widetilde{\rho}_{\Tau_{i-1}, \Tau_i}  \frac{k_{\bullet}^{1,1}}{k_{\bullet}^{0,0}} \bigg( 1 + \sqrt{1 + 3 d / \widetilde{\rho}_{\Tau_{i-1}, \Tau_i}^2} \bigg)},$$
with local and global optimal variances respectively defined as
\begin{eqnarray*}
AVAR_{[\Tau_{i-1}, \Tau_i]}^{(RK,rob,\widetilde{c}_i^*)} & = & a_0 \l( \Delta_B \calq_{(i)} \r)^{3/4} g (\widetilde{\rho}_{T_{i-1},T_i}),\\
AVAR_{B}^{(RK,rob)} & = & \sum_{i=1}^B  R_{(i)}^{1/2} AVAR_{[\Tau_{i-1}, \Tau_i]}^{(RK,rob,\widetilde{c}_i^*)}.
\end{eqnarray*}

\smallskip
As for the QMLE, the log likelihood function when $B=1$ keeps the same form (\ref{loglik}) but we replace $n$ by $N_n$ in the definition of $\Omega$ now defined as 
$$\Omega = \left(\begin{matrix}
                    \sigma^2 \widetilde{\Delta}+ 2 a^2 & - a^2 & 0 & \cdots & 0 \\
                    - a^2 & \sigma^2 \widetilde{\Delta}+ 2 a^2 & - a^2 & \ddots & \vdots \\
                    0 & - a^2 & \sigma^2 \widetilde{\Delta}+ 2 a^2 & \ddots & 0\\
                    \vdots & \ddots & \ddots & \ddots & - a^2\\
                    0 & \cdots & 0 & - a^2 & \sigma^2 \widetilde{\Delta}+ 2 a^2
                  \end{matrix}\right) \in \reels^{N_n\times N_n},\\[12pt] $$
where $\widetilde{\Delta} = T/N_n$. Each local QMLE estimator $(\widehat{\sigma}_i^2, \widehat{a}_i^2)$ is now defined as a maximizer of 
\begin{eqnarray}
\label{loglikRobust}
l_{(i)}(\sigma^2, a^2) = - \frac{1}{2} \log \det(\Omega_{(i)}) - \frac{N_{n,(i)}}{2} \log (2\pi) -\frac{1}{2} Y_{(i)}^T \Omega_{(i)}^{-1} Y_{(i)},
\end{eqnarray}
where $Y_{(i)}$ is the vector of price returns on the $i$th block, $N_{n,(i)} := N_n(\Tau_{i}) - N_n(\Tau_{i-1})$, and 
$$\Omega_{(i)} = \left(\begin{matrix}
                    \sigma^2 \widetilde{\Delta}_{(i)}+ 2 a^2 & - a^2 & 0 & \cdots & 0 \\
                    - a^2 & \sigma^2 \widetilde{\Delta}_{(i)}+ 2 a^2 & - a^2 & \ddots & \vdots \\
                    0 & - a^2 & \sigma^2 \widetilde{\Delta}_{(i)}+ 2 a^2 & \ddots & 0\\
                    \vdots & \ddots & \ddots & \ddots & - a^2\\
                    0 & \cdots & 0 & - a^2 & \sigma^2 \widetilde{\Delta}_{(i)}+ 2 a^2
                  \end{matrix}\right) \in \reels^{N_{n,(i)}\times N_{n,(i)}},\\[12pt] $$
with $\widetilde{\Delta}_{(i)} :=\Delta_B /N_{n,(i)}$. If we assume that $T\bar{\sigma}_0^2 \in [\underline{\Sigma}, \overline{\Sigma}]$ we obtain the following theorem.
\begin{QMLEjumps} \label{CLTQMLEjumps} (robust CLT for local QMLE) We have 
$$\left(\begin{matrix} N_n^{1/4} \l(\tilde{Q} - T\bar{\sigma}_0^2 \r) \\ 
N_n^{1/2} \l( B^{-1} \sum_{i=1}^B \widehat{a}_i^2 - a_0^2 \r) 
\end{matrix}\right) \overset{L_{\mathcal{G}}}{\rightarrow} \calm\caln \left( \left(\begin{matrix} 0 \\ 
0 
\end{matrix}\right), \left(\begin{matrix}  AVAR_B^{(QMLE,rob)} & 0\\ 0 & AVAR_B^{(QMLE,\epsilon)}
\end{matrix}\right)  \right),$$
where
\begin{eqnarray*}
AVAR_B^{(QMLE,rob)} &=& \sum_{i=1}^B R_{(i)}^{1/2} AVAR_{[\Tau_{i-1},\Tau_i]}^{(QMLE,rob)} \\
 AVAR_{[\Tau_{i-1},\Tau_i]}^{(QMLE,rob)} & = & \frac{ 5   a_0\Delta_B^{1/2} \calq_{(i)}}{ \bar{\sigma}_{i} } + 3 a_0 \bar{\sigma}_{i}^3\Delta_B^{3/2},\\ 
 AVAR_B^{(QMLE,\epsilon)} & = & \inv{B^2}\sum_{i=1}^B{R_{(i)}} \l\{2 a_0^4 + \textnormal{cum}_4[\epsilon]\r\}.
\end{eqnarray*}
\end{QMLEjumps}

\subsection{The good case: robustness to stochastic arrival times}
Here we assume a no-jump setting, i.e. $J_t=0$. The following two propositions provide the AVAR asymptotic behavior when $B \rightarrow \infty$ for the two methods. The limit is very similar to that in the regular observation case, and thus the local method is robust to stochastic observation times. Note that the conjectured bound of efficiency is affected by the setting and takes the form 
$$8 a_0 \l(\int_0^T\alpha_s^{-1}ds\r)^{1/2}\int_0^T \alpha_s^{1/2} \sigma_u^3 du.$$
\begin{RKcorRobust} \label{RKcorRobust} (Asymptotic behavior of local RK AVAR when sampling times are stochastic) When $B \to +\infty$, we have 
$$ AVAR_B^{(RK,rob)} \overset{a.s.}{\rightarrow} 8g(1)a_0\l(\int_0^T\alpha_s^{-1}ds\r)^{1/2} \int_0^T{\alpha_s^{1/2}\sigma_s^3}ds.$$
\end{RKcorRobust}

\begin{QMLEcorRobust} \label{QMLEcorRobust} (Asymptotic behavior of local QMLE AVAR when sampling times are stochastic) When $B \to +\infty$, we have 
$$ AVAR_B^{(QMLE,rob)} \overset{a.s.}{\rightarrow} 8a_0\l(\int_0^T\alpha_s^{-1}ds\r)^{1/2} \int_0^T{\alpha_s^{1/2}\sigma_s^3}ds.$$
\end{QMLEcorRobust}

\subsection{The bad case: adding jumps to the price process}
In this section, the price process can feature jumps. Actually in such setting the AVAR of the RK tends to a big value as $B$ increases, and that of QMLE explodes. This sheds light on a weak point of the local method in this case. 
\begin{RKcorRobust2} \label{RKcorRobust2} (Asymptotic behavior of local RK AVAR when $J \neq 0$)
As $B \to +\infty$, 
\beas AVAR^{(RK,rob)} &\overset{a.s.}{\rightarrow}&8g(1)a_0\l(\int_0^T\alpha_s^{-1}ds\r)^{1/2} \int_0^T{\alpha_s^{1/2}\sigma_s^3}ds\\
&+&\frac{16}{3}a_0\l(\inv{\sqrt{2}} + \sqrt{2}\r) \sqrt{k_{\bullet}^{0,0} k_{\bullet}^{1,1}} \l(\int_0^T\alpha_s^{-1}ds\r)^{1/2}\sum_{0 < s \leq T} \Delta J_s^2  \l(\sigma_s^2\alpha_s + \sigma_{s-}^2\alpha_{s-}\r)^{1/2} .
\eeas 
\end{RKcorRobust2}

\begin{QMLEcorRobust2} \label{QMLEcorRobust2}(Asymptotic behavior of local QMLE AVAR when $J \neq 0$)
As $B \to +\infty$, 
$$ AVAR_B^{(QMLE,rob)} \overset{a.s.}{\sim} 3a_0B^{1/2}T^{-1/2}\l(\int_0^T\alpha_s^{-1}ds\r)^{1/2}\sum_{0 < s \leq T}\alpha_s^{1/2}|\Delta J_s|^3  \overset{a.s.}{\rightarrow} +\infty.$$
\end{QMLEcorRobust2}

\section{Numerical study}\label{simulations}

\subsection{Goal of the study}
In this section, we discuss theoretical AVAR reduction and we examine the performance of the local RK $\tilde{K}_B$ and the local QMLE $\tilde{Q}_B$ in a finite sample context for several values of $B$. We carry out Monte Carlo simulations for three different volatility models having realistic values of $\rho$. We then check whether asymptotic approximations of several statistics correctly kick in to illustrate to what extent the theory is affected when the sample data is finite of size $n$. First, we assess the central limit theories for the two infeasible statistics
\beas
Z_n^{\tilde{K}_B} = \frac{n^{1/4}\l(\tilde{K}_B - \int_0^T{\sigma_u^2du}\r)}{\sqrt{ AVAR_{B}^{(RK)}}} \text{ , }
Z_n^{\tilde{Q}_B} = \frac{n^{1/4}\l(\tilde{Q}_B - \int_0^T{\sigma_u^2du}\r)}{\sqrt{AVAR_{B}^{(QMLE)}}},
\eeas 
and the two feasible statistics
\beas
\tilde{Z}_n^{\tilde{K}_B} = \frac{n^{1/4}\l(\tilde{K}_B - \int_0^T{\sigma_u^2du}\r)}{\sqrt{ \widehat{AVAR}_{B}^{(RK)}}} \text{ , }
\tilde{Z}_n^{\tilde{Q}_B} = \frac{n^{1/4}\l(\tilde{Q}_B - \int_0^T{\sigma_u^2du}\r)}{\sqrt{\widehat{AVAR}_{B}^{(QMLE)}}},
\eeas 
for $B=1,2,4,6,8$. In particular, we investigate how increasing $B$ affects the standard normal approximation of these two studentizations for several levels of sampling. Second, we compare the relative performance of the local RK and the local QMLE. To do so, we report the empirical loss defined as 
\beas 
\breve{L}_B^{(RK)} = \esp_M \l[\frac{n^{1/2}(\tilde{K}_B - \int_0^T{\sigma_s^2ds} )^2}{AVAR_{[0,T]}^{(Bound)}}\r] -1 \text{ , } \breve{L}_B^{(QMLE)} = \esp_M \l[\frac{n^{1/2}(\tilde{Q}_B - \int_0^T{\sigma_s^2ds} )^2}{AVAR_{[0,T]}^{(Bound)}}\r] -1
\eeas
where $\esp_M [X]$ denotes the sample mean of $X$ based on the $M$ Monte Carlo simulations and we recall that $AVAR_{[0,T]}^{(Bound)} = 8 a_0 T^{\frac{1}{2}} \int_0^T \sigma_u^3 du$ is the bound of efficiency for the asymptotic variance.  We also define the theoretical loss as 
$$L_B^{(\Sigma)} = \frac{AVAR_B^{(\Sigma)}}{AVAR_{[0,T]}^{(Bound)}} - 1.$$ 
and report the sample mean of the theoretical loss $\tilde{L}_B^{(\Sigma)} = \esp_M \big[ L_B^{(\Sigma)} \big]$ for $\Sigma \in \{RK,QMLE\}$.
Note that $\tilde{L}_B^{(\Sigma)}$ is close to the mean loss $\esp \big[L_B^{(\Sigma)} \big]$ if $M$ is large enough. The empirical loss $\breve{L}^{(\Sigma)}$, which gives us a simple criterion to compare the estimators, can be decomposed as 
$$\breve{L}_B^{(\Sigma)} = \underbrace{\tilde{L}_B^{(\Sigma)}}_{\text{theoretical loss due to the finiteness of }B} + \underbrace{(\breve{L}_B^{(\Sigma)} - \tilde{L}_B^{(\Sigma)})}_{\text{loss due to the finite sample } n}.$$
\subsection{Simulation design}\label{simDesign}

We implement the above procedures for $M=10,000$ Monte Carlo simulations of intraday returns on the time interval $[0,T]$, $T=1/252$ year (that is $T=1$ working day). One working day is in turn subdivided in $23,400$ seconds corresponding to $6.5$ hours of trading activity. For each model, the corresponding trajectories are generated from a classical Euler scheme based on $n=46,800$ intervals, that is one observation every $0.5$ seconds. We simulate $1000$ more observations prior and post main trading period in order to compute properly different $\gamma_h$ that are necessary for the RK. Indeed, using their truncated versions $\tilde{\gamma}_h = \sum_{j=H+1}^{n-H} (Z_{\Delta j} - Z_{\Delta (j-1)}) (Z_{\Delta (j - h)} - Z_{\Delta (j-h-1)})$ tend to generate a non-negligible bias as pointed out in \cite{xiu2010quasi} (see Table 2 on p. 243), so that we prefer to overcome this issue with a few minutes of out-of-sample data. Finally, we also use observations based on sparsely sampled versions of the original trajectories, for a number of intervals taking on the values $23,400$, $11,700$, and $5,850$, the latter corresponding to having one observation every $4$ seconds, which still corresponds to a fairly heavily traded stock. We do not report the results for lower frequencies, but the theory still kicks in for sparser samplings too.

\smallskip 
We consider three stochastic volatility models to simulate the intraday returns, along with three levels of mean noise-to-signal ratios $\xi^2 = 0.01$, $\xi^2 = 0.001$ and $\xi^2 = 0.0002$. The three values are empirically corroborated in \cite{hansen2006realized}, where the authors report empirical values of $\xi^2$ for several stocks ranging from $0.00004$ to $0.006$ (see Table 3 on p. 147). We introduce now the volatility models, which have been designed to reflect different average values of $\rho$ ranging from $0.89$ (corresponding to a high value) for Model 1, $0.77$ (corresponding to a regular value) for Model 2 to $0.64$ (corresponding to a low value) for Model 3 as reported on Table \ref{tableRho}. The three models can all be represented as a Heston stochastic volatility model (SV) with U-shape intraday volatility pattern and a possible jump whose occurrence time is picked up uniformly randomly on a subinterval $[T^{(0)},T^{(1)}]$ of $[0,T]$. Except for the jump component, this general model is directly inspired from Model 4 in \cite{andersen2012jump}, and \cite{xiu2010quasi} (see Section 6.1 on p. 242). We assume that the log price process $X_t$ and the volatility process $\sigma_t$ follow the dynamics
\begin{eqnarray*}
dX_t &=& \mu dt + \sigma_{t-}dW_t,\\
\sigma_t &=& \sigma_{t,SV}\sigma_{t,U},
\end{eqnarray*}
with
\begin{eqnarray*} 
d\sigma_{t,SV}^2 & = & \alpha(\bar{\sigma}^2 - \sigma_{t,SV}^2)dt + \delta \sigma_{t,SV}d\bar{W}_t,\\   
\sigma_{t,U} & = & C + Ae^{-at/T} + De^{-b(1-t/T)} -  \beta\sigma_{\tau-,U}\mathbb{1}_{\{t \geq \tau \}}.
\end{eqnarray*}
Here $W_t$ and $\bar{W}_t$ are two standard Brownian motions with $d\langle W,\bar{W} \rangle_t = \phi dt$. Note that $\sigma_{t,U}$ jumps at time $\tau$, that we define as a uniform random variable on $[T^{(0)},T^{(1)}]$. $\beta$ controls the size of the jump. The choice of making $\sigma_{t,U}$ jumps, instead of the global volatility $\sigma_t$, is merely a way to ensure that $\sigma_t$ remains positive. Finally, the drift parameter $\mu$ and the stochastic volatility part remain constant for each model. The corresponding parameters are chosen consistently with the ones from Section 6.1 (p. 242) in \cite{xiu2010quasi}, that is $\mu = 0.03$, $\alpha = 5$, $\bar{\sigma}^2 = 0.1$, $\delta = 0.4$, $\phi = -0.75$. Finally, $\sigma_{0,SV}^2$ is sampled from a Gamma distribution of parameters $(2\alpha\bar{\sigma}^2/\delta^2,\delta^2/2\alpha)$, which corresponds to the stationary distribution of the CIR process.

\subsubsection*{Model 1: SV + steep U (HIGH $\rho$)}

The first model does not incorporate the jump in volatility, i.e. we set $\beta = 0$. The parameters of the U-shape part are set to generate a steep slope, which in turn lowers somewhat the value of $\rho$  compared to Model 4 in \cite{andersen2012jump} where we find that the corresponding mean $\rho$ value is too high to be consistent with $\rho_{high}$ (which we recall is the empirical high value reported in Section \ref{limitations}). With $C=0.83$, $A=1.26$, $D=0.42$, $a=10$, $b=10$, this model presents a sample mean value of $\rho_{mean} = 0.89$, which is slightly bigger than $\rho_{high}=0.83$. We are conservative in this first model to show what happens to the local method in a very unlikely bad situation for AVAR reduction, i.e. a very high $
\rho_{mean}$.

\subsubsection*{Model 2: SV + normal U + 1 Jump (REGULAR $\rho$)}

In this model, the U-shape intraday volatility parameters are set to values that are consistent with those chosen in Model 4 in \cite{andersen2012jump}, that is $C=0.75$, $A=0.25$, $D=0.89$, and $a=b=10$. The jump size parameter is set to $\beta = 0.5$, that is a jump of 50\% in size at the random time $\tau$. We set $T^{(0)} = 0$, $T^{(1)} = T$ and thus let $\tau$ take values on the whole time interval. Such friction in the volatility process leads to lower values of $\rho$ and $\kappa$ compared to Model 1, with a sample mean equal to $\rho_{mean} = 0.77$. This is thus a very realistic model in terms of measure of heteroskedasticity as $\rho_{mean} = \rho_{regular}$. It is also possible to obtain $\rho_{mean} = \rho_{regular}$ in an alternative continuous volatility model with normal U by taking a 2-factor stochastic volatility model (SV2F) as in \cite{barndorff2008designing} (Section 6.2, p. 1511), with parameters tuned such that the trajectories are rough enough. The results from Section \ref{results} would be similar. As a byproduct, Model 2 shows that a jump in the volatility can lower significantly the measures of heteroskedasticity $\rho$ and $\kappa$. 

\subsubsection*{Model 3: SV + steep U + 1 Jump (LOW $\rho$)}
 
 This last model is a combination of the first two models. U-shape volatility parameters are set to give the same sloap as for Model 1, and the jump size parameter is set to $\beta = 0.5$ as in Model 2. However, to keep the positivity of $\sigma_t$ we restrain the values of the jump time and set $T^{(0)} = 0.05 T$, $T^{(1)} = 0.7T$. This third scenario is designed to reach volatility paths presenting an heteroskedasticity with a low value of $\rho$ and we report the sample mean $\rho_{mean} = 0.64$, which is almost equal to $\rho_{low} =0.62$. We are in the situation where the global estimators should deviate the most from the bound of efficiency.    
\bigskip

We now turn to the estimation procedure. First, to estimate $K$ on $[0, T]$, we work with the Tukey-Hanning 2 kernel as for the numerical study in \cite{barndorff2008designing} (Section 6, pp. 1510-1513) since it requires reasonable bandwidth sizes $H$, which makes the estimator computable in an acceptable amount of time. Moreover, we do not need too many out-of-period data to compute $\gamma_h$. We implement the feasible adaptive estimator. We arbitrary set the tuning parameters $\theta_i$ equal to 30 seconds and the triangular kernel $f(x) = x \wedge (1-x)$. In practice, we find that the realized kernel is not very sensitive to the dispersion of $\widehat{H}$ in terms of RMSE, so that it is not absolutely necessary to get very accurate pre-estimators. Such robustness proved to be crucial in our procedure as it is well known that estimators for the quarticity can be unstable in finite sample when the amount of data is not large. On each block $[\Tau_{i-1}, \Tau_i]$, we do the same procedure and obtain the corresponding $\tilde{K}_B$ by aggregation. Finally, we compute the QMLE by a numerical maximization of the quasi-likelihood function given in Section \ref{localQMLEsub}. This gives us $Q$ and the local estimates $\tilde{Q}_B$.

\subsection{Discussion on theoretical AVAR reduction}
\label{AVARreduction}
In this section, we propose to look at the theoretical AVAR reduction as a function of $\rho$, and investigate the practical question of how fast the convergence in (\ref{eqAVARRK}) and (\ref{eqAVARQMLE}) is. The model considered for volatility is a deterministic U shape + 1 Jump, which corresponds to Model 2 without the stochastic volatility part. Here we generate different values of $\rho$ as a function of the jump time, which we restrict to be in $\big[0.013T, T \big]$ so that each $\rho$ can be associated to a distinct jump time on that interval. We choose this particular model because the sample mean of $\rho$ is .77 which corresponds to a regular value, and the panel of generated $\rho$ values is sufficiently large compared to the other two models.

\smallskip
The values of $L_B^{(RK)}$ and $L_B^{(QMLE)}$ are plotted as a function of $\rho$ in the upper panels of Figure \ref{AVARblockmod2} for a realistic continuous U-shape with one jump volatility model where the sample mean $.77$ corresponds to a regular value of $\rho$\footnote{Useful details on this model can be found in Section \ref{simulations}.}. As we can see, the convergence in (\ref{eqAVARRK}) is very fast. When $\rho = .77$, the QMLE loss is almost divided by 4 when considering 2 blocks instead of 1, with $L_1^{(QMLE)} \approx 16 \%$ and $L_2^{(QMLE)} \approx 5 \%$. In the same setting the RK loss goes from $L_1^{(RK)} \approx 16 \%$ to $L_2^{(RK)} \approx 8 \%$. If we consider the lower value $\rho = .62$, the QMLE losses for the first four values of $B$ are equal to $L_1^{(QMLE)} \approx 35 \%$, $L_2^{(QMLE)} \approx 19 \%$, $L_3^{(QMLE)} \approx 11 \%$ and $L_4^{(QMLE)} \approx 6 \%$. The corresponding RK values are $L_1^{(RK)} \approx 28 \%$, $L_2^{(RK)} \approx 17 \%$, $L_3^{(RK)} \approx 11 \%$ and $L_4^{(RK)} \approx 8 \%$. This suggests that the convergence to the loss bounds (which we recall to be equal to $L_{\infty}^{(RK)} = 3.625 \%$ when considering the RK Tukey-Hanning 2 and $L_{\infty}^{(QMLE)} = 0 \%$ for the QMLE) is very fast for both approaches. Actually for any reasonable $\rho$ taken to be between 0.5 and 1, choosing $B=8$ is big enough for the loss to stay within $L_{\infty}^{(RK)} + 4\%$ (or $L_{\infty}^{(QMLE)} + 4\%$), and it is usually far below this threshold with regular and high values of $\rho$.  

\smallskip
Moreover, we can see on the left lower panel in Figure \ref{AVARblockmod2} that when $\rho$ is relatively high, the QMLE outperforms the RK approach when considering $B=1$, and the gap gets bigger as we increase $B$. In contrast when $\rho < .77$, the QMLE is outperformed when considering only one block, but eventually makes it back when incrementing the value of $B$. The actual value required to fill up the gap is getting bigger as $\rho$ decreases. This suggests that both approaches are complementary to each other. Finally, the lower left panel in Figure \ref{AVARblockmod2} documents that both approaches dominate the PAE regardless of the number of blocks.

\subsection{Results}
\label{results}
We first report the finite sample properties of the four statistics in Table \ref{stdRK}-\ref{stdQMLEfeasible} for Model 2 under the noise level $\xi^2=0.001$. We can see that the results are promising at any level of sampling, as the RMSE of the $Z$-statistic does not suffer much from the increasing in the number of blocks, especially for the QMLE for which the RMSE of $Z_n^{\tilde{Q}_B}$ stays closely in line with $Z_n^{Q}$. The results also indicate that the asymptotic theory eventually kicks in for all the estimators as the standard deviation of the statistics decreases to $1$ when the sampling frequency increases. Nevertheless, we can see a slight over dispersion compared to what was reported in \cite{xiu2010quasi} and \cite{barndorff2008designing}. For the QMLE, this is due to the strong difference with the noise-to-signal ratio that was used in \cite{xiu2010quasi} where $\xi^2 \approx 0.06$. Concerning the RK, the difference in the studentization is due to the fact that the authors in \cite{barndorff2008designing} do not employ $AVAR_{[0,T]}^{(RK)}$ for the studentization, but a non-asymptotic variance as documented in Section 4.4 (pp. 1496-1498) of their work. The feasible statistics are slightly biased, and this is due to the estimation of the AVAR procedure.

\smallskip
We then report the theoretical loss values $\tilde{L}_B^{(\Sigma)}$ and the empirical loss $\breve{L}_B^{\Sigma}$ for two levels of sampling $n=23,400$ and $n=46,800$, and three levels of noise-to-signal ratios $\xi^2 = 0.01$, $\xi^2 = 0.001$ and $\xi^2 = 0.0002$ in Table \ref{tableLossXimedium}. First, we can note that the theoretical loss behaves in a very similar way as in Section \ref{AVARreduction} for the three models. In particular, this implies that neither the SV part nor the steep U component seems to have a bad impact for the local method. Also, one can see that when choosing $B=8$ the theoretical loss is at most $3.2 \%$ more than the parametric loss (which we recall to be equal to 3.625 \% for the RK Tukey-Hanning 2 and 0\% in the case of the QMLE), which are in line with the threshold found in Section \ref{AVARreduction}.

\smallskip
Second, the loss due to the finite sample behaves in a very proper way when $B$ increases. For any setting and both estimators, it is roughly constant as a function of $B$, although suffering more when $\rho$ is higher and $n$ smaller. This is perfectly in line with the findings in Table \ref{stdRK} and Table \ref{stdQMLE}. In particular for Model 2 and Model 3, the finite sample effect is almost not moving as $B$ increases. For Model 1, this is basically the same picture for the QMLE, but the empirical loss seems to stagnate between $B=4$ and $B=8$ when using the RK. This is not surprising as the RK suffers more from the finite sample effect than the QMLE as seen in Table \ref{stdRK} and Table \ref{stdQMLE}.  

\smallskip
Third, note that the decomposition 
 \begin{eqnarray}
 \label{vardec}
\breve{L}_B^{\Sigma} + 1 \approx \underbrace{\l( \tilde{L}_B^{\Sigma} +1 \r)}_{\textnormal{Due to the theoretical loss}} \times \underbrace{\textnormal{Var}_M\l[Z_n^\Sigma\r]}_{\textnormal{Due to the finite sample}},
\end{eqnarray}
where $\textnormal{Var}_M [X]$ denotes the sample variance of $X$ based on the $M$ Monte Carlo simulations, is numerically well-verified and gives an intuitive interpretation of the main sources of deviation from the bound in practice. For instance, consider $\tilde{Q}_2$ on Model 2, with $n = 23,400$, $\xi^2 = 0.001$. In that case, the previous decomposition (\ref{vardec}) gives $0.218 + 1 = 1.218 $ for the left hand side, and $(0.124 +1) \times 1.048^2 \approx 1.220 $ for the right hand side which is very close to the other value indeed.  

\smallskip
Finally, this simulation study indicates that the local version of RK and the QMLE perform very well in practice, with the QMLE slightly more robust to the values of $n$ and $B$ as free of tuning parameters. 

\section{Empirical illustration} \label{empiricalIllustration}

We conclude this study by the application of our method on transaction log prices of Intel Corporation (INTC) shares recorded on the NASDAQ stock market over the year 2015. We exclude January 1, the day after Thanksgiving and December 24 which are less active, thus this leaves us with 250 trading days of data. Moreover, we only keep transactions that were carried out between 9:30am and 4pm. Finally, we consider the data in tick time, for an average of $6,139$ daily trades. The most active days include more than $15,000$ trades.

\smallskip
We first estimate the theoretical gain in AVAR. As for the numerical study, we do not cap $\widehat{\rho}_{\Tau_{i-1}, \Tau_i}$ by $1$. Across the days, values of $B$ and blocks corresponding to an overall of 5,250 estimates, the value $1.1$ was crossed only a few times. We report in Table \ref{tableRho2} key statistics for AVAR reduction. We get a global estimate of $\rho$ around $0.74$, which is very close to $\rho_{regular}$. Across the year the estimates of $\rho$ ranged from around $0.3$ to $1$, and actually crossed $1$ for two days where it reached $1.03$ and $1.04$. When $B$ increases, we find as expected that $\widehat{\rho}_B$, the mean estimated value of $\rho$ across days and blocks, also increases to reach a value of $0.86$ for $8$ blocks. Accordingly, the mean estimated ratios of AVAR decreases from $1$ to $0.9$ for the QMLE, and from $1$ to $0.92$ for the RK. Moreover, we find that those ratios are consistently smaller than $1$ for the 250 days and different values of $B$ bigger than 1, so that the local method never deteriorates the AVAR of the estimator. Note that the same ratios for Model 2 in our simulation study range from $1$ to $102.4/121.5 = 0.84$ for the QMLE and from $1$ to $105.6/118.2=0.89$ for the RK. The slight disparity between the empirical study and Model 2 can be explained in several ways. For example, it is likely that we still under-evaluate the difference between $\rho$ and $\rho_{\Tau_{i-1}, \Tau_i}$, or that the theoretical model is a little too optimistic about how fast $\rho_{\Tau_{i-1}, \Tau_i}$ gets close to $1$ on local blocks. To sum up, the results are approximately in line with what was expected, and present a substantial gain in terms of AVAR for both the QMLE and the RK.    

\smallskip
The last column in Table \ref{tableRho2} shows the empirical correlation between the correction terms $\tilde{Q}_B - Q$ and $\tilde{K}_B - K$ for several values of $B$. The positive correlation indicates that the local method tends to correct the global estimates in the same direction for both the QMLE and the RK. Moreover, increasing the number of blocks $B$ amplifies the phenomenon. Table \ref{tableEstEmp} shows the empirical mean and standard deviation of the $10$ estimators. Note that the main source of randomness being the target value itself, it is not surprising to find the mean and standard values very close to each other. We have reported in the last column the correlation between each estimator and the global QMLE. We find results very close to $1$ for all estimators. One should note that the global RK is less correlated to the QMLE than all the local QMLE $\tilde{Q}_B$. This indicates that the order of magnitude of the correction induced by the local method is smaller than the difference between the two global estimators. 

\smallskip
Finally, Figure \ref{CI} shows daily 95\% theoretical confidence intervals for $Q$, $\tilde{Q}_8$, $K$ and $\tilde{K}_8$ in May 2015. We can see that the confidence intervals for the local estimators are often shorter than their counterpart. Moreover, over the year the global and the local estimates confidence intervals always overlap, corroborating the fact that the local estimates are in line with their global versions.  

\section{Conclusion}
In this paper, we have looked at the efficiency of local methods to estimate integrated volatility. We have shown that for the RK and the QMLE, if we chop the data into $B$ blocks we can reduce the AVAR when $B$ is fixed and retrieve the parametric loss when $B$ goes to infinity. We have also seen that the theoretical gain is mostly preserved when looking at finite sample results. Finally, we have documented that the gain is substantial in practice.

\smallskip
Given how simple to implement the methodology is, we expect that it will be very helpful for practitioners. Our hope is that this simple and natural technique will be used on the QMLE and the RK, but also considered for a wider class of estimators. It is clear that the theory would work for the PAE and the MSRV, but econometricians should also try it on their own favorite estimator. Actually, the technique can be applied to other problems, such as the high-frequency covariance estimation, the estimation of functions of volatility, the leverage effect, the volatility of volatility, etc.

\section{Appendix: proofs} \label{Appendix}

\subsection{Simplification of the problem}  \label{simplification}
Since we want to prove stable convergence, in view of the componentwise local boundedness of the matrix
$$\left(\begin{matrix}  \sigma_t & 0 \\ 
\tilde{\sigma}_t^{(1)} & \tilde{\sigma}_t^{(2)} 
\end{matrix}\right),$$ 
 and because $\inf_t (\min(\sigma_t, \tilde{\sigma}_t^{(2)})) > 0$,
we can without loss of generality assume that for all $t \in [0,T]$ there exists some nonrandom constants $\underline{\sigma}$ and $\overline{\sigma}$ such that
\begin{eqnarray}
\label{smsp}
0 < \underline{\sigma} < \sigma_t, \tilde{\sigma}_t^{(1)}, \tilde{\sigma}_t^{(2)} < \overline{\sigma},
\end{eqnarray}
by using a standard localization argument (e.g., Section 2.4.5 of \cite{mykland2012econometrics}). One can further suppress $b_t$ as in Section 2.2 (pp. 1407-1409) of \cite{mykland2009inference}, and act as if 
$X_t$ is a martingale. Also, we follow a similar procedure to localize the random variables $U_i^n$ as, e.g, in the proof of Lemma 14.1.5 p.435, Equation (14.1.13), in \cite{jacod2011discretization}. Consequently, we will assume in the following of the proof: 

\textbf { (H) } We have $b = \tilde{b} = 0$. Moreover $\sigma$, $\sigma^{-1}$, $\tilde{\sigma}^{(1)}$, $(\tilde{\sigma}^{(1)})^{-1}$, $\tilde{\sigma}^{(2)}$, $(\tilde{\sigma}^{(2)})^{-1}$,  $\alpha$, $\alpha^{-1}$ are bounded. Given an \textit{a priori} number $\gamma >0$, we also have $\sup_{0 \leq i \leq N_n} U_i^n \leq n^{\gamma}$. 

\smallskip
In particular, \textbf{(H)} implies, taking $\gamma$ small enough, that $\pi_
T^n < 1$, for $n \in \naturels$ large enough.

\smallskip
We define $\calu := \sigma \l\{U_i^n | i,n \in \naturels \r\} \vee \sigma \l\{ \alpha_s | 0 \leq s \leq T\r\}$ the $\sigma$-field that generates the observation times and which is independent of $X$. We will often have to use the conditional expectation $\esp[.|\calu]$, that we hereafter denote for convenience $\esp_{\calu}$. We also define the discrete filtration $\calg_i^n := \calf_{\ti{i}}^X \vee \calu$, and recall the continuous version $\calg_t := \calf_t^X \vee \calu$ where $\calf_t^X$ is the canonical filtration associated to $X$. Note that by independence from $\alpha$, $X$ admits the same It\^{o} semimartingale dynamics in the extension $\calg$.

\smallskip
Note also that, by virtue of Lemma 14.1.5 in \cite{jacod2011discretization}, recalling $\pi_t^n := \sup_{i \geq 1} t_{i}^n-t_{i-1}^n$, and $N_n(t) = \sup \{ i\in \naturels - \{0\} | t_{i}^n  \leq t \}$ we have 
\bea 
\eta >0 \implies n^{1-\eta}\pi_t^n \overset{\proba}{\rightarrow} 0.
\label{eqStepsize}
\eea 
Throughout the proofs, we write $N_n$ for $N_n(T)$. We also define $L_n = N_n^{1/2 + \delta}$, for some $\delta >0$ to be adjusted, and we let $L$ be a positive constant that may vary from one line to the other. Finally we often refer to the continuous part of $X_t$ defined as 
\bea 
\tilde{X}_t :=  X_0 + \int_0^t{b_s ds} + \int_0^t{\sigma_s dW_s}.
\eea 

\subsection{Proof of (\ref{RhoKappa})}
We first show the left hand side inequality, that can be reformulated as $\kappa_{r,s}^{2/3} \geq \rho_{r,s}$. Note that by an immediate application of H\"older's inequality we have 
\beas  
\int_r^s{\sigma_u^2du}\leq (s-r)^{1/3}\l(\int_r^s{\sigma_u^3du}\r)^{2/3}.
\eeas  
Thus,
\beas  
\rho_{r,s} &=& \frac{\int_r^s{\sigma_u^2du}}{(s-r)^{1/2} \l(\int_r^s{\sigma_u^4du}\r)^{1/2}}\\
&\leq& \frac{\l(\int_r^s{\sigma_u^3du}\r)^{2/3}}{(s-r)^{1/6}\l(\int_r^s{\sigma_u^4du}\r)^{1/2}} = \kappa_{r,s}^{2/3}.
\eeas 
For the right hand side inequality, we first consider the domination
\beas  
\int_r^s{\sigma_u^3du} \leq \l(\int_r^s{\sigma_u^2du}\r)^{1/2}\l(\int_r^s{\sigma_u^4du}\r)^{1/2},
\eeas 
which is obtained by Cauchy-Schwarz inequality. Then we inject this expression in $\kappa_{r,s}$ and we get
\beas 
\kappa_{r,s} &=& \frac{\int_r^s{\sigma_u^3du}}{(s-r)^{1/4}\l(\int_r^s{\sigma_u^4du}\r)^{3/4} }\\
&\leq& \frac{\l(\int_r^s{\sigma_u^2du}\r)^{1/2}}{(s-r)^{1/4}\l(\int_r^s{\sigma_u^4du}\r)^{1/4}}= \rho_{r,s}^{1/2}.
\eeas

\subsection{Estimates for the efficient price $X$}

Hereafter, we adopt the following notation convention. For a process $V$ (including the noise process $\epsilon$ by a slight "abuse of notation"), and $t \in [0,T]$ we write $\Delta V_t = V_t - V_{t-}$, $\Delta V_{i}^n := V_{t_i^n} -V_{t_{i-1}^n}$ and $\Delta V^n := (\Delta V_{1}^n, \cdots,  \Delta V_{N_n}^n)$. Finally, for interpolation purpose we sometimes write the continuous version $\Delta V_{i,t}^n := V_{t_i^n \wedge t} -V_{t_{i-1}^n \wedge t} $, along with the time increment $\Delta t_{i,t}^n := \ti{i} \wedge t -\ti{i-1} \wedge t$. We introduce the two following quantities:
\bea 
\zeta_{i,t}^n := (\Delta \tilde{X}_{i,t}^n)^2 -\sigma_{t_{i-1}^n}^2\Delta t_{i,t}^n, \textnormal{ and } \bar{\zeta}_{i,t}^n := \esp \l[\zeta_{i,t}^n | \calg_{i-1}^n\r].   
\eea 

We have the following estimates 

\begin{lemma*}\label{lemmaEstimateX}
We have, for some constant $L >0$ independent of $i$,
\bea
 \esp \Big[  \sup_{t \in [\ti{i-1},\ti{i}]}|\Delta \tilde{X}_{i,t}^n|^p \Big| \calg_{i-1}^n \Big]  & \leq &  Ln^{-p/2}(U_i^n)^{p/2} ,
 \label{eqDeltaX} \\
\big|\bar{\zeta}_{i,t}^n \big| & \leq & L n^{-3/2} (U_i^n)^{3/2} , 
\label{eqZetaTilde}\\
 \esp \l[ \l. \l(\zeta_{t,i}^n\r)^p \r| \calg_{i-1}^n \r] & \leq & Ln^{-p} (U_i^n)^{p},
 \label{eqZeta2}\\
\esp \Big[ \Big|\int_{\ti{i-1} \wedge t }^{\ti{i} \wedge t}\sigma_s^2ds - \sigma_{t_{i-1}^n}^2\Delta t_{i,t}^n\Big|^p \Big| \calg_{i-1}^n \Big] & \leq & Ln^{-3p/2}(U_i^n)^{3p/2}.
 \label{eqDevSigma}
\eea 
\end{lemma*}

\begin{proof}
For (\ref{eqDeltaX}), this is a consequence of the fact that by the conditional Burkholder-Davis-Gundy inequality, we have
\beas  
\esp \l[ \l. \sup_{t \in [\ti{i-1},\ti{i}]} \l| \int_{\ti{i-1} \wedge t}^{\ti{i} \wedge t}{\sigma_sdW_s}\r|^p \r| \calg_{i-1}^n \r] &\leq& \esp \l[ \l. \sup_{t \in [\ti{i-1},\ti{i}]} \l|\int_{\ti{i-1} \wedge t}^{\ti{i} \wedge t}{\sigma_s^2ds}\r|^{p/2} \r| \calg_{i-1}^n \r]\\
&\leq& L \l(\ti{i} -\ti{i-1}\r)^{p/2}.
\eeas  
Since $\alpha$ is bounded by assumption \textbf{(H)}, and since $\ti{i} -\ti{i-1} < 1$, we get (\ref{eqDeltaX}). The other estimates are straightforwardly obtained using the same line of reasoning and It\^{o} formula.
\end{proof}

\subsection{Proof of Theorem \ref{CLTQMLE} and Theorem \ref{CLTQMLEjumps}} 
We adopt the general setting introduced in Section \ref{robustness} and Section \ref{simplification}. We start by showing the consistency of the QMLE along with other estimates in the case $B=1$. We then adapt and combine those results in the case $B \geq 1$ to derive the central limit theorem stated in Theorem \ref{CLTQMLEjumps}. As a byproduct, Theorem \ref{CLTQMLE} will also be proven. 

\smallskip
When $B=1$, we recall that for any $\xi = (\sigma^2,a^2) \in \Xi := [\underline{\Sigma}, \overline{\Sigma}] \times [\underline{a}^2, \overline{a}^2]$, $\underline{a}^2 > 0$, we have, up to a constant term  
\bea 
l_n(\xi) = - \half \textnormal{log det}(\Omega) -\half Y^T \Omega^{-1} Y,
\eea
with $\Omega^{-1} = \l[\omega^{i,j}\r]_{1 \leq i \leq N_n, 1 \leq j \leq N_n}$. The exact definition of the coefficients $\omega^{i,j}$ can be found in e.g. (28), p. 245 of \cite{xiu2010quasi}, replacing $n$ by $N_n$. We define the approximate log-likelihood random field as 
\bea 
\bar{l}_n(\xi) = -\half \textnormal{log det}(\Omega) -\half  \textnormal{Tr}\l(\Omega^{-1} \l\{\Sigma_0^c + \Sigma_0^d \r\} \r),
\eea 
with 
\beas 
\Sigma_0^c &=& \left(\begin{matrix}
                    \int_0^{t_{1}^n}{\sigma_s^2ds} + 2 a^2 & - a^2 & 0 & \cdots & 0 \\
                    - a^2 & \int_{t_{1}^n}^{t_{2}^n}{\sigma_s^2ds} +  2 a^2 & - a^2 & \ddots & \vdots \\
                    0 & - a^2 & \int_{t_{2}^n}^{t_{3}^n}{\sigma_s^2ds}+ 2 a^2 & \ddots & 0\\
                    \vdots & \ddots & \ddots & \ddots & - a^2\\
                    0 & \cdots & 0 & - a^2 & \int_{t_{N_n-1}^n}^{t_{N_n}^n}{\sigma_s^2ds} + 2 a^2
                  \end{matrix}\right), \\[12pt]
\eeas 
and
\beas 
\Sigma_0^d = \textnormal{diag}\l( \sum_{0 < s\leq t_1^n} {\Delta J_s^2} ,\sum_{t_1^n < s\leq t_2^n} {\Delta J_s^2},\cdots, \sum_{t_{N_n-1}^n < s\leq t_{N_n}^n} {\Delta J_s^2}\r). 
\eeas
We further define the diagonal scaling matrix 
\beas  
\Phi_n = \textnormal{diag}(N_n^{1/2},N_n),
\eeas  
and consider for $\xi \in \Xi$ the scaled score functions 
$$ \Psi_n(\xi) = -\Phi_n^{-1} \frac{\partial l_n(\xi)}{\partial \xi} \textnormal{ and } \bar{\Psi}_n = -\Phi_n^{-1} \frac{\partial \bar{l}_n(\xi)}{\partial \xi}.$$
We start by showing the consistency of the QMLE. Before stating the result, we give a few definitions. For a matrix $A = \l[ a_{i,j}\r]_{1 \leq i \leq N_n, 1 \leq j \leq N_n}\in \reels^{N_n \times N_n}$, we associate the matrix $\dot{A} = [\dot{a}_{i,j}]_{0 \leq i \leq N_n ,1 \leq j \leq N_n } \in \reels^{(N_n+1)\times N_n }$ and $\ddot{A} = [\ddot{a}_{i,j}]_{0 \leq i N_n ,0 \leq j \leq N_n } \in \reels^{(N_n+1)\times (N_n+1) }$ whose components respectively satisfy 
\beas  
\dot{a}_{i,j} = a_{i+1,j} - a_{i,j},  
\eeas  
and
\beas  
\ddot{a}_{i,j} = \dot{a}_{i,j+1} - \dot{a}_{i,j} = a_{i+1,j+1} - a_{i,j+1} + a_{i,j} - a_{i+1,j},
\eeas
with the convention $a_{i,j} = 0$ when $i=0$ or $j=0$. This will be useful to disentangle some quadratic expressions using the following result.
\begin{lemma*}  \label{lemmaTransfo}
Let $y,z \in \reels^{N_n+1}$, with $y = (y_0,\cdots,y_{N_n})^T$, $z = (z_0,\cdots,z_{N_n})^T$. We define $\Delta y = (\Delta y_1,\cdots,\Delta y_{N_n}) := \l(y_1-y_0,\cdots,y_{N_n} - y_{N_n-1}\r)^T \in \reels^{N_n}$, and $\Delta z$ the same way.  Then we have the by-part summation identities

$$\Delta y^T A \Delta z = - y^T \dot{A} \Delta z = y^T \ddot{A} z.$$


\end{lemma*}
We now show a preliminary lemma to get the consistency of the QMLE.
\begin{lemma*} \label{lemmaUnifConvScore} (Asymptotic score) For any $\xi \in \Xi$, let
\beas  
\Psi_{\infty}(\xi) = \l( \begin{matrix}  -\inv{8a\sigma^3\sqrt T}\l(\int_0^T\sigma_s^2ds+ \sum_{0<s\leq T}\Delta J_s^2 -\sigma^2T\r)-\frac{\sqrt T}{8a^3\sigma}\l(a^2-a_0^2\r) \\ \inv{2a^4}\l(a^2-a_0^2\r) \end{matrix} \r).
\eeas
We have 

\bea \sup_{\xi \in \Xi} \l|\Psi_n(\xi)-\Psi_\infty(\xi)\r| \overset{\proba}{\rightarrow} 0.
\label{eqConsistency0}
\eea 

\end{lemma*}

\begin{proof} 
We start by treating the case where the jump part $J=0$. We have the decomposition 
\begin{eqnarray}
\label{Psin}
\Psi_n = \bar{\Psi}_n + R_n,
\end{eqnarray}
with 
\beas 
R_n(\xi) & = &  \l( \begin{matrix} \inv{2\sqrt N_n}\l\{ Y^T \frac{\partial \Omega^{-1}}{\partial \sigma^2}Y - \textnormal{tr} \l( \frac{\partial \Omega^{-1}}{\partial \sigma^2} \l\{\Sigma_0^c + \Sigma_0^d \r\}\r)\r\}  \\ \inv{2 N_n}\l\{ Y^T \frac{\partial \Omega^{-1}}{\partial a^2}Y - \textnormal{tr} \l( \frac{\partial \Omega^{-1}}{\partial a^2} \l\{\Sigma_0^c + \Sigma_0^d \r\}\r)\r\}  \end{matrix} \r).
\eeas
By a straightforward adaptation of the proof of Lemma 1-2 and Theorem 4 in \cite{xiu2010quasi}, we have immediately that $R_n = o_\proba(1)$ uniformly in the parameters since the step size of the observation grid $\pi_T^n \overset{\proba}{\rightarrow} 0$ by (\ref{eqStepsize}). Thus it is sufficient to show that we have 
$$\sup_{\xi \in \Xi} \l|\bar{\Psi}_n(\xi)-\Psi_\infty(\xi)\r| \overset{\proba}{\rightarrow} 0.$$
Equality (\ref{eqConsistency0}) is then a direct consequence of equations (38) and (40) pp. 247-248 in \cite{xiu2010quasi} that are obtained following exactly the same proof as pp.247-248 for an irregular grid such that $\pi_T^n \overset{\proba}{\rightarrow} 0$.

\smallskip
When there are jumps, there is an additional term in (\ref{Psin}) which is equal to 
\bea 
A_n(\xi) =  \l( \begin{matrix} \inv{2\sqrt N_n}\l\{\l(\Delta J^n\r)^T \frac{\partial \Omega^{-1}}{\partial \sigma^2} \Delta J^n + 2\l(\Delta J^n\r)^T \frac{\partial \Omega^{-1}}{\partial \sigma^2} \l\{\Delta \tilde{X}^n + \Delta \epsilon^n\r\} \r\} \\ \inv{2 N_n}\l\{\l(\Delta J^n\r)^T \frac{\partial \Omega^{-1}}{\partial a^2} \Delta J^n + 2\l(\Delta J^n\r)^T \frac{\partial \Omega^{-1}}{\partial a^2} \l\{\Delta \tilde{X}^n + \Delta \epsilon^n\r\} \r\} \end{matrix} \r),
\eea 
so that it is sufficient to show that we have 
\bea 
\sup_{\xi \in \Xi} \l|A_n(\xi)+\l( \begin{matrix}\inv{8a\sigma^3\sqrt T}\sum_{0<s\leq T}\Delta J_s^2 \\0 \end{matrix} \r)\r| \overset{\proba}{\rightarrow} 0. 
\eea 
We first compute the limit of the term $\inv{2\sqrt N_n}\l(\Delta J^n\r)^T \frac{\partial \Omega^{-1}}{\partial \sigma^2} \Delta J^n $. Recalling that $\omega^{i,j}$ is the $(i,j)$-th index of $\Omega^{-1}$, we provide the following decomposition: 
\bea 
\inv{2\sqrt N_n} \l(\Delta J^n\r)^T \frac{\partial \Omega^{-1}}{\partial \sigma^2} \Delta J^n = \inv{2\sqrt N_n} \sum_{i = 1}^{N_n} \frac{\partial \omega^{i,i}}{\partial \sigma^2} \l(\Delta J_{i}^n\r)^2 + \inv{\sqrt N_n} \sum_{1 \leq j < i \leq N_n}{\frac{\partial \omega^{i,j}}{\partial \sigma^2} \Delta J_{i}^n\Delta J_{j}^n}.
\eea
Now, we define $\tau_1,\cdots,\tau_{N^J}$ the jump times of $J$, where $N^J$ is the random number of jumps of $J$ on $[0,T]$. Since $N^J$ is finite, there exists a random number $K^J$ such that for $n \geq K^J$ we have 
\bea 
\inv{\sqrt N_n} \sum_{i = 1}^{N_n} \frac{\partial \omega^{i,i}}{\partial \sigma^2} \l(\Delta J_{i}^n\r)^2 = \inv{2\sqrt N_n} \sum_{k=1}^{N^J} \frac{\partial \omega^{N_n(\tau_k),N_n(\tau_k)}}{\partial \sigma^2} \Delta J_{\tau_k}^2.
\eea 
By direct calculation from the expression of the coefficients of $\Omega^{-1}$ in (28) p. 245 in \cite{xiu2010quasi}, we easily deduce that for each $k$ we have $\inv{2\sqrt N_n} \frac{\partial \omega^{N_n(\tau_k),N_n(\tau_k)}}{\partial \sigma^2} \overset{\proba}{\rightarrow} -\inv{8a\sigma^3\sqrt T}$ uniformly in $\xi \in \Xi$. Since the sum is finite, this yields the uniform convergence 
\bea 
\sup_{\xi \in \Xi}{\l|\inv{2\sqrt N_n} \l(\Delta J^n\r)^T \frac{\partial \Omega^{-1}}{\partial \sigma^2} \Delta J^n + \inv{8a\sigma^3\sqrt T}  \sum_{0<s\leq T}\Delta J_s^2  \r|}\overset{\proba}{\rightarrow} 0.
\eea
By a similar argument, we also have for $k\neq l $ that $\frac{\partial \omega^{N_n(\tau_k),N_n(\tau_l)}}{\partial \sigma^2} \overset{\proba}{\rightarrow} 0$ exponentially so that we have $\inv{\sqrt N_n} \sum_{1 \leq j < i \leq N_n}{\frac{\partial \omega^{i,j}}{\partial \sigma^2} \Delta J_{i}^n\Delta J_{j}^n} \overset{\proba}{\rightarrow} 0$ uniformly. As for $\inv{\sqrt N_n}\l(\Delta J^n\r)^T \frac{\partial \Omega^{-1}}{\partial \sigma^2} \l\{\Delta \tilde{X}^n + \Delta \epsilon^n\r\}$, on the one hand the same computation yields that the leading term of $\inv{\sqrt N_n}\l(\Delta J^n\r)^T \frac{\partial \Omega^{-1}}{\partial \sigma^2} \Delta \tilde{X}^n$ is
\beas 
 \sum_{k=1}^{N^J} \underbrace{\inv{\sqrt N_n}\frac{\partial \omega^{N_n(\tau_k),N_n(\tau_k)}}{\partial \sigma^2}}_{O_\proba(1)} \Delta J_{\tau_k} \underbrace{\Delta \tilde{X}_{\ti{i_k}} }_{o_\proba(1)},
\eeas 
where 
$ \ti{i_k} \leq \tau_k \leq \ti{i_k+1}$, so that as the sum is finite the expression is negligible. On the other hand, we also have by Lemma \ref{lemmaTransfo} that the leading term of the noise part is 
\beas 
 -\sum_{k=1}^{N^J} \underbrace{\inv{\sqrt N_n}\frac{\partial \dot{\omega}^{N_n(\tau_k),N_n(\tau_k)}}{\partial \sigma^2}}_{O_\proba(N_n^{-1/2})} \Delta J_{\tau_k} \underbrace{\epsilon_{\ti{i_k}} }_{O_\proba(1)}
\eeas 
since $\frac{\partial \dot{\omega}^{i,j}}{\partial \sigma^2} = O_\proba(1)$ by direct calculation. Finally, similar reasoning shows that the second component of $A_n$ is negligible because of the scaling in $N_n^{-1}$ instead of $N_n^{-1/2}$, and we are done.  
\end{proof} 
Now we turn to the consistency of the QMLE. 
\begin{theorem*} (consistency).
If $\widehat{\xi}_n = (\widehat{\sigma}_n^2,\widehat{a}_n^2)$ is the QMLE, we have 

\bea 
\widehat{\xi}_n \overset{\proba}{\rightarrow} \xi_0:= \l(\overline{\sigma}_0^2, a_0^2\r), 
\label{eqConsistency}
\eea 
where we recall that $T\overline{\sigma}_0^2 = \int_0^T{\sigma_s^2ds} + \sum_{0<s\leq T}\Delta J_s^2$.
\label{thmConsistencyH1}
\end{theorem*}

\begin{proof} 
To get (\ref{eqConsistency}), it is sufficient to have 
\bea 
\sup_{\xi \in \Xi} \l|\Psi_n(\xi)-\Psi_\infty(\xi)\r| \overset{\proba}{\rightarrow} 0,
\eea 
which has been proven in Lemma \ref{lemmaUnifConvScore}, and for any $\epsilon >0$
\bea 
\inf_{\xi\in\Xi:\|\xi-\xi_0\|\geq\epsilon} \| \Psi_\infty(\xi)\|^2>0 = \| \Psi_\infty(\xi_0 )\|^2  \textnormal{  }\proba\textnormal{-a.s},
\label{conditionIdentifiability}
\eea 
by a classical statistical argument (see e.g. \cite{van2000asymptotic}, Theorem 5.9). Given the form of $\Psi_\infty$, the equality $\Psi_\infty(\xi_0) = 0$ is immediate. Note also that the left hand side inequality of (\ref{conditionIdentifiability}) will be automatically satisfied if we show that $\| \Psi_\infty(\xi) \|^2 > 0$ as soon as $\xi \neq \xi_0$ by a continuity argument since $\Xi$ is compact. Let us then take $\xi \in \Xi - \{\xi_0\}$ such that $\Psi_\infty(\xi ) = 0$, and assume first that $a^2 \neq a_0^2$. In that case, we have  
$$ 0 = \| \Psi_\infty(\xi ) \|^2 \geq \inv{4a^8}\l(a^2-a_0^2\r)^2,$$
which leads to a contradiction. Similarly, the first component of $\Psi_\infty$ leads to the domination 
$$ 0 = \| \Psi_\infty(\sigma^2,a_0^2 ) \|^2 \geq \frac{T}{64a_0^2\sigma^6 }\l(\overline{\sigma}_0^2-\sigma^2\r)^2,$$
so that we can conclude $\sigma^2 = \overline{\sigma}_0^2$.

\end{proof} 
We now turn to the convergence of the Fisher information related to our likelihood field. Let $H_n$ and $\bar{H}_n$ be the scaled Hessian matrices of the likelihood fields, defined for any $\xi \in \Xi$ as  

\bea 
H_n(\xi) = -\Phi_n^{-1/2} \frac{\partial^2 l_n(\xi)}{\partial \xi^2} \Phi_n^{-1/2} \textnormal{ and } \overline{H}_n(\xi) = -\Phi_n^{-1/2} \frac{\partial^2 \bar{l}_n(\xi)}{\partial \xi^2} \Phi_n^{-1/2}.
\label{defFisherMatrix}
\eea 

\begin{lemma*} (Asymptotic Fisher information) \label{lemmaFisherConsistency}
Let $\Gamma(\xi_0)$ be the matrix 
\bea 
\Gamma(\xi_0) = \l( \begin{matrix}  \frac{\sqrt T}{8a_0 \overline{\sigma}_0^3}  & 0\\ 
0 & \inv{2a_0^4} \end{matrix} \r).
\eea 
We have, for any ball $V_n$ centred on $\xi_0$, shrinking to $\{\xi_0\}$,
\bea 
\sup_{\xi_n \in V_n} \l\| H_n(\xi_n) - \Gamma(\xi_0) \r\| \overset{\proba}{\rightarrow} 0.
\eea 

\end{lemma*}

\begin{proof}
First note that a small adaptation of Lemma 1-2 with second order derivatives of $\Omega^{-1}$ from \cite{xiu2010quasi} yields $\sup_{\xi \in \Xi }\l\{H_n(\xi) - \overline{H}_n(\xi) \r\} \overset{\proba}{\rightarrow} 0$ since $\pi_T^n \overset{\proba}{\rightarrow} 0$. Now, \cite{xiu2010quasi}, bottom of p. 247 and after equation (41) on p. 248, can be easily adapted to our case replacing $\int_0^T{\sigma_s^2ds}$ by $T\overline{\sigma}_0^2$ as in the previous lemma, so that we have
\bea 
\sup_{\xi \in \Xi} \l\{\overline{H}_n(\xi) -  \overline{H}_\infty(\xi)\r\} \overset{\proba}{\rightarrow} 0,
\eea 
with 
\bea 
\overline{H}_\infty(\xi) := \l( \begin{matrix}  \frac{\sqrt T}{8a\sigma^3} + \frac{(a^2-a_0^2)\sqrt T}{16a^3\sigma^3} +\frac{3T(\overline{\sigma}_0^2-\sigma^2)}{16a\sigma^5\sqrt T} &0 \\ 0&\frac{a_0^2}{a^6}- \inv{2a^4} \end{matrix} \r).
\eea 
It is immediate to check that 
\bea 
\overline{H}_\infty(\xi_0) =  \l( \begin{matrix}  \frac{\sqrt T}{8a_0\overline{\sigma}_0^3}  &0 \\ 0& \inv{2a_0^4}\end{matrix} \r).
\eea 
\end{proof}
We now adopt similar notations to \cite{xiu2010quasi} in the proof of Lemma 3 (p. 248) and define the processes involved in the derivation of the central limit theorem. For $(\beta) \in \{(\sigma^2),(a^2)\}$, and $t \in [0,T]$, we define 
\bea 
M_{1}^{(\beta)}(t) & := & \sum_{i=1}^{N_n(t)}{\frac{ \partial \omega^{i,i}}{\partial \beta}\l\{\l(\Delta X_{i,t}^n\r)^2 - \int_{t_{i-1}^n \wedge t }^{t_{i}^n \wedge t}{\sigma_s^2ds} - \sum_{t_{i-1}^n \wedge t <s \leq t_{i}^n \wedge t} \Delta J_s^2\r\}},
\label{eqM1}\\
 M_{2}^{(\beta)}(t) & := & \sum_{i=1}^{N_n(t)} \l\{\sum_{1 \leq j < i} \frac{ \partial \omega^{i,j}}{\partial \beta} \Delta X_{j,t}^n\r\} \Delta X_{i,t}^n,  
\label{eqM2}\\
 M_{3}^{(\beta)}(t) & := &- 2\sum_{i=0}^{N_n(t)} \l\{\sum_{j = 1}^{N_n(t)} \frac{ \partial \dot{\omega}^{i,j}}{\partial \beta} \Delta X_{j,t}^n\r\} \epsilon_{t_i^n}, 
\label{eqM3}\\
 M_{4}^{(\beta)}(t) & := & \sum_{i=0}^{N_n(t)}{\frac{ \partial \ddot{\omega}^{i,i}}{\partial \beta}\l\{\epsilon_{t_i^n}^2 - a_0^2\r\} } + 2\sum_{i=0}^{N_n(t)} \l\{\sum_{0 \leq j < i} \frac{ \partial \ddot{\omega}^{i,j}}{\partial \beta} \epsilon_{t_j^n}\r\} \epsilon_{t_i^n},
\label{eqM4}
\eea 
where in all the definitions (\ref{eqM1})-(\ref{eqM4}), the terms involving the parameters such as $\Omega^{-1}$, $\dot{\Omega}^{-1}$, $\ddot{\Omega}^{-1}$, $\cdots$ are evaluated at point $\xi := (\sigma^2,a_0^2)$, for some $\sigma^2 \in [\underline{\Sigma}^2,\overline{\Sigma}^2]$. We also define the two-dimensional vectors $M_i(t) := \l(M_{i}^{(\sigma^2)}(t),M_{i}^{(a^2)}(t)\r)$ for $i \in \{1, \cdots ,4\}$. Note that we have the key decomposition 
\beas 
2\Phi_n^{1/2}\l\{\Psi_n(\xi) - \bar{\Psi}_n(\xi)\r\} = \Phi_n^{-1/2} \l\{M_1(T)+2M_2(T)+M_3(T)+M_4(T) \r\}.
\eeas
In the next few lemmas we investigate the limit of each one of those terms. In the presence of jumps and random observation times, we will see that they are not mere extensions of Lemma 3 in \cite{xiu2010quasi} and that additional variance terms appear in the limits. We start by $M_1(T)$. 

\begin{lemma*} \label{lemmaM1}
We have

\beas  
\Phi_n^{-1/2} M_1(T) \overset{\proba}{\rightarrow} 0.  
\eeas  
\end{lemma*}

\begin{proof}
We have to show $N_n^{-1/4} M_1^{(\sigma^2)}(T) \overset{\proba}{\rightarrow} 0$ and $N_n^{-1/2} M_1^{(a^2)}(T) \overset{\proba}{\rightarrow} 0$. We start with the case where $J=0$. We are going to show that for any $(\beta) \in \{(\sigma^2),(a^2)\}$ we actually have $ N_n^{-1/4} M_1^{(\beta)}(T) \overset{\proba}{\rightarrow} 0$. To do so, note that we can write 
\beas  
M_1^{(\beta)}(T) = \sum_{i=1}^{N_n}{\chi_i^n}, 
\eeas  
where 
\beas 
\chi_i^n = \frac{ \partial \omega^{i,i}}{\partial \beta} \l\{\l(\Delta X_{i}^n\r)^2 - \int_{t_{i-1}^n }^{t_{i}^n}{\sigma_s^2ds} \r\}. 
\eeas 
Now, since $\frac{ \partial \omega^{i,i}}{\partial \beta} \in \calu \subset \calg_{i-1}^n$, $\chi_i^n \in  \calg_{i}^n$. Moreover, $\esp \l[\chi_i^n | \calg_{i-1}^n \r] = 0$, thus by Lemma 2.2.11 in \cite{jacod2011discretization}, it is sufficient to show that $N_n^{-1/2}\sum_{i=1}^{N_n}{\esp \l[(\chi_i^n)^2 | \calg_{i-1}^n \r]}\overset{\proba}{\rightarrow}0$. By Burkholder-Davis-Gundy inequality, we have 
\beas  
N_n^{-1/2}\sum_{i=1}^{N_n}\esp \l[(\chi_i^n)^2 | \calg_{i-1}^n \r]  &\leq& 4N_n^{-1/2}\sum_{i=1}^{N_n}\l(\frac{ \partial \omega^{i,i}}{\partial \beta}\r)^2 \int_{\ti{i-1} }^{\ti{i} }{\esp\l[\l.\l(\Delta X_{i,s}^n\r)^2\sigma_s^2 \r|\calg_{i-1}^n\r]ds}\\
&\leq& K N_n^{1/2}n^{-1+\gamma}\sum_{i=1}^{N_n} (\ti{i}  - \ti{i-1} )\\
&\leq& K N_n^{1/2}n^{-1+\gamma} \to^{\proba} 0,
\eeas
where we have used the fact that $\frac{ \partial \omega^{i,i}}{\partial \beta} = O_\proba(N_n^{1/2})$ uniformly in $i$. In the presence of jumps, it remains to show that the additional terms $$N_n^{-1/4}\sum_{i=1}^{N_n}\frac{ \partial \omega^{i,i}}{\partial \beta}\l\{ \l(\Delta J_{i}^n\r)^2 - \sum_{t_{i-1}^n < s \leq t_{i}^n} \Delta J_s^2\r\}$$ 
and $$2N_n^{-1/4}\sum_{i=1}^{N_n}\frac{ \partial \omega^{i,i}}{\partial \beta} \Delta J_{i}^n\Delta \tilde{X}_{i}^n$$ 
are negligible. From the finite activity property, note that the first one is identically $0$ for $n$ sufficiently large. Again, for $n$ sufficiently large, defining $N^J$ the finite number of jumps of $J$ on $[0,T]$, we can write the second term as 
\beas 
2N_n^{-1/4}\sum_{k=1}^{N^J}\underbrace{\frac{ \partial \omega^{N_n(\tau_k),N_n(\tau_k)}}{\partial \beta}}_{O_\proba(N_n^{1/2})} \Delta J_{\tau_k} \underbrace{\Delta \tilde{X}_{i_k}^n}_{O_\proba(n^{-1/2+1/2\gamma})} \overset{\proba}{\rightarrow} 0.
\eeas 
where $i_k$ is such that $\ti{i_k} \leq \tau_k \leq \ti{i_k+1}$, and where we have used \textbf{(H)}. This concludes the proof.

\end{proof}

\begin{lemma*} \label{lemmaM2}
We have $\calg_T$-stably in law that 

\beas  
N_n^{-1/4} M_{2}^{(\sigma^2)}(T) \to \calm\caln\l( 0, \frac{5}{64T^{3/2} \sigma^7 a_0}\int_0^T{\alpha_s^{-1}ds}\l\{\int_0^T{\sigma_s^4 \alpha_sds} + \sum_{0 \leq s \leq T} \Delta J_s^2 (\sigma_s^2\alpha_s + \sigma_{s-}^2\alpha_{s-}) \r\}\r) 
\eeas  
and
\beas  
N_n^{-1/2} M_{2}^{(a^2)}(T) \overset{\proba}{\rightarrow} 0. 
\eeas  
\end{lemma*}

\begin{proof}
As usual, we start by the case with no jumps, that is $J=0$. We show the result for $M_2^{(\sigma^2)}$. The proof is conducted in three steps. 

\textbf{Step 1.} We consider $\beta_{i,k,t}^n := \sigma_{\ti{k}} \Delta W_{i,t}^n$, and we define $\tilde{M}_2^{(\sigma^2)}$ as 
\bea 
\tilde{M}_2^{(\sigma^2)}(t) := \sum_{i=1}^{N_n} \l\{\sum_{(i- L_n) \wedge 1 \leq j < i} \frac{ \partial \omega^{i,j}}{\partial \sigma^2} \beta_{j,i-L_n-1,t}^n\r\} \beta_{i,i-L_n-1,t}^n,  
\label{eqM2tilde}
\eea 
that is when the increments are replaced by variables of the form $\sigma_{\ti{i-L_n-1}} \Delta W_{j,t}^n$, where $\sigma_{\ti{i-L_n-1}}$ is the value of the volatility process at the beginning of the truncated sum. We show that we have $N_n^{-1/4}\l\{M_2^{(\sigma^2)}(T) -\tilde{M}_2^{(\sigma^2)}(T) \r\} \overset{\proba}{\rightarrow} 0$. We decompose 
\bea 
N_n^{-1/4}\l\{M_2^{(\sigma^2)} -\tilde{M}_2^{(\sigma^2)} \r\} = R_n^{(1)}+R_n^{(2)}+R_n^{(3)},
\eea 
with 
\bea 
R_n^{(1)} & = & N_n^{-1/4} \sum_{i=1}^{N_n}\sum_{1 \leq j < i - L_n}\frac{ \partial \omega^{i,j}}{\partial \sigma^2} \Delta X_{j,t}^n \Delta X_{i,t}^n,
\label{eqR1}\\
R_n^{(2)} & = & N_n^{-1/4} \sum_{i=1}^{N_n}\sum_{(i- L_n) \wedge 1 \leq j < i }\frac{ \partial \omega^{i,j}}{\partial \sigma^2} \Delta X_{j,t}^n (\Delta X_{i,t}^n - \beta_{i,i-L_n-1,t}^n),
\label{eqR2}\\
R_n^{(3)} & = & N_n^{-1/4} \sum_{i=1}^{N_n}\sum_{(i- L_n) \wedge 1 \leq j < i}\frac{ \partial \omega^{i,j}}{\partial \sigma^2} ( \Delta X_{j,t}^n - \beta_{j,i-L_n-1,t}^n ) \beta_{i,i-L_n-1,t}^n.
\label{eqR3}
\eea 
Now, proving that $R_n^{(1)}$ is negligible is immediate because when $|i-j| \geq L_n$, we have the domination $\frac{ \partial \omega^{i,j}}{\partial \sigma^2} \leq L \sqrt{N_n}e^{-N_n^\delta}$ for some $L >0$ so that by an easy application of Cauchy-Schwarz inequality and estimates from Lemma \ref{lemmaEstimateX} we get $\esp_\calu \l| R_n^{(1)} \r| \overset{\proba}{\rightarrow} 0$. Now we show the negligibility of $R_n^{(2)}$. Assume first that $\sigma$ has no jumps, i.e $\tilde{J} = 0$. $R_n^{(2)}$ being a sum of martingale increments, it is sufficient to show that 
$$N_n^{-1/2}\sum_{i=1}^{N_n}\esp_\calu[(A_i^n)^2\l(\Delta X_{i,t}^n - \beta_{i,i-L_n-1,t}^n\r)^2] \overset{\proba}{\rightarrow} 0,$$
where $A_i^n = \sum_{j=(i-L_n) \wedge 1}^{i-1}\frac{ \partial \omega^{i,j}}{\partial \sigma^2} \Delta X_{j,t}^n$. Introducing $v_{i,k,t} := \sigma_t-\sigma_{\ti{i-k-1}} $, $\delta_{i,k,t} := \int_{\ti{i-1} \wedge t}^{\ti{i} \wedge t}{v_{i,k,s}^2ds}$, we thus need to show that
$$N_n^{-1/2}\sum_{i=1}^{N_n}\esp_\calu[(A_i^n)^2\delta_{i,L_n,t}] \overset{\proba}{\rightarrow} 0.$$
It\^{o}'s formula applied to $v_{i,k,t}^2$ when $\tilde{J} = 0$ yields 
$$ v_{i,k,t}^2 = \underbrace{\int_{\ti{i-k-1}}^{t}{2v_{i,k,s}\tilde{\sigma}_s^{(1)}dW_s}}_{u_{i,k,t}^{(1)}}+\underbrace{\int_{\ti{i-k-1}}^{t}{2v_{i,k,s}\tilde{\sigma}_s^{(2)}d\tilde{W}_s}}_{u_{i,k,t}^{(2)}}+\underbrace{\int_{\ti{i-k-1}}^{t}{\l\{\l(\tilde{\sigma}_s^{(1)}\r)^2+\l(\tilde{\sigma}_s^{(2)}\r)^2\r\}ds}}_{u_{i,k,t}^{(3)}} ,$$
so that defining $\delta_{i,k,t}^{(l)} := \int_{\ti{i-1} \wedge t}^{\ti{i} \wedge t}{u_{i,k,s}^{(l)}ds}$ for $l\in \{1,2,3\}$, we now show
\bea 
N_n^{-1/2}\sum_{i=1}^{N_n}\esp_\calu [(A_i^n)^2\delta_{i,L_n,t}^{(l)}] \overset{\proba}{\rightarrow} 0.
\label{convDeltaL}
\eea 
For $l=3$, we have $|\delta_{i,L_n,t}^{(3)}| \leq L\Delta t_{i,t}^n (\ti{i}-\ti{i-L_n-1}) \leq L n^{-2+2\gamma}L_n$ by \textbf{(H)}, and thus (\ref{convDeltaL}) boils down to showing that  
\bea 
N_n^{-1/2}n^{-2+2\gamma}L_n\sum_{i=1}^{N_n}\esp_\calu[(A_i^n)^2] \overset{\proba}{\rightarrow} 0.
\eea 
Using $\esp_\calu[\Delta X_k^n \Delta X_j^n] = 0$ for $j\neq k$, we deduce 
\beas  
N_n^{-1/2}n^{-2+2\gamma}L_n\sum_{i=1}^{N_n}\esp_\calu[(A_i^n)^2] &=& N_n^{-1/2}n^{-2+2\gamma}L_n\sum_{i=1}^{N_n} \sum_{j=(i-L_n)\wedge1}^{i-1}{\l(\frac{ \partial \omega^{i,j}}{\partial \sigma^2}\r)^2 \esp_\calu \l[\l(\Delta X_{j}^n\r)^2\r]}\\
&\leq& L N_n^{-1/2}n^{-3+3\gamma}L_n\sum_{i=1}^{N_n} \sum_{j=(i-L_n)\wedge 1}^{i-1}{\l(\frac{ \partial \omega^{i,j}}{\partial \sigma^2}\r)^2}\\
&\leq& L N_n^{2} n^{-3+3\gamma}L_n \to 0, 
\eeas 
where we have used that by direct calculation we have $\sum_{i=1}^{N_n} \sum_{(i- L_n) \wedge 1 \leq j < i}  \l(\frac{ \partial \omega^{i,j}}{\partial \sigma^2}\r)^2 = \bop{N_n^{5/2}}$, and that $L_n = N_n^{1/2+\delta}$. For $l = 1$, we split (\ref{convDeltaL}) into two terms

\bea 
N_n^{-1/2}\sum_{i=1}^{N_n}\esp_\calu[(A_i^n)^2\delta_{i,L_n,t}^{(l)}] = P_n^{(1)}+ P_n^{(2)}, 
\label{decompositionP1P2}
\eea 
where 
\beas 
P_n^{(1)} &= & N_n^{-1/2}\esp_\calu\sum_{i=1}^{N_n}  \sum_{(i- L_n) \wedge 1 \leq j < i}  \l(\frac{ \partial \omega^{i,j}}{\partial \sigma^2}\r)^2 \l(\Delta X_{j}^n\r)^2  \delta_{i,L_n,t}^{(1)},\\
P_n^{(2)} & = & N_n^{-1/2}\esp_\calu\sum_{i=1}^{N_n} \sum_{(i- L_n) \wedge 1 \leq j \neq k < i}  \frac{ \partial \omega^{i,j}}{\partial \sigma^2}\frac{ \partial \omega^{i,k}}{\partial \sigma^2} \Delta X_{j}^n\Delta X_{k}^n \delta_{i,L_n,t}^{(1)}.
\eeas
We have by Cauchy-Schwarz inequality
\beas 
P_n^{(1)} &\leq&  N_n^{-1/2}\sum_{i=1}^{N_n} \sum_{(i- L_n) \wedge 1 \leq j < i}  \l(\frac{ \partial \omega^{i,j}}{\partial \sigma^2}\r)^2 \l(\esp_{\calu} \l[ \l(\Delta X_{j}^n\r)^4 \r] \esp_{\calu} [(\delta_{i,L_n,t}^{(1)})^2]\r)^{1/2},\\
&\leq& L N_n^{-1/2}n^{-3+3\gamma}L_n  \sum_{i=1}^{N_n} \sum_{(i- L_n) \wedge 1 \leq j < i}  \l(\frac{ \partial \omega^{i,j}}{\partial \sigma^2}\r)^2\\
&\leq& L N_n^{2}n^{-3+3\gamma}L_n \overset{\proba}{\rightarrow} 0,
\eeas
as $\esp_{\calu} \l[ \l(\Delta X_{j}^n\r)^4 \r] \leq L n^{-2+2\gamma}$ by (\ref{eqDeltaX}), and $$\esp_{\calu} \l[\l(\delta_{i,L_n,t}^{(1)}\r)^2\r] \leq 
L \Delta t_{i,t}^n (\ti{i}-\ti{i-L_n-1})\esp_\calu\l[\sup_{s \in [\ti{i-L_n-1},\ti{i}]} v_{i,L_n,s}^2\r] \leq L n^{-4+4\gamma}L_n^2,$$ 
by the same estimate as for (\ref{eqDeltaX}) for the It\^{o} semimartingale $v_{i,L_n,s}$. For $P_n^{(2)}\overset{\proba}{\rightarrow}0$, we first note that for $k < j$ we have
\beas 
\l|\esp_\calu \l[\Delta X_k^n \Delta X_j^n \delta_{i,L_n,t}^{(1)} \r]\r| &\leq&  \esp_\calu \l[|\Delta X_k^n|  \l|\int_{\ti{i-1}\wedge t}^{\ti{i}\wedge t} \esp \l[\l. \Delta X_j^n u_{i,L_n,s}^{(1)} \r| \calg_{j-1}^n\r]ds \r| \r] \\
&\leq& L \esp_\calu \l[|\Delta X_k^n|  \int_{\ti{i-1} \wedge t}^{\ti{i} \wedge t} \l|\esp \l[\l. \int_{\ti{j-1}}^{\ti{j}} v_{i,L_n,u} \tilde{\sigma}_u^{(1)} \sigma_u du \r| \calg_{j-1}^n\r]ds\r| \r] \\
&\leq& L n^{-3 + 3\gamma} L_n^{1/2},
\eeas
where the last step is obtained using \textbf{(H)} as for the previous estimates. Overall, we get
\beas 
P_n^{(2)} &\leq& L N_n^{-1/2}n^{-3 + 3\gamma} L_n^{1/2}  \sum_{i=1}^{N_n} \sum_{(i-L_n) \wedge 1 \leq j \neq k < i}  \frac{ \partial \omega^{i,j}}{\partial \sigma^2}\frac{ \partial \omega^{i,k}}{\partial \sigma^2}  \\
&\leq& L N_n^{-1/2}n^{-3 + 3\gamma} L_n^{3/2}  \sum_{i=1}^{N_n} \sum_{(i-L_n)\wedge 1 \leq j < i}  \l(\frac{ \partial \omega^{i,j}}{\partial \sigma^2}\r)^2 \\
&\leq& L N_n^{2}n^{-3 + 3\gamma} L_n^{3/2} \overset{\proba}{\rightarrow} 0.
\eeas 
Finally, when $l=2$, we write the same decomposition as (\ref{decompositionP1P2}), and we note that the exact same calculation as in the case $l=1$ for $P_n^{(1)}$ remains valid. Moreover, following closely the calculation above, we get $P_n^{(2)} = 0$ by orthogonality of the Brownian motions $W$ and $\tilde{W}$. When $\sigma$ has jumps of finite activity, we easily show as for previous calculations that an additional negligible term appears in $R_n^{(2)}$, and thus combining all those results we have $R_n^{(2)} \overset{\proba}{\rightarrow} 0$. Finally, $R_n^{(3)} \overset{\proba}{\rightarrow} 0$ is proven following the same line of reasoning as for $R_n^{(2)}$. \\
\medskip 

\textbf{Step 2.}
We are going to apply Theorem 2-1 p. 238 from \cite{jacod1997} to the continuous martingale  $N_n^{-1/4}\tilde{M}_2^{(\sigma^2)}$. Condition (2.8) is automatically satisfied with $B_t=0$. We now show the variance condition (2.9). This boils down to showing that there exists an increasing limit process $C_t$ such that for any $t \in [0,T]$

\bea 
\l\langle \tilde{M}_2^{(\sigma^2)},\tilde{M}_2^{(\sigma^2)} \r\rangle_t \overset{\proba}{\rightarrow} C_t,
\label{convBracketM2}
\eea  
and $C_T = \frac{5}{64T^{3/2} \sigma^7 a_0}\int_0^T{\alpha_s^{-1}ds}\int_0^T{\sigma_s^4 \alpha_sds}$. We introduce 
\beas  
L_n^{(1)} & := & N_n^{-1/2}\sum_{i=1}^{N_n}\sum_{ (i-L_n) \wedge 1 \leq j < i} {\l(\frac{ \partial \omega^{i,j}}{\partial \sigma^2}\r)^2 \sigma_{\ti{i-L_n-1}}^4\l(\Delta W_{j,t}^n\r)^2\Delta \ti{i,t}}, \\
L_n^{(2)} & := & N_n^{-1/2}\sum_{i=1}^{N_n} \sum_{ (i-L_n) \wedge 1 \leq j  \neq k < i} {\frac{ \partial \omega^{i,j}}{\partial \sigma^2}\frac{ \partial \omega^{i,k}}{\partial \sigma^2}\sigma_{\ti{i-L_n-1}}^4\Delta W_{j,t}^n\Delta W_{k,t}^n \Delta \ti{i,t}}. \\
\eeas  
 we have $\l\langle \tilde{M}_2^{(\sigma^2)},\tilde{M}_2^{(\sigma^2)} \r\rangle_t = L_n^{(1)} + L_n^{(2)}$, so that our strategy to show (\ref{convBracketM2}) will be to prove that 
\bea 
L_n^{(1)} & \overset{\proba}{\rightarrow} & C_t
\label{convLi1}\\
L_n^{(2)} & \overset{\proba}{\rightarrow} & 0.\label{convLi2}
\eea 
For $L_n^{(2)}$, we have directly that $\esp_{\calu}\l[\l(L_n^{(2)}\r)^2 \r]$ is equal to
\beas  
 N_n^{-1} \sum_{\underset{1 \leq i_1,i_2 \leq N_n}{ |i_1 - i_2| \leq L_n}}\sum_{\underset{j \neq k}{(i_1 \vee i_2) - L_n \leq j, k < (i_1 \wedge i_2)}}{\frac{ \partial \omega^{i_1,j}}{\partial \sigma^2}\frac{ \partial \omega^{i_2,j}}{\partial \sigma^2}\frac{ \partial \omega^{i_1,k}}{\partial \sigma^2}\frac{ \partial \omega^{i_2,k}}{\partial \sigma^2}} \esp_{\calu}\l[\sigma_{\ti{i_1-L_n-1}}^4\sigma_{\ti{i_2-L_n-1}}^4\r] \Delta t_{j,t}^n \Delta t_{k,t}^n \Delta t_{i_1,t}^n\Delta t_{i_2,t}^n,
\eeas 
where we have used that for $l < \textnormal{min}(j_1,j_2,k_1,k_2)$, we have $\esp[\Delta W_{t,j_1}^n\Delta W_{t,j_2}^n\Delta W_{t,k_1}^n\Delta W_{t,k_2}^n | \calg_{l}^n] = \Delta t_{j,t}^n\Delta t_{k,t}^n $ when $j_1=j_2=j$ and $k_1=k_2=k$, and the expectation is null otherwise. Now, using the boundedness of $\sigma$ and the fact that $\Delta t_{j,t}^n \leq L n^{-1+\gamma}$ by assumption \textbf{(H)}, we obtain 
\beas 
\esp_{\calu}\l[\l(L_n^{(2)}\r)^2 \r] \leq L N_n^{-1}n^{-4+4\gamma}\sum_{\underset{1 \leq i_1,i_2 \leq N_n}{ |i_1 - i_2| \leq L_n}}\sum_{\underset{j \neq k}{(i_1 \vee i_2) - L_n \leq j, k < (i_1 \wedge i_2)}}{\frac{ \partial \omega^{i_1,j}}{\partial \sigma^2}\frac{ \partial \omega^{i_2,j}}{\partial \sigma^2}\frac{ \partial \omega^{i_1,k}}{\partial \sigma^2}\frac{ \partial \omega^{i_2,k}}{\partial \sigma^2}},
\eeas 
which by direct calculation on the coefficients yields 
\beas 
\esp_{\calu}\l[\l(L_n^{(2)}\r)^2 \r] &\leq& L N_n^{-1} n^{-4+4\gamma} N_n^4 L_n \\
&\leq& N_n^{7/2 +\delta}n^{-4+4\gamma}\to 0,
\eeas 
for $\gamma$ and $\delta$ small enough. Now we turn to (\ref{convLi1}). We define $C_t := \frac{5}{64\sqrt T \sigma^7 a_0}\int_0^T{\alpha_s^{-1}ds}\int_0^t{\sigma_s^4 \alpha_sds} $, and we further decompose $L_n^{(1)} - C_t$ into 
\bea 
L_n^{(1)} - C_t = \sum_{i=1}^{6} B_n^{(i)},
\eea 
with 
\beas 
B_n^{(1)} & = & N_n^{-1/2} \sum_{i=1}^{N_n} \sum_{ (i-L_n) \wedge 1 \leq j < i} {\l(\frac{ \partial \omega^{i,j}}{\partial \sigma^2}\r)^2 \sigma_{\ti{i-L_n-1}}^4 \l(\l(\Delta W_{j,t}^n\r)^2 - \Delta t_{j,t}^n\r)\l(\ti{i}\wedge t - \ti{i-1}\wedge t\r)},\\
B_n^{(2)} & = &N_n^{-1/2} \Delta_n \sum_{i=1}^{N_n} \sum_{ (i-L_n) \wedge 1 \leq j < i} {\l(\frac{ \partial \omega^{i,j}}{\partial \sigma^2}\r)^2 \sigma_{\ti{i-L_n-1}}^4 \l(\alpha_{\ti{j-1}}-\alpha_{\ti{i-L_n-1}}\r)U_j^n\l(\ti{i}\wedge t - \ti{i-1}\wedge t\r)},\\
B_n^{(3)} & = & N_n^{-1/2} \Delta_n \sum_{i=1}^{N_n} \sum_{ (i-L_n) \wedge 1 \leq j < i} {\l(\frac{ \partial \omega^{i,j}}{\partial \sigma^2}\r)^2 \sigma_{\ti{i-L_n-1}}^4 \alpha_{\ti{i-L_n-1}} \l(U_j^n - 1 \r)\l(\ti{i}\wedge t - \ti{i-1}\wedge t\r)},\\
B_n^{(4)} & = & \sum_{i=1}^{N_n} \l\{N_n^{-1/2} \Delta_n\sum_{ (i-L_n) \wedge 1 \leq j < i} {\l(\frac{ \partial \omega^{i,j}}{\partial \sigma^2}\r)^2 - \frac{5}{64T^{3/2}\sigma^7a_0}\int_0^T{\alpha_s^{-1}ds}}\r\} \sigma_{\ti{i-L_n-1}}^4 \alpha_{\ti{i-L_n-1}} \l(\ti{i}\wedge t - \ti{i-1}\wedge t\r),\\
B_n^{(5)} & = & \frac{5}{64T^{3/2}\sigma^7a_0}\int_0^T{\alpha_s^{-1}ds}\sum_{i=1}^{N_n}{\l\{\sigma_{\ti{i-L_n-1}}^4 \alpha_{\ti{i-L_n-1}} - \sigma_{\ti{i}}^4 \alpha_{\ti{i}} \r\} \l(\ti{i}\wedge t - \ti{i-1}\wedge t\r)},\\
B_n^{(6)} & = & \frac{5}{64T^{3/2}\sigma^7a_0}\int_0^T{\alpha_s^{-1}ds}\l\{\sum_{i=1}^{N_n}{\sigma_{\ti{i}}^4 \alpha_{\ti{i}}  \l(\ti{i}\wedge t - \ti{i-1}\wedge t\r)} - \int_0^t{\sigma_s^4\alpha_sds}\r\}.
\eeas
Using that $\esp_{\calu}\l[\l(\l(\Delta W_{j,t}^n\r)^2 - \Delta t_{j,t}^n\r)\l(\l(\Delta W_{k,t}^n\r)^2 - \Delta t_{k,t}^n\r)\r] = 0$ if $j \neq k$, and $2\l(\Delta t_{j,t}^n\r)^2$ otherwise, we obtain the estimate 
\beas  
\esp_{\calu}\l[\l(B_n^{(1)}\r)^2\r] \leq L  N_n^{3}L_n n^{-4+4\gamma} \overset{\proba}{\rightarrow} 0.    
\eeas 
Moreover, by the same deviation inequality as (\ref{eqDeltaX}) for $\alpha$ (recall that $\alpha$ is an It\^{o} semimartingale) we have $\esp|\alpha_{\ti{j-1}}-\alpha_{\ti{i-L_n-1}}| \leq L n^{-1/2}L_n^{1/2}$ so that we obtain easily $\esp |B_n^{(2)}| \leq L N_n^2 n^{-5/2+2\gamma} L_n^{1/2} \to 0$. Similar computation to that of $B_n^{(1)}$ shows that $\esp\big[\big(B_n^{(3)}\big)^2\big] \to 0$ since $\esp [U_j^n-1] = 0$ and $\esp [(U_j^n-1)(U_i^n-1)] = 0$ when $i \neq j$. $B_n^{(4)} \overset{\proba}{\rightarrow} 0$ is a direct consequence of the fact that by a direct calculation we have uniformly in $i$ that $N_n^{-3/2} \Delta_n\sum_{ (i-L_n) \wedge 1 \leq j < i} {\l(\frac{ \partial \omega^{i,j}}{\partial \sigma^2}\r)^2} \overset{\proba}{\rightarrow} \frac{5}{64T^{3/2}\sigma^7a_0}$ and that $N_n\Delta_n \overset{\proba}{\rightarrow} \int_0^T{\alpha_s^{-1}ds}$ by (\ref{eqNumberJumps}), recalling that $\Delta_n = T/n$. $B_n^{(5)} \overset{\proba}{\rightarrow} 0$ is, again a simple consequence of the deviation inequality (\ref{eqDeltaX}) for the It\^{o} semimartingale $\sigma^4\alpha$, and finally $B_n^{(6)} \overset{\proba}{\rightarrow} 0$ is just the convergence of the Riemann sum toward the integral limit, and we are done. We show condition (2.10), i.e. that
\bea 
N_n^{-1/4} \langle \tilde{M}_2^{(\sigma^2)}, W \rangle_t \overset{\proba}{\rightarrow} 0.
\eea 
Note that 
\bea
N_n^{-1/4} \langle \tilde{M}_2^{(\sigma^2)}, W \rangle_t = N_n^{-1/4} \sum_{i=1}^{N_n} \sum_{ (i-L_n) \wedge 1 \leq j < i} \frac{ \partial \omega^{i,j}}{\partial \sigma^2} \sigma_{\ti{i-L_n-1}}^2 \Delta W_{j,t}^n \Delta t_{i,t}^n,
\eea 
so that by a straightforward calculation on the Brownian motion increments we have \bea 
\esp_\calu \l[\l(\langle \tilde{M}_2^{(\sigma^2)}, W \rangle_t\r)^2\r] \leq L N_n^2 L_n n^{-3+3\gamma} \overset{\proba}{\rightarrow} 0.
\eea 
Moreover, condition (2.11) is satisfied because $\tilde{M}_2^{(\sigma^2)}$ is continuous. Finally we show condition (2.12). But note that for any bounded martingale $\overline{N}$ orthogonal to $W$ we have directly 
\bea 
\langle \tilde{M}_2^{(\sigma^2)}, \overline{N} \rangle_t = 0
\eea 
by (\ref{eqM2tilde}), so that all the conditions required for the theorem hold.
\smallskip

\textbf{Step 3.} In the presence of jumps, for $n$ large enough, an additional term appears in $M_2^{(\sigma^2)}(T)$. First, since $J$ is of finite activity and by the Grigelionis decomposition for It\^{o}-semimartingales (see e.g. Theorem 2.1.2 in \cite{jacod2011discretization}), we can assume without loss of generality that the jump times of $J$ are a subset of the support of a Poisson random measure $\mu$ on $\reels_+ \times E$ for $E$ some arbitrary Polish space, adapted to $\calf_t$, and with finite intensity measure $\nu$. Let thus $\tau_1$, $\cdots$ ,$\tau_p$, $\cdots$ be an exhausting sequence for the jumps of $\mu$. Since $J$ is of finite activity, for $n$ sufficiently large we cannot have more than a single jump on intervals of the form $[\ti{i-L_n}, \ti{i}]$ because $\sup_{ L_n < i \leq N_n } \ti{i} - \ti{i-L_n} \to^{a.s} 0$ by assumption \textbf{(H)}. Therefore, if $n$ is large enough, after a simple rearrangement of the terms that contain jumps, and by the previous calculation in the continuous case, we can write $M_{2}^{(\sigma^2)}(t)$ under the form 
\bea  
 M_{2}^{(\sigma^2)}(t) = \tilde{M}_2^{(\sigma^2)}(t) + A_n^+(t) + A_n^-(t) + o_\proba(1),  
\label{eqM2J}
\eea 
with 
\bea 
A_n^+(t) = \sum_{p \geq 1}{\Delta J_{\tau_p}\sum_{j = i_p + 1}^{i_p + L_n} \frac{ \partial \omega^{i_p,j}}{\partial \sigma^2} \Delta \tilde{X}_{j,t}^n } \textnormal{ and } A_n^-(t) = \sum_{p \geq 1}{\Delta J_{\tau_p}\sum_{j = i_p - L_n }^{i_p - 1 } \frac{ \partial \omega^{i_p,j}}{\partial \sigma^2} \Delta \tilde{X}_{j,t}^n },
\eea 
where $i_p$ is such that $\ti{i_p-1} < \tau_p \leq \ti{i_p }$. We define

\bea 
M_n^+(t,p) =  \sum_{j = i_p + 1}^{i_p + L_n} \frac{ \partial \omega^{i_p,j}}{\partial \sigma^2} \Delta \tilde{X}_{j,t}^n  \textnormal{ and }  M_n^-(t,p) =  \sum_{j = i_p - L_n }^{i_p - 1 } \frac{ \partial \omega^{i_p,j}}{\partial \sigma^2} \Delta \tilde{X}_{j,t}^n,
\eea 
along with the following infinite dimensional vector $(G, (R_\infty^+(p), R_\infty^-(p))_{p \geq 0})$ such that $G$, $R_\infty^+(p)$ and $ R_\infty^-(p)$ are i.i.d standard normal random variables. We can assume that $\Omega$ and $\calg_T$ are rich enough to include such random variables information without loss of generality, since we can always construct a very good filtered extension as explained in pp. 36-37 of \cite{jacod2011discretization}. Now define 
\bea 
V_\infty := \frac{5}{64T^{3/2}\sigma^7a_0}\int_0^T{\alpha_s^{-1}ds},
\eea 
\beas 
M_\infty^+(p) :=  \sigma_{\tau_p}\alpha_{\tau_p}^{1/2}V_\infty^{1/2}R_\infty^+(p) \textnormal{ and } M_\infty^-(p) := \sigma_{\tau_p-}\alpha_{\tau_p-}^{1/2}V_\infty^{1/2} R_\infty^-(p),
\eeas
and 
\beas 
\tilde{G} := C_T^{1/2} G,
\eeas 
where $C_T$ was defined in (\ref{convBracketM2}). We are going to show that $\calg_T$-stably in law, we have the convergence 
\bea 
N_n^{-1/4} (\tilde{M}_{2}^{(\sigma^2)}(T), (M_n^+(T,p),M_n^-(T,p))_{p \geq 1}) \to (\tilde{G}, (M_\infty^+(p),M_\infty^-(p))_{p \geq 1}).
\label{vectorStable}
\eea 
 As the subset of finite dimensional cylinders is a convergence determining class for the product topology of $\reels^{\naturels}$, it is sufficient to show that the above convergence holds for all finite families of the form $(\tilde{M}_{2}^{(\sigma^2)}(T),M_n^+(T,p_1),M_n^-(T,p_1), \cdots,M_n^+(T,p_k),M_n^-(T,p_k))$, $k \geq 1$. Now, let us consider the filtration $\tilde{\calg}_t$ which is the smallest filtration containing $\calg_t$ and the jump times of $\mu$, $(\tau_p)_{p \geq 1}$. By independence of $\mu$ and the Wiener process $W$, $\tilde{X}$ is also a continuous It\^{o} process with respect to the filtration $\tilde{\calg}_t$, so that $(\tilde{M}_{2}^{(\sigma^2)}(t),M_n^+(t,p_1),M_n^-(t,p_1),\cdots,M_n^+(t,p_k),M_n^-(t,p_k))_{t \in [0,T]}$ is a multi-dimensional continuous $\tilde{\calg}_t$-martingale. Now, for $n$ large enough and by the finite activity property, we have for any $1 \leq i \neq j \leq k $,
\beas
\langle M_n^+(.,p_i), M_n^+(.,p_j) \rangle_t = \langle M_n^-(.,p_i), M_n^-(.,p_j) \rangle_t = 0 \textnormal{ a.s,}
\eeas 
and 
\beas 
\langle M_n^+(.,p_i), M_n^-(.,p_i) \rangle_t = 0 \textnormal{ a.s.}
\eeas 
Moreover
\beas 
N_n^{-1/2}\langle M_n^+(.,p_i),M_n^+(.,p_i)\rangle_t = N_n^{-1/2} \sum_{j = i_p +1}^{i_p + L_n} \l(\frac{ \partial \omega^{i_p,j}}{\partial \sigma^2}\r)^2 \int_{\ti{j-1} \wedge t}^{\ti{j} \wedge t}{\sigma_s^2ds},
\eeas
since the random index $i_p$ is $\tilde{\calg}_0$-measurable. By a similar (but easier) calculation than for $L_n^{(1)}$ above, we have 
\beas 
N_n^{-1/2}\langle M_n^+(.,p_i),M_n^+(.,p_i)\rangle_T \overset{\proba}{\rightarrow} \frac{5}{64T^{3/2}\sigma^7a_0} \sigma_{\tau_{p_i}}^2\alpha_{\tau_{p_i}} \int_0^T{\alpha_s^{-1}ds},
\eeas 
and also 
\beas 
N_n^{-1/2}\langle M_n^-(.,p_i),M_n^-(.,p_i)\rangle_T \overset{\proba}{\rightarrow} \frac{5}{64T^{3/2}\sigma^7a_0} \sigma_{\tau_{p_i}-}^2\alpha_{\tau_{p_i}-} \int_0^T{\alpha_s^{-1}ds}.
\eeas
Finally we show the negligibility of $N_n^{-1/2} \langle \tilde{M}_2^{(\sigma^2)}, M_n^+(.,p_i) \rangle_t$ and $N_n^{-1/2} \langle \tilde{M}_2^{(\sigma^2)}, M_n^-(.,p_i) \rangle_t$. We have 

\bea 
N_n^{-1/2} \langle \tilde{M}_2^{(\sigma^2)} ,M_n^+(.,p_i)\rangle_t = N_n^{-1/2} \sum_{j=i_p + 1}^{i_p + L_n} \frac{ \partial \omega^{i_p,j}}{\partial \sigma^2} \sum_{k=(j-L_n) \wedge 1 }^{j-1} \frac{ \partial \omega^{j,k}}{\partial \sigma^2} \Delta \tilde{X}_{k,t}^n \sigma_{j-L_n-1}^2 \Delta t_{j,t}^n,
\eea 
so that by Assumption \textbf{(H)} we have $N_n^{-1} \esp_\calu \l[ \langle \tilde{M}_2^{(\sigma^2)} ,M_n^+(.,p_i)\rangle_t^2 \r]$ bounded by
\beas 
 & & L N_n^{-1} n^{-2+2\gamma} \sum_{j_1,j_2 = i_p + 1}^{i_p + L_n}  \frac{ \partial \omega^{i_p,j_1}}{\partial \sigma^2} \frac{ \partial \omega^{i_p,j_2}}{\partial \sigma^2} \sum_{k = (j_1 \vee j_2) - L_n \wedge 1 }^{j_1 \wedge j_2 -1}  \frac{ \partial \omega^{j_1,k}}{\partial \sigma^2}\frac{ \partial \omega^{j_2,k}}{\partial \sigma^2} \esp_\calu \l[ \l(\Delta \tilde{X}_{k,t}^n\r)^2\r]  \\
&\leq& L N_n^{-1} n^{-3+3\gamma}\sum_{j_1,j_2 = i_p + 1}^{i_p + L_n}  \frac{ \partial \omega^{i_p,j_1}}{\partial \sigma^2} \frac{ \partial \omega^{i_p,j_2}}{\partial \sigma^2} \sum_{k = (j_1 \vee j_2) - L_n \wedge 1 }^{j_1 \wedge j_2 -1}  \frac{ \partial \omega^{j_1,k}}{\partial \sigma^2}\frac{ \partial \omega^{j_2,k}}{\partial \sigma^2}\\
&\leq& L N_n^2 L_n n^{-3+3\gamma} \overset{\proba}{\rightarrow} 0, 
\eeas  
and thus the bracket is negligible. By a similar calculation we get that the bracket involving $M_n^-(p_i,.)$ is also negligible. Moreover, the convergence of $\langle \tilde{M}_2^{(\sigma^2)} , \tilde{M}_2^{(\sigma^2)}\rangle_t$ was shown in (\ref{convBracketM2}). Finally, as above we easily check the bracket of each martingale with either $W$ or a bounded martingale orthogonal to $W$ is negligible so that by another application of Theorem 2-1 in \cite{jacod1997} we have (\ref{vectorStable}). From the representation 
\bea 
N_n^{-1/4}M_2^{(\sigma^2)}(T) = N_n^{-1/4} \l\{\tilde{M}_2^{(\sigma^2)}(T) + \sum_{p \geq 1} \Delta J_{\tau_p} (M_n^+(T,p)+M_n^+(T,p))  \r\} + o_\proba(1),
\eea 
along with the fact that $\{ p | \Delta J_{\tau_p} \neq 0\}$ is finite, we deduce by the stable convergence (\ref{vectorStable}) that $\tilde{\calg}$ (and \textit{a fortiori} $\calg$) stably in law 
\bea 
N_n^{-1/4}M_2^{(\sigma^2)}(T) \to \tilde{G} + \sum_{p \geq 1} \Delta J_{\tau_p} (M_\infty^+(p)+M_\infty^-(p)),
\eea
which is equal to the claimed distribution.
\medskip 

Finally, to show the convergence $N_n^{-1/2} M_{2}^{(a^2)}(T) \overset{\proba}{\rightarrow} 0$, note that $\frac{ \partial \omega^{i,j}}{\partial \sigma^2}$ and $\frac{ \partial \omega^{i,j}}{\partial a^2}$ are equivalent up to a constant term so that all the above computations apply, and thus the scaling in $N_n^{-1/2}$ instead of $N_n^{-1/4}$ yields the negligibility of this term.

\end{proof}

Before turning to the limiting distribution of the other terms, we recall that for a $\sigma$-field $\calh$, a random vector $Z$ and a sequence of random vectors $Z_n$ in $\reels^b$ , we say that $Z_n$ converges in law toward $Z$ conditioned on $\calh$ if we have for any $u \in \reels^b$ 

\bea 
\esp \l[\l.e^{iu^TZ_n} \r| \calh \r] \overset{\proba}{\rightarrow} \esp \l[\l.e^{iu^T Z} \r| \calh \r].
\eea 
Moreover, we recall in the following proposition a key result to combine stable convergence and conditional convergence. The proof of the result can be consulted in \cite{barndorff2008designing} (proof of Proposition 5 on p. 1524).

\begin{proposition*} \label{propositionConditionalStable}
Let $\calh$ be a given sub-$\sigma$-field, and let $(Y_n)$ and $(Z_n)$ be sequences of random vectors, such that each $Y_n$ is $\calh$-measurable and the sequence converges $\calh$-stably toward a limiting distribution $Y$, and $(Z_n)$ converges in law conditioned on $\calh$ to some $Z$. Then $(Y_n,Z_n) \to (Y,Z)$ $\calh$-stably in distribution. 
\end{proposition*}

\begin{lemma*} \label{lemmaM3}
We have conditioned on $\calg_T$ the convergence in distribution  
\bea 
N_n^{-1/4}  M_{3}^{(\sigma^2)}(T) \to \calm \caln \l( 0, \frac{\sqrt{T}\overline{\sigma}_0^2}{8\sigma^5a_0} \r), 
\eea 
and
\bea 
N_n^{-1/2} M_{3}^{(a^2)}(T) \overset{\proba}{\rightarrow} 0,
\label{eqM3a}
\eea 
where we recall the definition $\overline{\sigma}_0^2 = T^{-1}\l\{\int_0^T\sigma_s^2ds + \sum_{0\leq s\leq T}\Delta J_s^2\r\}$.
\end{lemma*}

\begin{proof} 
We start with $M_{3}^{(\sigma^2)}(T)$. We apply a conditional version of Theorem 5.12 from \cite{KallenbergFoundation2002}(p. 92). Accordingly, we note that $M_{3}^{(\sigma^2)}(T)$ can be written as \bea 
N_n^{-1/4}M_{3}^{(\sigma^2)}(T) = \sum_{i=0}^{N_n} \tilde{\chi}_i^n,
\eea 
where $\tilde{\chi}_i^n = -2N_n^{-1/4}\l\{\sum_{j = 1}^{N_n} \frac{ \partial \dot{\omega}^{i,j}}{\partial \sigma^2} \Delta X_j^n\r\} \epsilon_{t_i^n}$, are rowwise conditionally independent and centered given $\calg_T$. To get the theorem, it is thus sufficient to show that 
\bea 
\sum_{i=1}^{N_n} \esp\l[\l. \l(\tilde{\chi}_i^n\r)^2 \r|\calg_T\r] \overset{\proba}{\rightarrow} \frac{1}{8\sqrt{T}\sigma^5a_0}\l\{\int_0^T\sigma_s^2ds + \sum_{0\leq s\leq T}\Delta J_s^2\r\},
\label{convVarianceM3}
\eea
and the Lindeberg condition, for any $\epsilon >0$,
\bea 
\sum_{i=0}^{N_n} \esp\l[\l. \l(\tilde{\chi}_i^n\r)^2 \mathbb{1}_{\{|\tilde{\chi}_i^n| \geq \epsilon \}} \r|\calg_T\r] \overset{\proba}{\rightarrow} 0.
\eea
For (\ref{convVarianceM3}), we can write $\sum_{i=0}^{N_n} \esp\l[\l. \l(\tilde{\chi}_i^n\r)^2 \r|\calg_T\r] = T_n^{(1)} + T_n^{(2)}$ with 
\bea 
T_n^{(1)} = 4a_0^2 N_n^{-1/2} \sum_{i=0}^{N_n}\sum_{j=1}^{N_n} \l(\frac{ \partial \dot{\omega}^{i,j}}{\partial \sigma^2}\r)^2\l(\Delta X_j^n\r)^2,
\eea 
and
\bea 
 T_n^{(2)} =4a_0^2 N_n^{-1/2} \sum_{i=0}^{N_n}\sum_{j\neq k=1}^{N_n} \frac{ \partial \dot{\omega}^{i,j}}{\partial \sigma^2}\frac{ \partial \dot{\omega}^{i,k}}{\partial \sigma^2}\Delta X_j^n\Delta X_k^n,
\eea 
and using same techniques as for the proof of Lemma \ref{lemmaM2} we easily get by direct calculation on the coefficients $\frac{ \partial \dot{\omega}^{i,j}}{\partial \sigma^2}$ that we have $T_n^{(1)} \overset{\proba}{\rightarrow} \frac{1}{8\sqrt{T}\sigma^5a_0}\l\{\int_0^T\sigma_s^2ds + \sum_{0\leq s\leq T}\Delta J_s^2\r\}$, and $T_n^{(2)} \overset{\proba}{\rightarrow} 0$. As for the Lindeberg condition, it is sufficient to notice that by independence of the Brownian increments and similar computation we have $\sum_{i=0}^{N_n} \esp\big[ \l(\tilde{\chi}_i^n\r)^4 \big| \calg_T\big] \overset{\proba}{\rightarrow} 0$. Finally, for $M_{3}^{(a^2)}(T)$, all the previous calculation holds but now the scaling in $N_n^{-1/2}$ implies that $N_n^{-1/2}M_{3}^{(a^2)}(T) \overset{\proba}{\rightarrow} 0$.
\end{proof}

\begin{lemma*} \label{lemmaM4}
We have conditioned on $\calg_T$ the convergence in distribution  
\bea 
\Phi_n^{-1/2} M_4(T) \to  \caln \l(0,\l(\begin{matrix} \frac{\sqrt{T}}{16a_0\sigma^3} &0\\ 0& \frac{2}{a_0^4}+ \frac{\textnormal{cum}_4[\epsilon]}{a_0^8}\end{matrix}\r)\r). 
\eea 

\end{lemma*}

\begin{proof}

This is an immediate adaptation of (45) and (47) pp.248-249 in \cite{xiu2010quasi} conditioned on $\calg_T$ in lieu of $\sigma(X)$, since $\epsilon$ is independent of $\calg_T$.
\end{proof}

We consider now the general case $B \geq 1$, and accordingly we define for $i \in \{1, \cdots,B\}$ the local QMLE $\widehat{\xi}_{n,(i)} = (\widehat{\sigma}_{n,(i)}^2,\widehat{a}_{n,(i)}^2)$, and $\Psi_{n,(i)}$, $\bar{\Psi}_{n,(i)}$ the score functions on the block $i$ where all quantities are taken in the time interval $(\Tau_{i-1}, \Tau_{i}]$. We also introduce the notation 
$$\boldsymbol{\widehat{\xi}}_n := (\widehat{\sigma}_{n,(1)}^2,\widehat{a}_{n,(1)}^2,\cdots,\widehat{\sigma}_{n,(B)}^2,\widehat{a}_{n,(B)}^2),$$ $\boldsymbol{\Psi}_n := (\Psi_{n,(1)},\cdots,\Psi_{n,(B)})$, and $\boldsymbol{\bar{\Psi}}_n := (\bar{\Psi}_{n,(1)},\cdots,\bar{\Psi}_{n,(B)})$. The next lemma states the limit distribution of the vector $\boldsymbol{\Psi}_n - \boldsymbol{\bar{\Psi}}_n$. Finally we introduce the scaling factors $N_{n,(i)}:=N_n\l(\Tau_i\r)- N_n\l(\Tau_{i-1}\r)$ along with the global scaling matrix $\boldsymbol{\Phi}_n = \textnormal{diag}(N_{n,(1)}^{1/2},N_{n,(1)}, \cdots ,N_{n,(B)}^{1/2},N_{n,(B)}) \in \reels^{2B \times 2B}$.

\begin{lemma*} \label{lemmaCLTPsi}
We have for any $\boldsymbol{\sigma}^2 := (\sigma_{(1)}^2,\cdots,\sigma_{(B)}^2) \in [\underline{\sigma}^2, \overline{\sigma}^2]^{B}$, taking $\boldsymbol{\xi} := (\sigma_{(1)}^2,a_0^2, \cdots,\sigma_{(B)}^2,a_0^2)$, stably in $\calg_T$, the convergence in distribution

\beas
\boldsymbol{\Phi}_n^{1/2}\l\{\boldsymbol{\Psi}_n(\boldsymbol{\xi}) - \boldsymbol{\bar{\Psi}}_n(\boldsymbol{\xi})\r\} \to  \calm\caln\l(0,\l( \begin{matrix}  V_{(1)} & 0 & \cdots & 0 \\ 0 & V_{(2)} & 0 & \vdots \\ \vdots &0&\ddots&\vdots\\0&\cdots&\cdots&V_{(B)}\end{matrix} \r)\r),
\eeas 
where for $i \in \{1, \cdots ,B\}$, $V_{(i)}$ is the two dimensional matrix defined by
\beas  
V_{(i)} := \l( \begin{matrix}  \inv{4a_0}\l(\frac{5\calq_{(i)}}{16\sigma_{(i)}^7\Delta_B^{1/2}}+ \frac{\bar{\sigma}_{i}^2 \sqrt {\Delta_B}}{8\sigma_{(i)}^5} + \frac{\sqrt {\Delta_B}}{16 \sigma_{(i)}^3}\r) & 0 \\ 0 & \inv{2a_0^{4}} + \frac{\textnormal{cum}_4[\epsilon]}{4a_0^8} \end{matrix} \r),
\eeas 
with
$$\Delta_B\bar{\sigma}_{i}^2 := \int_{\Tau_{i-1}}^{\Tau_i } \sigma_s^2ds + \sum_{\Tau_{i-1} < s \leq \Tau_i}\Delta J_s^2,$$
and we recall that
$$\calq_{(i)} = \qtermlocal{\Tau_{i-1}}{\Tau_i }.$$
\end{lemma*}

\begin{proof}
First, for $i \in \{1, \cdots,B\}$, we define the processes $M_{1,(i)},\cdots, M_{4,(i)}$ following the definitions  (\ref{eqM1})-(\ref{eqM4}) adapted to the time interval $(\Tau_{i-1},\Tau_i]$ of length $\Delta_B$. Accordingly, for $k \in \{1, \cdots ,4\}$, we denote by $\boldsymbol{M}_{k}$ the vector process $(M_{k,(1)}^{(\sigma^2)},M_{k,(1)}^{(a^2)},\cdots,M_{k,(B)}^{(\sigma^2)},M_{k,(B)}^{(a^2)})$, and we note that we have the decomposition
 
\beas 
2\boldsymbol{\Phi}_n^{1/2}\l\{\boldsymbol{\Psi}_n(\xi) - \boldsymbol{\bar{\Psi}}_n(\xi)\r\} = \boldsymbol{\Phi}_n^{-1/2} \l\{\boldsymbol{M}_1(T)+2\boldsymbol{M}_2(T)+\boldsymbol{M}_3(T)+\boldsymbol{M}_4(T) \r\}.
\eeas
For $i\in \{1,\cdots,B\}$, we consider the two terms $M_{3,(i)}(T)$ and $M_{4,(i)}(T)$. By independence of $\epsilon$ with the other processes we deduce that the conditional covariance term between those two processes is null. We use this fact along with the marginal convergences obtained in Lemma \ref{lemmaM3} and Lemma \ref{lemmaM4} to obtain the convergence in law conditioned on $\calg_T$

\beas 
\Phi_{n,(i)}^{-1/2} \l\{M_{3,(i)}(T)+M_{4,(i)}(T) \r\} \to \calm\caln\l(0,\l( \begin{matrix}  \inv{a_0}\l(\frac{\bar{\sigma}_{i}^2 \sqrt {\Delta_B}}{8\sigma_{(i)}^5} + \frac{\sqrt {\Delta_B}}{16 \sigma_{(i)}^3}\r) & 0 \\ 
0 & \frac{2}{a_0^{4}} + \frac{\textnormal{cum}_4[\epsilon]}{a_0^8} \end{matrix} \r)\r), 
\eeas 
where $\Phi_{n,(i)} := \textnormal{diag}(N_{n,(i)}^{1/2},N_{n,(i)})$. Now, by Slutsky's lemma, Lemma \ref{lemmaM1} and Lemma \ref{lemmaM2} we also have the $\calg_T$-stable convergence in distribution  
$$\Phi_{n,(i)}^{-1/2}\l\{M_{1,(i)}(T) + 2M_{2,(i)}(T)\r\}\to \calm\caln\l(0,\l( \begin{matrix}  \frac{5\calq_{(i)}}{16a_0\sigma_{(i)}^7\Delta_B^{3/2}} & 0 \\ 
0 & 0 \end{matrix} \r)\r). $$
Finally, by application of Proposition \ref{propositionConditionalStable} with sub-$\sigma$-field $\calg_T$ since $M_{1,(i)}(T) + 2M_{2,(i)}(T)$ is $\calg_T$-measurable, we deduce the joint $\calg_T$-stable convergence of 
$$\Phi_{n,(i)}^{-1/2}\l(M_{1,(i)}(T) + 2M_{2,(i)}(T), M_{3,(i)}(T)+M_{4,(i)}(T) \r),$$ 
hence the convergence of $\Phi_{n,(i)}^{-1/2}\l(M_{1,(i)}(T) + 2M_{2,(i)}(T) + M_{3,(i)}(T)+M_{4,(i)}(T) \r)$  toward a mixed normal distribution of random variance $4V_{(i)}$. Finally, as blocks are non overlapping, we deduce that for any $k,l \in \{1, \cdots,4\}$, for any $i \neq j \in \{1, \cdots ,B\}$ the martingales $M_{k,(i)}$ and $M_{l,(j)}$ are orthogonal so that we have automatically the joint convergence of $\boldsymbol{\Phi}_n^{1/2}\l\{\boldsymbol{\Psi}_n(\xi) - \boldsymbol{\bar{\Psi}}_n(\xi)\r\}$ to a mixed normal with block diagonal random variance matrix whose submatrices are $V_{(1)},\cdots,V_{(B)}$, and we are done.  

\end{proof}

Finally, we derive a central limit theorem for $\boldsymbol{\widehat{\xi}}_n$ to the limit $\boldsymbol{\xi}_0 := (\bar{\sigma}_{1}^2,a_0^2,\cdots,\bar{\sigma}_{B}^2,a_0^2)$, and as a byproduct Theorem \ref{CLTQMLEjumps} (and Theorem \ref{CLTQMLE}).

\begin{theorem*}
We have $\calg_T$-stably in law that
\beas
\boldsymbol{\Phi}_n^{1/2}\l\{\boldsymbol{\widehat{\xi}}_n-\boldsymbol{\xi}_0 \r\}\to  \calm\caln\l(0,\l( \begin{matrix}  V_{(1)}^{'} & 0 & \cdots & 0 \\ 0 & V_{(2)}^{'} & 0 & \vdots \\ \vdots &0&\ddots&\vdots\\0&\cdots&\cdots&V_{(B)}^{'}\end{matrix} \r)\r),
\eeas 
where for $i \in \{1, \cdots,B\}$, $V_{(i)}^{'}$ is the two dimensional matrix defined by
\beas  
V_{(i)}^{'} := 
\l( \begin{matrix}  a_0\l(\frac{5\calq_{(i)}}{\bar{\sigma}_{i} \Delta_B^{3/2}} + \frac{3\bar{\sigma}_{i}^3}{\sqrt {\Delta_B}}\r) & 0 \\ 
0 & 2a_0^{4} + \textnormal{cum}_4[\epsilon]  \\
\end{matrix} \r).
\eeas 
In particular, Theorem \ref{CLTQMLEjumps} (and Theorem \ref{CLTQMLE}) hold.

\end{theorem*}

\begin{proof}
First, note that we can easily extend Lemma \ref{lemmaCLTPsi} to get a central limit theorem at the point $\boldsymbol{\xi}_0 = (\bar{\sigma}_{1}^2,a_0^2,\cdots,\bar{\sigma}_{B}^2,a_0^2)$ for $\boldsymbol{\Phi}_n^{1/2}\l\{\boldsymbol{\Psi}_n(\boldsymbol{\xi}_0) - \boldsymbol{\bar{\Psi}}_n(\boldsymbol{\xi}_0)\r\}$ by a generalization of Slutsky's Lemma for stably convergent sequences (see e.g. Theorem 3.18 (b) in \cite{hausler2015stable}), where now the submatrices $V_{(i)}$ in the asymptotic variance of the mixed normal distribution have the form

\beas
 V_{(i)} = \l( \begin{matrix}  \inv{64a_0}\l(\frac{5\calq_{(i)}}{\bar{\sigma}_{i}^7 \Delta_B^{1/2}}+ \frac{ 3\sqrt {\Delta_B}}{\bar{\sigma}_{i}^3}\r) & 0 \\ 
0 & \inv{2a_0^{4}} + \frac{\textnormal{cum}_4[\epsilon]}{4a_0^8}  \\
\end{matrix} \r).
\eeas 
To derive the CLT for the $2B$-dimensional estimator $\boldsymbol{\widehat{\xi}}_n$, we follow the standard procedure and expand the score function around $\boldsymbol{\xi}_0$. Thus, starting from the first order conditions on the score functions, we have

\bea 
0 = \boldsymbol{\Psi}_n\l(\boldsymbol{\widehat{\xi}}_n \r) = \boldsymbol{\Psi}_n(\boldsymbol{\xi}_0) + \boldsymbol{\Phi}_n^{-1/2} \boldsymbol{H}_n(\zeta_n)\boldsymbol{\Phi}_n^{1/2}\left(\boldsymbol{\widehat{\xi}}_n - \boldsymbol{\xi}_0\r),
\label{eqTaylorScore}
\eea 
for some $\zeta_n \in \big[\boldsymbol{\xi}_0, \boldsymbol{\widehat{\xi}}_n\big]$, and where $\boldsymbol{H}_n$ is the block diagonal matrix with submatrices $H_{n,(1)},\cdots, H_{n,(B)}$, and for $i \in \{1, \cdots,B\}$, $H_{n,(i)}$ is the scaled Hessian matrix of the log-likelihood field on block $i$, defined as in (\ref{defFisherMatrix}) adapted to the time interval $(\Tau_{i-1},\Tau_i]$. In the same way, we define $\boldsymbol{\Gamma}(\boldsymbol{\xi_0})$ as the block diagonal matrix whose subcomponents are $\Gamma_{(i)}(\xi_{0,(i)})$ where 
\beas 
\Gamma_{(i)}(\xi_{0,(i)}) := \l( \begin{matrix}  \frac{\sqrt {\Delta_B}}{8a_0 \bar{\sigma}_{i}^3}  & 0\\ 
0 & \inv{2a_0^4} \end{matrix} \r),
\eeas 
and $\xi_{0,(i)} := (\bar{\sigma}_{i}^2,a_0^2)$. We can rewrite (\ref{eqTaylorScore}) as 
\bea 
 \boldsymbol{\Gamma}(\boldsymbol{\xi}_0)^{-1}\boldsymbol{H}_n(\zeta_n)\boldsymbol{\Phi}_n^{1/2}\left(\boldsymbol{\widehat{\xi}}_n - \boldsymbol{\xi}_0\r) = - \boldsymbol{\Gamma}(\boldsymbol{\xi}_0)^{-1}  \boldsymbol{\Phi}_n^{1/2}\l\{\boldsymbol{\Psi}_n\l(\boldsymbol{\xi}_0\r) -\boldsymbol{\bar{\Psi}}_n\l(\boldsymbol{\xi}_0 \r) \r\} + \boldsymbol{\Phi}_n^{1/2} \boldsymbol{\bar{\Psi}}_n(\boldsymbol{\xi}_0).
 \label{eqTaylorReforumulate}
\eea 
Note that, again, by a direct adaptation of (38) and (40) in \cite{xiu2010quasi} (pp. 247-248) to the case of an irregular grid with $\pi_T^n \overset{\proba}{\rightarrow} 0$ and on the interval $(\Tau_{i-1},\Tau_i]$ we automatically get that each $\Phi_{n,(i)}^{1/2} \bar{\Psi}_{n,(i)}(\xi_{0,(i)}) = o_\proba(1)$ so that $\boldsymbol{\Phi}_n^{1/2} \boldsymbol{\bar{\Psi}}_n(\boldsymbol{\xi}_0)$ is negligible. Now, $\boldsymbol{\widehat{\xi}}_n$ is consistent by application of Theorem \ref{thmConsistencyH1} to each $\widehat{\xi}_{n,(i)}$ on block $i$. Therefore, $\zeta_n \overset{\proba}{\rightarrow} \boldsymbol{\xi}_0$, and by virtue of Lemma \ref{lemmaFisherConsistency} applied to each submatrix $H_{n,(i)}$, we conclude on the one hand that $\boldsymbol{\Gamma}(\boldsymbol{\xi}_0)^{-1}\boldsymbol{H}_n(\zeta_n) \overset{\proba}{\rightarrow} \I$ where $\I \in \reels^{2B\times2B}$ is the identity matrix, and on the other hand by Slutsky's Lemma and the stable CLT for $\boldsymbol{\Phi}_n^{1/2}\l\{\boldsymbol{\Psi}_n\l(\boldsymbol{\xi}_0\r) -\boldsymbol{\bar{\Psi}}_n\l(\boldsymbol{\xi}_0 \r) \r\}$ that the left-hand side of (\ref{eqTaylorReforumulate}) tends $\calg_T$-stably in law to a mixed normal distribution of block diagonal random variance matrix with submatrices of the form

\beas 
\l( \begin{matrix}  \frac{64a_0^2 \bar{\sigma}_{i}^6}{\Delta_B} &  0\\ 
0 & 4a_0^8  \\\end{matrix} \r) &\times& \l( \begin{matrix}  \inv{64a_0}\l(\frac{5\calq_{(i)}}{\bar{\sigma}_{i}^7 \Delta_B^{1/2}}+ \frac{ 3\sqrt {\Delta_B}}{\bar{\sigma}_{i}^3}\r) &  0\\ 
0 & \inv{2a_0^{4}} + \frac{\textnormal{cum}_4[\epsilon]}{4a_0^8}  \\\end{matrix} \r)\\
&=& \l( \begin{matrix}  a_0\l(\frac{5\calq_{(i)}}{\bar{\sigma}_{i} \Delta_B^{3/2}} + \frac{3\bar{\sigma}_{i}^3}{\sqrt {\Delta_B}}\r) &  0\\ 
0 & 2a_0^{4} + \textnormal{cum}_4[\epsilon]  \\\end{matrix} \r)\\
&=&V_{(i)}^{'},
\eeas 
and thus we have shown the CLT for $\boldsymbol{\widehat{\xi}}_n$. Now to get Theorem \ref{CLTQMLEjumps}, it is sufficient to notice that 
\begin{eqnarray}
\label{last}
\left(\begin{matrix} N_n^{1/4} \l(\tilde{Q} - T\bar{\sigma}_0^2 \r) \\ 
N_n^{1/2} \l( B^{-1} \sum_{i=1}^B \widehat{a}_{n,(i)}^2 - a_0^2 \r) 
\end{matrix}\right) = \Phi_n A \boldsymbol{\Phi}_n^{-1/2} \boldsymbol{\Phi}_n^{1/2} \l(\boldsymbol{\widehat{\xi}}_n - \boldsymbol{\xi}_0 \r),
\end{eqnarray}
where $A \in \reels^{2 \times 2B}$ has the form 
$$A = \l(\begin{matrix}  \Delta_B&0&\cdots&\Delta_B&0 \\ 0&B^{-1}&\cdots&0&B^{-1}\end{matrix}\r),$$
and from here we easily conclude that the left-hand side of (\ref{last}) admits a CLT with the claimed asymptotic variance. Finally Theorem \ref{CLTQMLE} is a particular case of Theorem \ref{CLTQMLEjumps}.


\end{proof}

\subsection{Proof of Theorem \ref{RKth}}
Some details of the proof are omitted as the techniques used are very close to the QMLE case. We need to introduce some notation. We consider the block constant processes defined as
\begin{eqnarray*}
\tilde{c}_t  =  c_i &\text{ where } &\Tau_{i-1} \leq t < \Tau_{i},\\
\rho_t  =  \rho_{\Tau_{i-1}, \Tau_i} &\text{ where } &\Tau_{i-1} \leq t < \Tau_{i}.\\
\xi_t^2  =  \xi_{\Tau_{i-1}, \Tau_i}^2 &\text{ where } &\Tau_{i-1} \leq t < \Tau_{i}.
\end{eqnarray*}
Condition (\ref{RKth2}) in Theorem \ref{RKth} can be re-expressed as
\begin{eqnarray*}
n^{1/4} \l( \tilde{K} - \int_0^T \sigma_u^2 du \r) &  \overset{L_X}{\rightarrow} & \calm\caln \l(0, 4 B^{1/2} \Delta_B \int_0^T \sigma_u^4 (\tilde{c}_u k_{\bullet}^{0,0} + \tilde{c}_u^{-1} 2 k_{\bullet}^{1,1} \rho_u \xi_u^2 + \tilde{c}_u^{-3} k_{\bullet}^{2,2} \xi_u^4 )du \r).
\end{eqnarray*}
We also define the kernels for general processes $A_t$ and $C_t$ as
$$K(A,C) = \gamma_0 (A,C) + \sum_{h=1}^{H} k \l( \frac{h-1}{H} \r) \l( \gamma_h (A,C) + \gamma_{-h} (A,C) \r),$$ 
where the realized autocovariance is defined as
$$\gamma_h (A,C) = \sum_{j=1}^{n} (A_{\Delta j} - A_{\Delta (j-1)}) (C_{\Delta (j - h)} - C_{\Delta (j-h-1)}),$$
with $h = -H, \cdots, -1, 0, 1 , \cdots, H$. We further define $K_i (A,C)$ the estimate on the $i$th block and we aggregate the local estimates to define the adapted version of $K(A,C)$ as 
$$\tilde{K}(A,C) = \sum_{i=1}^{B} K_i(A,C).$$ 
We follow the same line of reasoning as in the proof of Theorem 4 (p. 1530, \cite{barndorff2008designing}). Accordingly, we just need to show an adapted version of Theorem 3 (p. 1492). Theorem \ref{RKth} will then follow from Lemma 1 (p. 1523) and Proposition \ref{propositionConditionalStable}. From now on, we aim to show the adapted version of Theorem 3 (p. 1492, \cite{barndorff2008designing}) which is stated in what follows.
\begin{RKth3} (Adapted version of Theorem 3 in \cite{barndorff2008designing}) We assume that $H = c n^{1/2}$. As $n \rightarrow \infty$ we have that
\begin{eqnarray}
\label{th31} n^{1/4} \l( \tilde{K}(X,X) - \int_0^T \sigma_u^2 du \r) & \overset{L_X}{\rightarrow} & \calm\caln \l(0, 4 k_{\bullet}^{0,0} B^{1 - \alpha} \Delta_B \int_0^T \sigma_u^4 \tilde{c}_u^{-1} du \r),\\
\label{th32} n^{1/4} ( \tilde{K} (X,U) + \tilde{K} (U,X) ) & \overset{L_X}{\rightarrow} &  \calm\caln \l(0, 8 \omega^2 k_{\bullet}^{1,1} B^{\alpha} \int_0^T \sigma_u^2 \tilde{c}_u du \r).
\end{eqnarray}
In addition, when $k'(0)^2 + k'(1)^2 = 0$, the asymptotic variance of $\tilde{K} (U)$ is equivalent to
\begin{eqnarray}
\label{th34} 4 \omega^4 \l( n^{-1/4} B^{3 \alpha - 1} k_{\bullet}^{2,2} \sum_{i=1}^{B} c_i^{3}  + (B^{1/2}/(n^{1/2} m))  \big\{k_{\bullet}^{1,1} \sum_{i=1}^{B} c_i + \sum_{i=2}^{B} \tilde{k}_{\bullet}^{1,1} (c_i, c_{i-1}) \sqrt{c_i c_{i-1}} \big\} \r), 
\end{eqnarray}
where $\tilde{k}_{\bullet}^{1,1} (c_1, c_{2}) = \int_0^1 k'(x) k'(ax) dx$ with $a = \min(c_1, c_2)/ \max(c_1, c_2)$ and if $ 1/m \rightarrow 0$
\begin{eqnarray}
\label{th35}
 \sqrt{\frac{\tilde{H}^3}{n}} \tilde{K} (U,X)  & \overset{L_X}{\rightarrow} &  N \l(0, 4 \omega^4 k_{\bullet}^{2,2} B^{3 \alpha - 1} \sum_{i=1}^{B} c_i^3 \r).
\end{eqnarray}
\end{RKth3}
To show (\ref{th31}), we consider the continuous interpolated martingale $M_t = \sqrt{\frac{n}{\tilde{H}}} \l( \tilde{K}(X,X) - \int_0^t \sigma_u^2 du \r)$. As for the QMLE, we aim to use Theorem 2.1 (\cite{jacod1997}). To show condition (2.9), i.e. that $[M, M]_t \overset{\proba}{\rightarrow} 4 k_{\bullet}^{0,0} B^{1 - \alpha} \Delta_B \int_0^t \sigma_u^4 \tilde{c}_u^{-1} du$, we express $M_t$ as $\sum_{i=1}^B M_t^{(i)}$, where $M_t^{(i)}$ are such that $M_t^{(i)} = 0$ for $t \in [0,\Tau_{i-1}]$, $M_t^{(i)} = M_t$ on $[\Tau_{i-1}, \Tau_i]$ and $M_t^{(i)} = M_{\Tau_i}$ for $t \in [\Tau_i, T]$. We can easily show that 
\begin{eqnarray}
\label{MM} [ M, M]_t = \sum_{i=1}^{B} [M^{(i)}, M^{(i)}]_{t}.
\end{eqnarray}
The $K(X)$ case in the proof of Theorem 3 (p. 1528, \cite{barndorff2008designing}) is based on a martingale theorem which shows that 
\begin{eqnarray}
\label{MiMi}
[M^{(i)}, M^{(i)}]_{t} \overset{\proba}{\rightarrow} 4 k_{\bullet}^{0,0} B^{1 - \alpha} (t - \Tau_{i-1}) \int_{\Tau_{i-1} \wedge t}^{\Tau_{i} \wedge t} \sigma_u^4 \tilde{c}_u^{-1} du.
\end{eqnarray}
In view of (\ref{MM}) and (\ref{MiMi}), we have thus shown that $[M, M]_t \overset{\proba}{\rightarrow} 4 k_{\bullet}^{0,0} B^{1 - \alpha} \Delta_B \int_0^t \sigma_u^4 \tilde{c}_u^{-1} du$.
We show condition (2.10), i.e. that $\langle M, W \rangle_t \overset{\proba}{\rightarrow} 0$,
by a straightforward calculation on the Brownian motion increments.
Also, condition (2.11) is satisfied because $M$ is continuous. Finally we show that condition (2.12) hold, i.e. for any bounded martingale $\overline{N}$ orthogonal to $W$ we have that 
\bea 
\langle M, \overline{N} \rangle_t = 0.
\eea 
This can be proven with the same line of reasoning as for Lemma \ref{lemmaM2} for the QMLE.

\smallskip
The proof for (\ref{th32}) can adapt directly from the cross-term $K(X,U) + K(U,X)$ part in the proof of Theorem 3 (p. 1528, \cite{barndorff2008designing}). Indeed, on each block we have the convergence discussed on p. 1525, and it is clear that as the block terms are uncorrelated to each other conditioned on $X_t$, we obtain the convergence of the vector block estimates, with correlation limit between two different block terms equal to 0.

\smallskip
We aim to show now (\ref{th34}). In view of (A.3) on p. 1528 in \cite{barndorff2008designing}, we have 
$$\tilde{K} (U) = \sum_{i=1}^{B} \Big\{ \underbrace{- \sum_{h=1}^{H^{(i)}} (w_{h+1}^{(i)} - 2 w_{h}^{(i)} + w_{h-1}^{(i)} ) V_h^{(i)}}_{A_i} \underbrace{- \sum_{h=1}^{H^{(i)}} (w_{h+1}^{(i)} - w_{h-1}^{(i)}) R_h^{(i)}}_{C_i} \Big\},$$
where $w_h^{(i)} = k (\frac{h-1}{H^{(i)}})$ and $V_h^{(i)} = \sum_{j=(i-1) n/B + 1 }^{in/B} (U_{t_{j}} U_{t_{j-h}} +  U_{t_j} U_{t_{j+h}} + U_{t_{j-1}} U_{t_{j-1-h}} + U_{t_{j-1}} U_{t_{j-1+h}})$ and $C_i$ is due to end-effects. We have that $A_i \overset{\mathcal{L}}{\rightarrow} A_i^{(l)}$ and  $C_i \overset{\mathcal{L}}{\rightarrow} C_i^{(l)}$ for some normally distributed variables $A_i^{(l)}$ and $C_i^{(l)}$ from the proof on p. 1529 in \cite{barndorff2008designing}. Actually, we can show that the convergence still holds for the random vector $(A_1, \cdots, A_B, C_1, \cdots, C_B)$ and thus we have that $\tilde{K} (U) \overset{\mathcal{L}}{\rightarrow} N(0,V)$ where $V$ is equal to 
\begin{eqnarray}
\label{sumcov}
\sum_{1 \leq i,j \leq B} \mathrm{Cov} (A_i^{(l)}, A_j^{(l)}) + \mathrm{Cov} (C_i^{(l)}, C_j^{(l)}).
\end{eqnarray}
 We have that $n^{1/2} \sum_{i=1}^{B} \var (A^{(i)}) = B^{1/2} k_{\bullet}^{2,2} \sum_{i=1}^{B} c_i^{3}$, which shows the convergence to the first term in (\ref{th34}). The second term is obtained as 
$$(B^{1/2}/(n^{1/2} m))^{-1} \sum_{i=1}^B \var (C_i) + 2 \sum_{i=2}^B \mathrm{Cov} (A_i, A_{i-1}) = \big\{k_{\bullet}^{1,1} \sum_{i=1}^{B} c_i + \sum_{i=2}^{B} \tilde{k}_{\bullet}^{1,1} (c_i, c_{i-1}) \sqrt{c_i c_{i-1}} \big\}.$$
The other terms in (\ref{sumcov}) go to 0, thus we have shown (\ref{th34}). The convergence (\ref{th35}) is obtained as a straightforward consequence of (\ref{th34}).

\subsection{Proof of Theorem \ref{RKthjumps}}
The proof adding jumps and stochastic observation times follows the same line of reasoning as for the QMLE case.

\subsection{Proof of Corollary \ref{RKcor} and Corollary \ref{QMLEcor}}
By Slutsky's Lemma, both corollaries will be proved if we have the consistency of the AVAR estimators. This is a consequence of the consistency of the estimators $\widehat{\int_{\Tau_{i-1}}^{\Tau_i} \sigma_u^2 du}$ and $\widehat{\int_{\Tau_{i-1}}^{\Tau_i} \sigma_u^4 du}$ by Theorem 3.1 and Remark 4 in \cite{jacod2009microstructure} along with the consistency of $\widehat{a}^2$ by, e.g., (21) in \cite{zhang2005tale}. 
\subsection{Proof of Proposition \ref{propAVARRK}}

 $AVAR_{[\Tau_{i-1}, \Tau_i]}^{(RK,c_i^*,2)}$ takes on the form
\beas  
AVAR_{[\Tau_{i-1}, \Tau_i]}^{(RK,c_i^*,2)} =  a_0\sqrt{B} \l( \Delta_B \int_{T_{i-1}}^{T_i} \sigma_u^4 du \r)^{3/4} g (\rho_{\Tau_{i-1},\Tau_i}).
\eeas
In view of (\ref{smsp}), we easily obtain $0 < \underline{\rho} \leq \rho_{T_{i-1},T_i} \leq 1$ where $\underline{\rho} = \frac{\underline{\sigma}^2}{\overline{\sigma}^2}$. This gives us the estimate $8 \leq g (\rho_{\Tau_{i-1},\Tau_i}) \leq \overline{g} < \infty$ for some $\overline{g}$.

Let us define for $B \in \naturels$, $B \geq 1$, the random set $J_B := \{i \in \{ 1, \cdots, B \} | \sigma \text{ jumps on } (\Tau_{i-1}, \Tau_i] \}$. Because the jumps in $\sigma$ are of finite activity, almost surely the cardinal of $J_B$, defined as $|J_B|$, tends to a finite value. Thus we can get rid of the terms $AVAR_{[\Tau_{k-1}, \Tau_k]}^{(RK,c_i^*,2)}$ for which $k$ is contained into $J_B$ because 
\beas 
\sum_{i \in J_B} AVAR_{[\Tau_{i-1}, \Tau_i]}^{(RK,c_i^*,2)} \leq |J_B| \Delta_B T^{1/2}a_0 \overline{\sigma}^3 \overline{g} \overset{a.s.}{\rightarrow} 0,
\eeas
and similarly
\beas 
g(1) a_0 T^{1/2}\sum_{i \in J_B} \int_{\Tau_{i-1}}^{\Tau_i} \sigma_u^3 du \overset{a.s.}{\rightarrow} 0.
\eeas
Thus, the proposition will be proved if we show
\beas 
\sum_{i \notin J_B}{AVAR_{[\Tau_{i-1}, \Tau_i]}^{(RK,c_i^*,2)}} - g (1) a_0 T^{1/2}\sum_{i \notin J_B} \int_{T_{i-1}}^{T_i} \sigma_u^3 du \overset{a.s.}{\rightarrow} 0.
\eeas

As the continuous part of $\sigma$ is assumed to be an It\^{o} process with bounded components, some calculation shows that for any $p >0 ,q \geq 1$, and uniformly in $i \notin J_B$ we have the following expansion
\beas 
\int_{\Tau_{i-1}}^{\Tau_i}{\sigma_u^p du} = \sigma_{T_{i-1}}^p\Delta_B+ O_{\mathbb{L}^q}(\Delta_B^{3/2}),   
\eeas
where $A= O_{\mathbb{L}^q}(C)$, $C > 0$ means that $\esp \big| \frac{A}{C} \big|^q$ is bounded. Thus, using again (\ref{smsp}), we also obtain the expansions $\rho_{T_{i-1},T_i} = 1 + O_{\mathbb{L}^q}(\Delta_B^{1/2})$, $g(\rho_{\Tau_{i-1},\Tau_i}) = g(1) + O_{\mathbb{L}^q}(\Delta_B^{1/2})$, and $(\Delta_B \int_{T_{i-1}}^{T_i} \sigma_u^4 du )^{3/4} = \Delta_B^{1/2} \int_{T_{i-1}}^{T_i} \sigma_u^3 du  +O_{\mathbb{L}^q}(\Delta_B^{2})$ to get finally the estimate
\beas 
AVAR_{[\Tau_{i-1}, \Tau_i]}^{(RK,c_i^*,2)} = g(1) a_0 T^{1/2} \int_{T_{i-1}}^{T_i} \sigma_u^3 du + O_{\mathbb{L}^q}(\Delta_B^{3/2})
\eeas 
uniformly in $i \notin J_B$. At this stage we have thus proved that
\beas 
\sum_{i \notin J_B}{AVAR_{[\Tau_{i-1}, \Tau_i]}^{(RK,c_i^*,2)}} - g(1) a_0 T^{1/2}\sum_{i \notin J_B} \int_{T_{i-1}}^{T_i} \sigma_u^3 du = O_{\mathbb{L}^q}(\Delta_B^{1/2}).
\eeas 
To get the almost sure convergence to 0, we define $Y_B$ as the left hand side of the previous equality and note that $\esp \sum_{B=1}^{+\infty}{|Y_B|^q} < +\infty$ for any $q >2$. This gives us that $\sum_{B=1}^{+\infty}{|Y_B|^q} < +\infty \text{ a.s.}$ and so $|Y_B|^q \overset{a.s.}{\rightarrow} 0$, which completes the proof.

\subsection{Proof of Remark \ref{rklossprop}}
\label{proofrk}
We show the inequality $g(\rho) \kappa^{-1} \geq g(1)$ for any admissible couple $(\rho,\kappa)$. Note that by the domination $\kappa \leq \rho^{1/2}$ obtained on the account of (\ref{RhoKappa}), it is sufficient to show that  the function $f : \rho \to \rho^{-1/2}g(\rho)$ is decreasing on the interval $(0,1]$. We let $p(\rho) = \sqrt{1 + \sqrt{1 + 3 d / \rho^2}}$, and a short calculation shows us that $f'(\rho)$ has the same sign as $p'(\rho)(1-p(\rho)^{-2})$. Therefore, the inequality $p(\rho) \geq 1$ implies that $f$ is decreasing if and only if $p$ is, which is obvious.

\subsection{Proof of Proposition \ref{propAVARQMLE}}
This proof follows the same line of reasoning as for the proof of Proposition \ref{propAVARRK}.

\subsection{Proof of Proposition \ref{RKcorRobust}}

When $J=0$, this is a straightforward adaptation of the proof of Proposition \ref{propAVARRK} using the new estimates for any $q \geq 1$ 
$$ R_{(i)}^{1/2} = \Delta_B^{-1/2} \l(\int_0^T\alpha_s^{-1}ds\r)^{1/2} \alpha_{\Tau_{i-1}} + o_{\mathbb{L}^q}(\Delta_B^{-1/2}),$$
$$ \l(\Delta_B \calq_{(i)}\r)^{3/4} = \Delta_B^{3/2} \sigma_{\Tau_{i-1}}^3 + O_{\mathbb{L}^q}(\Delta_B^2),$$
and 
$$ g(\widetilde{\rho}_{\Tau_{i-1},\Tau_i}) = g(1) + O_{\mathbb{L}^q}\l(\Delta_B^{1/2}\r).$$

\subsection{Proof of Proposition \ref{RKcorRobust2}}
When $J \neq 0$, the situation is fairly different. Let us define the random set 
$$J_B^X := \{i \in \{ 1, \cdots, B \} | X \text{ jumps on } (\Tau_{i-1}, \Tau_i] \}.$$
Since $J$ is of finite activity, by taking $n$ sufficiently large, we may assume that for any $i \in J_B^X$, $X$ jumps exactly once on $(\Tau_{i-1}, \Tau_i]$. Splitting the sum of local variances
\beas 
AVAR_B^{(RK,rob)} &=& \sum_{i \in J_B^X} R_{(i)}^{1/2}AVAR_{[\Tau_{i-1}, \Tau_i]}^{(RK,rob,\widetilde{c}_i^{*})} + \sum_{i \notin J_B^X} R_{(i)}^{1/2}AVAR_{[\Tau_{i-1}, \Tau_i]}^{(RK,rob,\widetilde{c}_i^{*})} \\
&=& I + II,
\eeas 
again by the finite activity property of $J$ we easily deduce from the proof of Proposition \ref{RKcorRobust} that 
$$II  \overset{a.s.}{\rightarrow} 8g(1)a_0\l(\int_0^T\alpha_s^{-1}ds\r)^{1/2} \int_0^T{\alpha_s^{1/2}\sigma_s^3}ds.$$ 
Now we derive the limit of $I$. We write $\tau_1,\cdots,\tau_{\widetilde{N}_J}$ the jump times of $J$ labeled such that for any $i \in J_B^X$, $\Tau_{i-1} < \tau_i \leq \Tau_i$. For any $i \in J_B^X$, we have the estimates 
$$ \bar{\sigma}_{i}^2 \overset{a.s.}{\sim} \Delta_B^{-1} \Delta J_{\tau_i}^2,$$
$$ \calq_{(i)} \overset{a.s.}{\sim} \Delta J_{\tau_i}^2 \l(\sigma_{\tau_i}^2 + \sigma_{\tau_i-}^2 \r),$$
where for the latter expression we used the continuity of $\alpha$ at time $\tau_i$ which is a consequence of the independence of $\alpha$ and $X$. We thus have
$$ \widetilde{\rho}_{\Tau_{i-1}, \Tau_i} \overset{a.s.}{\sim} \Delta_B^{-1/2} |\Delta J_{\tau_i}|\l(\sigma_{\tau_i}^2 + \sigma_{\tau_i-}^2 \r)^{-1/2}.$$
Combined with $g(\rho) \overset{\rho \to +\infty}{\sim}  \frac{16}{3} \sqrt{\rho k_{\bullet}^{0,0} k_{\bullet}^{1,1}} \l(\inv{\sqrt 2} + \sqrt{2}\r)$, we deduce that 
$$ I \overset{a.s.}{\rightarrow} \frac{16}{3}a_0\l(\inv{\sqrt{2}} + \sqrt{2}\r) \sqrt{k_{\bullet}^{0,0} k_{\bullet}^{1,1}} \l(\int_0^T\alpha_s^{-1}ds\r)^{1/2}\sum_{0 < s \leq T} \Delta J_s^2  \l(\sigma_s^2\alpha_s  + \sigma_{s-}^2\alpha_{s-} \r)^{1/2}.$$

\subsection{Proof of Proposition \ref{QMLEcorRobust} and Proposition \ref{QMLEcorRobust2}}
The proofs follow exactly the same line of reasoning as the proofs of Proposition \ref{RKcorRobust} and Proposition \ref{RKcorRobust2}.

\bibliography{biblio}
\bibliographystyle{apalike} 

\newpage
\begin{table}[h]
\centering
\caption{Sample mean and standard error of $\rho$ and $\kappa$ for the three models.}
\label{tableRho}
\begin{tabular}{lcccc}
\toprule
\toprule
Model   & $\rho_{mean}$ & $\rho_{stdv.}$ &$\kappa_{mean}$& $\kappa_{stdv.}$  \\
\toprule
Model 1 & 0.89  & 0.01  & 0.92 & 0.01 \\
Model 2 & 0.77  & 0.15  & 0.83 & 0.12\\
Model 3 & 0.64  & 0.14  & 0.74 & 0.1 \\
\bottomrule
\end{tabular}
\end{table}

\begin{table}[]
\centering
\caption{Finite sample properties of $Z_n^{\tilde{K}_B}$ (Model 2)$^\dag$} 
\label{stdRK}
\begin{tabular}{@{}rccccccccc@{}}
\toprule
\toprule
\multicolumn{1}{l}{No. Obs.} & \multicolumn{1}{l}{Mean} & \multicolumn{1}{l}{Stdv.} & \multicolumn{1}{l}{RMSE} & \multicolumn{1}{l}{0.5\%} & \multicolumn{1}{l}{2.5\%} & \multicolumn{1}{l}{5\%} & \multicolumn{1}{l}{95\%} & \multicolumn{1}{l}{97.5\%} & \multicolumn{1}{l}{99.5\%} \\ \toprule
\multicolumn{10}{l}{B = 1 block} \\
5,850 & -0.042&1.102&1.103&0.29&1.75&3.77&96.62&98.62&99.85\\
11,700 & -0.032&1.067&1.068&0.39&1.96&3.98&96.01&98.20&99.80\\
23,400 & -0.030&1.044&1.044&0.41&2.13&4.16&95.63&97.84&99.70\\
46,800 &  -0.027&1.041&1.041&0.46&2.25&4.35&95.58&98.18&99.74\\
\multicolumn{10}{l}{B = 2 blocks} \\
5,850 & -0.065&1.105&1.106&0.24&1.55&3.53&96.49&98.43&99.82\\
11,700 & -0.048&1.069&1.070&0.32&1.85&3.65&95.89&98.20&99.71\\
23,400 &  -0.042&1.048&1.049&0.37&2.01&3.91&95.52&97.88&99.65\\
46,800 &  -0.037&1.044&1.045&0.43&2.11&4.10&95.54&98.06&99.71\\
\multicolumn{10}{l}{B = 4 blocks} \\
5,850 & -0.105&1.110&1.115&0.21&1.38&3.02&96.25&98.37&99.81\\
11,700 & -0.082&1.074&1.077&0.28&1.54&3.33&95.74&98.15&99.66\\
23,400 &  -0.069&1.051&1.053&0.36&1.77&3.66&95.22&97.73&99.65\\
46,800 & -0.059&1.043&1.044&0.37&1.89&3.89&95.31&97.96&99.64\\
\multicolumn{10}{l}{B = 6 blocks} \\
5,850 & -0.144&1.115&1.124&0.19&1.23&2.72&95.81&98.33&99.75\\
11,700 & -0.114&1.077&1.083&0.23&1.40&3.09&95.37&97.88&99.67\\
23,400 & -0.099&1.054&1.059&0.31&1.65&3.49&94.95&97.57&99.60\\
46,800 & -0.086&1.043&1.047&0.38&1.65&3.56&94.95&97.76&99.57\\
\multicolumn{10}{l}{B = 8 blocks} \\
5,850 &-0.193&1.119&1.136&0.15&1.03&2.31&95.40&98.14&99.75\\
11,700 & -0.154&1.080&1.091&0.21&1.26&2.83&95.27&97.88&99.66\\
23,400 & -0.128&1.054&1.062&0.28&1.52&3.27&94.91&97.56&99.59\\
46,800 &-0.109&1.042&1.047&0.32&1.64&3.51&94.72&97.56&99.55\\
\bottomrule

\end{tabular}

\scriptsize $^\dag$This table shows summary statistics and empirical quantiles benchmarked to the $N$(0,1) distribution for the infeasible Z-statistics related to the global and local RK (Tukey-Hanning 2). The simulation design is Model 2 with $M = 10,000$ Monte-Carlo simulations. 
\end{table}

\begin{table}[]
\centering
\caption{Finite sample properties of $Z_n^{\tilde{Q}_B}$ (Model 2)$^\dag$}
\label{stdQMLE}
\begin{tabular}{@{}rccccccccc@{}}
\toprule
\toprule
\multicolumn{1}{l}{No. Obs.} & \multicolumn{1}{l}{Mean} & \multicolumn{1}{l}{Stdv.} & \multicolumn{1}{l}{RMSE} & \multicolumn{1}{l}{0.5\%} & \multicolumn{1}{l}{2.5\%} & \multicolumn{1}{l}{5\%} & \multicolumn{1}{l}{95\%} & \multicolumn{1}{l}{97.5\%} & \multicolumn{1}{l}{99.5\%} \\ \toprule
\multicolumn{10}{l}{B = 1 block} \\
5,850 &-0.024&1.084&1.084&0.36&2.09&4.12&96.48&98.57&99.84\\
11,700 & -0.015&1.058&1.058&0.43&2.26&4.51&96.32&98.34&99.75\\
23,400 &  -0.012&1.039&1.039&0.51&2.19&4.48&95.87&97.97&99.70\\
46,800 &  -0.013&1.034&1.034&0.59&2.38&4.67&95.74&98.06&99.73\\
\multicolumn{10}{l}{B = 2 blocks} \\
5,850 & -0.023&1.086&1.086&0.34&1.90&4.03&96.48&98.46&99.85\\
11,700 & -0.011&1.06&1.06&0.42&2.08&4.28&96.24&98.31&99.72\\
23,400 & -0.007&1.042&1.042&0.54&2.12&4.31&95.71&98.01&99.66\\
46,800 & -0.009&1.036&1.036&0.56&2.22&4.51&95.71&98.08&99.63\\
\multicolumn{10}{l}{B = 4 blocks} \\
5,850 & -0.016&1.089&1.089&0.34&1.98&3.86&96.42&98.57&99.82\\
11,700 & -0.007&1.063&1.063&0.42&2.09&4.19&96.24&98.34&99.72\\
23,400 & -0.002&1.042&1.042&0.51&2.16&4.52&95.64&98.02&99.65\\
46,800 &  -0.005&1.035&1.035&0.56&2.20&4.74&95.60&98.08&99.69\\
\multicolumn{10}{l}{B = 6 blocks} \\
5,850 & -0.012&1.089&1.089&0.36&2.01&3.93&96.46&98.56&99.82\\
11,700 & -0.002&1.062&1.062&0.43&2.01&4.26&96.33&98.33&99.74\\
23,400 &  -0.0&1.041&1.041&0.51&2.10&4.63&95.60&98.06&99.71\\
46,800 &  -0.004&1.033&1.033&0.57&2.21&4.72&95.64&98.07&99.69\\
\multicolumn{10}{l}{B = 8 blocks} \\
5,850 & -0.014&1.093&1.093&0.36&1.82&3.83&96.42&98.63&99.82\\
11,700 & -0.005&1.066&1.066&0.40&1.94&4.15&96.33&98.37&99.74\\
23,400 & -0.001&1.043&1.043&0.47&2.05&4.53&95.64&98.15&99.67\\
46,800 &  -0.003&1.033&1.033&0.57&2.29&4.67&95.60&98.10&99.66\\
\bottomrule
\end{tabular}

\scriptsize $^\dag$This table shows summary statistics and empirical quantiles benchmarked to the $N$(0,1) distribution for the infeasible Z-statistics related to the global and local QMLE. The simulation design is Model 2 with $M = 10,000$ Monte-Carlo simulations. 
\end{table}

\begin{table}[]
\centering
\caption{Finite sample properties of $\tilde{Z}_n^{\tilde{K}_B}$ (Model 2)$^\dag$} 
\label{stdRKfeasible}
\begin{tabular}{@{}rccccccccc@{}}
\toprule
\toprule
\multicolumn{1}{l}{No. Obs.} & \multicolumn{1}{l}{Mean} & \multicolumn{1}{l}{Stdv.} & \multicolumn{1}{l}{RMSE} & \multicolumn{1}{l}{0.5\%} & \multicolumn{1}{l}{2.5\%} & \multicolumn{1}{l}{5\%} & \multicolumn{1}{l}{95\%} & \multicolumn{1}{l}{97.5\%} & \multicolumn{1}{l}{99.5\%} \\ \toprule
\multicolumn{10}{l}{B = 1 block} \\
5,850 & -0.117&1.176&1.182&0.02&0.66&1.67&95.42&97.67&99.66
\\
11,700 & -0.080&1.128&1.131&0.01&0.58&2.01&95.00&97.08&99.38\\
23,400 & -0.083&1.098&1.101&0.02&0.52&2.26&94.89&97.22&99.04\\
46,800 &  -0.069&1.087&1.089&0.02&0.73&3.44&94.49&97.27&99.15\\
\multicolumn{10}{l}{B = 2 blocks} \\
5,850 & -0.130&1.148&1.155&0.02&0.58&1.76&95.40&97.27&99.40
\\
11,700 & -0.099&1.111&1.115&0.01&0.66&2.17&94.53&96.23&99.08\\
23,400 &  -0.086&1.083&1.086&0.05&0.71&2.28&95.03&97.17&99.24\\
46,800 & -0.072&1.071&1.073&0.06&1.17&4.10&94.74&97.63&99.41
\\
\multicolumn{10}{l}{B = 4 blocks} \\
5,850 & -0.173&1.136&1.149&0.03&0.64&1.71&94.56&97.28&98.98
\\
11,700 &-0.139&1.107&1.115&0.01&0.64&1.78&93.28&96.21&99.13\\
23,400 & -0.107&1.083&1.089&0.08&0.88&2.37&94.97&96.83&99.03\\
46,800 & -0.092&1.079&1.083&0.03&1.07&3.93&94.99&97.64&99.34
\\
\multicolumn{10}{l}{B = 6 blocks} \\
5,850 & -0.225&1.145&1.167&0.02&0.60&1.25&93.88&96.93&98.80
\\
11,700 & -0.177&1.103&1.117&0.01&0.53&1.52&93.21&95.86&99.04
\\
23,400 & -0.145&1.077&1.087&0.04&0.72&1.91&94.12&96.61&99.06\\
46,800 &-0.122&1.07&1.077&0.04&1.01&3.55&94.82&97.07&99.35
\\
\multicolumn{10}{l}{B = 8 blocks} \\
5,850 &-0.270&1.152&1.183&0.01&0.45&1.080&94.41&96.88&98.66
\\
11,700 & -0.219&1.106&1.128&0.02&0.50&1.65&92.63&95.76&98.96\\
23,400 &-0.176&1.089&1.103&0.07&0.55&1.96&93.63&97.34&98.97\\
46,800 &-0.146&1.078&1.088&0.03&0.72&3.27&94.02&96.98&99.38
\\
\bottomrule

\end{tabular}

\scriptsize $^\dag$This table shows summary statistics and empirical quantiles benchmarked to the $N$(0,1) distribution for the feasible Z-statistics related to the global and local RK (Tukey-Hanning 2). The simulation design is Model 2 with $M = 10,000$ Monte-Carlo simulations. 
\end{table}

\begin{table}[]
\centering
\caption{Finite sample properties of $\tilde{Z}_n^{\tilde{Q}_B}$ (Model 2)$^\dag$}
\label{stdQMLEfeasible}
\begin{tabular}{@{}rccccccccc@{}}
\toprule
\toprule
\multicolumn{1}{l}{No. Obs.} & \multicolumn{1}{l}{Mean} & \multicolumn{1}{l}{Stdv.} & \multicolumn{1}{l}{RMSE} & \multicolumn{1}{l}{0.5\%} & \multicolumn{1}{l}{2.5\%} & \multicolumn{1}{l}{5\%} & \multicolumn{1}{l}{95\%} & \multicolumn{1}{l}{97.5\%} & \multicolumn{1}{l}{99.5\%} \\ \toprule
\multicolumn{10}{l}{B = 1 block} \\
5,850 &-0.114&1.200&1.205&0.01&0.50&1.29&95.50&97.99&99.59
\\
11,700 &-0.090&1.148&1.152&0.01&0.35&1.68&95.36&97.18&99.12
\\
23,400 &-0.075&1.109&1.112&0.01&0.62&2.00&95.13&96.91&98.18\\
46,800 &-0.062&1.093&1.095&0.01&0.61&2.98&94.38&96.67&99.08
\\
\multicolumn{10}{l}{B = 2 blocks} \\
5,850 &-0.099&1.170&1.174&0.02&0.56&1.49&95.70&97.66&99.38
\\
11,700 &-0.080&1.130&1.133&0.01&0.49&1.76&94.89&96.81&99.05\\
23,400 &-0.057&1.094&1.095&0.03&0.91&2.38&95.25&97.32&98.61\\
46,800 &-0.049&1.079&1.081&0.02&0.99&3.61&94.93&97.24&99.40
\\
\multicolumn{10}{l}{B = 4 blocks} \\
5,850 &-0.089&1.150&1.154&0.04&0.82&1.62&95.56&97.50&99.18
\\
11,700 & -0.077&1.112&1.114&0.05&0.56&1.92&94.80&96.73&99.17
\\
23,400 &-0.049&1.083&1.084&0.06&1.11&2.91&95.15&97.36&98.77\\
46,800 &-0.046&1.083&1.084&0.02&1.16&3.38&95.23&97.32&99.53\\
\multicolumn{10}{l}{B = 6 blocks} \\
5,850 &-0.090&1.146&1.149&0.07&0.91&1.65&95.89&97.47&99.00
\\
11,700 &-0.076&1.105&1.107&0.05&0.58&2.36&94.70&96.98&99.21
\\
23,400 & -0.050&1.077&1.078&0.06&1.07&3.10&95.10&97.62&98.78\\
46,800 & -0.046&1.076&1.077&0.03&1.19&3.69&95.35&97.39&99.44
\\
\multicolumn{10}{l}{B = 8 blocks} \\
5,850 &-0.090&1.145&1.148&0.08&0.78&1.96&95.71&97.39&99.27
\\
11,700 & -0.073&1.099&1.101&0.06&0.68&2.46&94.83&96.81&99.11
\\
23,400 &-0.045&1.076&1.077&0.08&1.17&2.51&95.22&97.50&98.96\\
46,800 &-0.046&1.080&1.081&0.03&1.46&3.86&95.39&97.26&99.46
\\
\bottomrule
\end{tabular}

\scriptsize $^\dag$This table shows summary statistics and empirical quantiles benchmarked to the $N$(0,1) distribution for the feasible Z-statistics related to the global and local QMLE. The simulation design is Model 2 with $M = 10,000$ Monte-Carlo simulations. 
\end{table}

\begin{table}
\centering
\caption{Losses$^\dag$}
\label{tableLossXimedium}
\resizebox{.92 \textwidth}{!}{\begin{tabular}{@{}rccccccccccc@{}}
\toprule
\toprule
Model  &   & $Q$ & $\tilde{Q}_2$ & $\tilde{Q}_4$ & $\tilde{Q}_6$ & $\tilde{Q}_8$ & $K$ & $\tilde{K}_2$ & $\tilde{K}_4$ & $\tilde{K}_6$ & $\tilde{K}_8$  \\
\toprule
\multicolumn{12}{l}{$n = 23,400$, $\xi^2 = 0.01$}                                                                                                                         \\
\multirow{2}{*}{Model 1} & Emp. & 8.4\%&6.7\%&5.3\%&3.1\%&4.1\%&14.3\%&12.4\%&13.9\%&14.6\%&18.0\%\\
  &Theo. & 6.9\%&5.4\%&2.6\%&1.4\%&0.9\%&9.9\%&8.4\%&6.0\%&5.0\%&4.5\%\\
\multirow{2}{*}{Model 2} & Emp. &29.3\%&18.4\%&12.8\%&10.7\%&8.6\%&30.2\%&23.4\%&20.3\%&17.3\%&22.5\%\\
  &Theo. &21.5\%&12.4\%&5.8\%&3.5\%&2.4\%&18.2\%&12.1\%&8.0\%&6.4\%&5.6\%\\
\multirow{2}{*}{Model 3} & Emp. & 40.0\%&21.0\%&9.4\%&6.4\%&5.1\%&29.3\%&20.3\%&14.0\%&12.0\%&13.3\%\\
  &Theo. & 38.7\%&20.9\%&9.0\%&5.0\%&3.2\%&26.8\%&17.0\%&10.3\%&7.6\%&6.3\%\\
\multicolumn{12}{l}{$n = 46,800$, $\xi^2 = 0.01$}  \\                                               \multirow{2}{*}{Model 1} & Emp. & 8.1\%&5.8\%&3.5\%&3.0\%&2.5\%&11.6\%&9.4\%&9.9\%&9.2\%&12.5\%\\
  &Theo. & 6.9\%&5.4\%&2.6\%&1.4\%&0.9\%&9.9\%&8.4\%&6.0\%&5.0\%&4.5\%\\
\multirow{2}{*}{Model 2} & Emp. & 25.5\%&15.5\%&8.7\%&5.9\%&5.7\%&22.5\%&17.2\%&15.3\%&15.8\%&15.5\%\\
  &Theo. & 21.5\%&12.4\%&5.8\%&3.5\%&2.4\%&18.2\%&12.1\%&8.0\%&6.4\%&5.6\%\\
\multirow{2}{*}{Model 3} & Emp. & 38.1\%&20.0\%&8.3\%&2.3\%&1.7\%&29.6\%&18.3\%&11.7\%&10.3\%&11.2\%\\
  &Theo. & 38.7\%&20.9\%&9.0\%&5.0\%&3.2\%&26.8\%&17.0\%&10.3\%&7.6\%&6.3\%\\
\multicolumn{12}{l}{$n = 23,400$, $\xi^2 = 0.001$}                                                                                                                         \\
\multirow{2}{*}{Model 1} & Emp. & 17.2\%&15.8\%&12.9\%&11.7\%&11.2\%&21.7\%&20.9\%&19.4\%&19.7\%&20.1\%\\
  &Theo. & 6.9\%&5.4\%&2.6\%&1.4\%&0.9\%&9.9\%&8.4\%&6.0\%&5.0\%&4.5\%\\
\multirow{2}{*}{Model 2} & Emp. &30.7\%&21.8\%&14.8\%&12.1\%&11.2\%&28.7\%&23.3\%&19.7\%&19.2\%&19.0\%\\
  &Theo. &21.5\%&12.4\%&5.8\%&3.5\%&2.4\%&18.2\%&12.1\%&8.0\%&6.4\%&5.6\%\\
\multirow{2}{*}{Model 3} & Emp. & 51.1\%&33.0\%&20.6\%&16.2\%&14.7\%&43.4\%&32.8\%&26.1\%&23.8\%&23.5\%\\
  &Theo. & 38.7\%&20.9\%&9.0\%&5.0\%&3.2\%&26.8\%&17.0\%&10.3\%&7.6\%&6.3\%\\
\multicolumn{12}{l}{$n = 46,800$, $\xi^2 = 0.001$}  \\                                               \multirow{2}{*}{Model 1} & Emp. & 15.3\%&13.9\%&10.9\%&9.6\%&9.1\%&20.0\%&19.0\%&16.7\%&16.3\%&16.4\%\\
  &Theo. & 6.9\%&5.4\%&2.6\%&1.4\%&0.9\%&9.9\%&8.4\%&6.0\%&5.0\%&4.5\%\\
\multirow{2}{*}{Model 2} & Emp. & 29.7\%&20.6\%&13.3\%&10.4\%&9.2\%&28.2\%&22.4\%&17.8\%&16.5\%&15.8\%\\
  &Theo. & 21.5\%&12.4\%&5.8\%&3.5\%&2.4\%&18.2\%&12.1\%&8.0\%&6.4\%&5.6\%\\
\multirow{2}{*}{Model 3} & Emp. & 47.6\%&29.2\%&16.9\%&12.6\%&11.0\%&38.3\%&26.8\%&20.6\%&17.6\%&17.2\%\\
  &Theo. & 38.7\%&20.9\%&9.0\%&5.0\%&3.2\%&26.8\%&17.0\%&10.3\%&7.6\%&6.3\%\\
\multicolumn{12}{l}{$n = 23,400$, $\xi^2 = 0.0002$}                                                                                                                         \\
\multirow{2}{*}{Model 1} & Emp. & 25.2\%&23.8\%&20.6\%&19.4\%&18.6\%&32.8\%&31.7\%&30.3\%&30.0\%&30.5\%\\
  &Theo. & 6.9\%&5.4\%&2.6\%&1.4\%&0.9\%&9.9\%&8.4\%&6.0\%&5.0\%&4.5\%\\
\multirow{2}{*}{Model 2} & Emp. &45.5\%&35.6\%&28.2\%&25.7\%&24.5\%&46.2\%&40.5\%&37.0\%&36.4\%&36.8\%\\
  &Theo. &21.8\%&12.6\%&5.9\%&3.6\%&2.5\%&18.4\%&12.2\%&8.0\%&6.4\%&5.6\%\\
\multirow{2}{*}{Model 3} & Emp. &64.3\%&45.3\%&32.5\%&28.4\%&26.2\%&56.0\%&46.4\%&40.8\%&39.5\%&39.4\%\\
  &Theo. &38.1\%&20.5\%&8.8\%&4.9\%&3.1\%&26.6\%&16.9\%&10.2\%&7.5\%&6.2\%\\
\multicolumn{12}{l}{$n = 46,800$, $\xi^2 = 0.0002$}  \\                                               \multirow{2}{*}{Model 1} & Emp. &19.7\%&18.1\%&14.8\%&13.6\%&12.7\%&24.9\%&23.6\%&21.4\%&21.0\%&20.7\%\\
  &Theo. &6.9\%&5.4\%&2.6\%&1.4\%&0.9\%&9.9\%&8.5\%&6.0\%&5.0\%&4.5\%\\
\multirow{2}{*}{Model 2} & Emp. &38.9\%&29.0\%&22.2\%&19.8\%&18.1\%&37.7\%&31.8\%&28.5\%&27.6\%&27.0\%\\
  &Theo. &21.8\%&12.6\%&5.9\%&3.6\%&2.5\%&18.4\%&12.2\%&8.0\%&6.4\%&5.6\%\\
\multirow{2}{*}{Model 3} & Emp. &57.3\%&38.3\%&26.3\%&22.1\%&19.6\%&47.3\%&37.1\%&31.4\%&29.5\%&28.1\%\\
  &Theo. &38.1\%&20.5\%&8.8\%&4.9\%&3.1\%&26.6\%&16.9\%&10.2\%&7.5\%&6.2\%\\
\bottomrule
\end{tabular}}
\scriptsize $^\dag$Empirical losses $\breve{L}_B^{(\Sigma)}$ and theoretical losses $\tilde{L}_B^{(\Sigma)}$ for the three models and the 10 estimators. Two levels of sampling $n = 23,400$, $n = 46,800$ and three noise-to-signal ratios $\xi^2 = 0.01$, $\xi^2 = 0.001$ and $\xi^2=0.0002$ are considered.
\end{table}

\begin{table}
\centering
\caption{Estimates of $\rho$, AVAR ratio estimates and empirical correlation of corrections.$^\dag$}
\label{tableRho2}
\begin{tabular}{lcccc}
\toprule
\toprule
$B$   & $\widehat{\rho}_{B}$ & $AVAR_B^{(QMLE)}/AVAR_1^{(QMLE)}$ &$AVAR_B^{(RK)}/AVAR_1^{(RK)}$& $\widehat{\textnormal{Corr}}(\tilde{Q}_B - Q, \tilde{K}_B - K)$ \\
\toprule
1 & 0.74  & 1  & 1 &  -\\
2 & 0.8  & 0.96  & 0.97 & 0.689\\
4 & 0.84  & 0.92  & 0.94 & 0.769\\
6 & 0.85  & 0.91  & 0.93 & 0.868\\
8 & 0.86  & 0.9  & 0.92 & 0.879\\

\bottomrule
\end{tabular}

\scriptsize $^\dag$For $B = 1,2,4,6,8$, $\widehat{\rho}_B$ refers to the empirical mean value of estimates of $\rho$ on blocks $[\Tau_{i-1}, \Tau_i]$ across days and values of $i$ for INTC in 2015. The AVAR ratios are estimated by plugging estimates of the integrated volatility, the integrated quarticity and $\rho$ on blocks of different sizes. The last column shows the empirical correlation between the corrections induced by the local method.    
\end{table}

\begin{table}[]
\centering
\caption{Summary statistics for the global and local estimators$^\dag$}
\label{tableEstEmp}
\begin{tabular}{lcccc}
\toprule
\toprule
Estimator   & Mean & Stdv. & $\widehat{\textnormal{Corr}}(.,Q)$  \\
\toprule
$Q$ & 1.771  & 1.789& 1  \\
$\tilde{Q}_2$ & 1.770  & 1.781  & $\approx$ 1\\
$\tilde{Q}_4$ & 1.766  & 1.769  & $\approx$ 1  \\
$\tilde{Q}_6$ & 1.761  & 1.762  & 0.9999  \\
$\tilde{Q}_8$ & 1.757  & 1.753  & 0.9999  \\
$K$ & 1.818  & 1.795  & 0.9994  \\
$\tilde{K}_2$ & 1.818  & 1.780  & 0.9992\\
$\tilde{K}_4$ & 1.813  & 1.770  & 0.9991  \\
$\tilde{K}_6$ & 1.808  & 1.756  & 0.9989  \\
$\tilde{K}_8$ & 1.804  & 1.751  & 0.9988  \\

\bottomrule
\end{tabular}

\scriptsize $^\dag$Sample means, standard deviations, and correlations with the global QMLE for the 10 estimators implemented for INTC data in 2015. The estimators are scaled by a factor $10^4$. 

\end{table}

\newpage
\newpage
\clearpage

\begin{figure}
\includegraphics[width=1\linewidth]{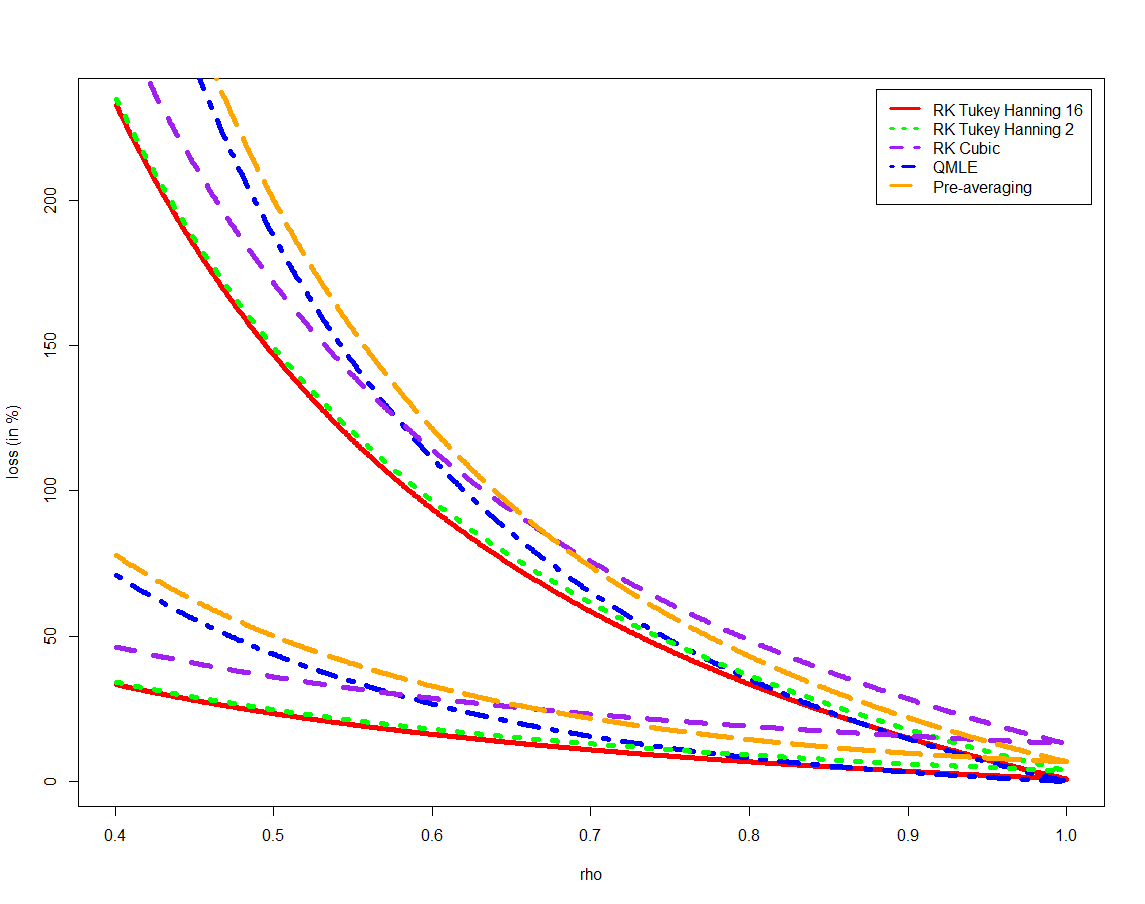}
\centering
\caption{Feasible loss region for three typical RK (Tukey-Hanning 16, Tukey-Hanning 2, Cubic), the QMLE and the PAE with triangle kernel. For each estimator, the lower line corresponds to the lower boundary when considering the best possible scenario $\kappa = \rho^{1/2}$ and the upper line stands for the upper boundary in the worst case scenario $\kappa = \rho^{3/2}$. The feasible loss region lies between those two lines. Note that a loss of 100 \% corresponds to an AVAR twice as big as the bound of efficiency.}
\label{efficiencyGlob}
\end{figure}

\begin{figure}
\includegraphics[width=1\linewidth]{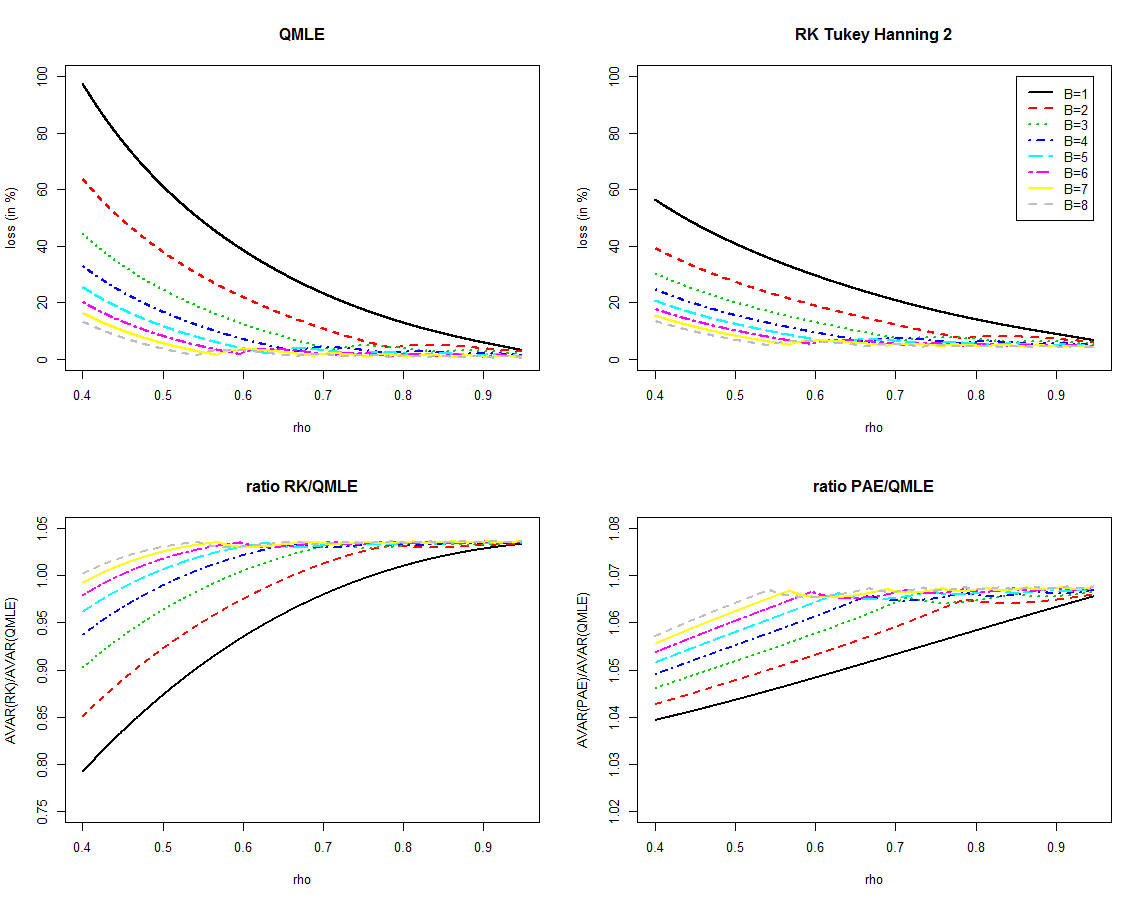}
\centering
\caption{For $B=1, \cdots,8$ we plot $L_B^{(QMLE)}$ (upper left panel), $L_B^{(RK)}$ for Tukey-Hanning 2 kernel (upper right panel), the corresponding AVAR ratio defined as $AVAR_{B}^{(RK)}/AVAR_{B}^{(QMLE)}$ (lower left panel) and the ratio of pre-averaging AVAR using $B$ blocks over $AVAR_{B}^{(QMLE)}$ (lower right panel) as a function of $\rho$.}
\label{AVARblockmod2}
\end{figure}

\begin{figure}
\includegraphics[width=1\linewidth]{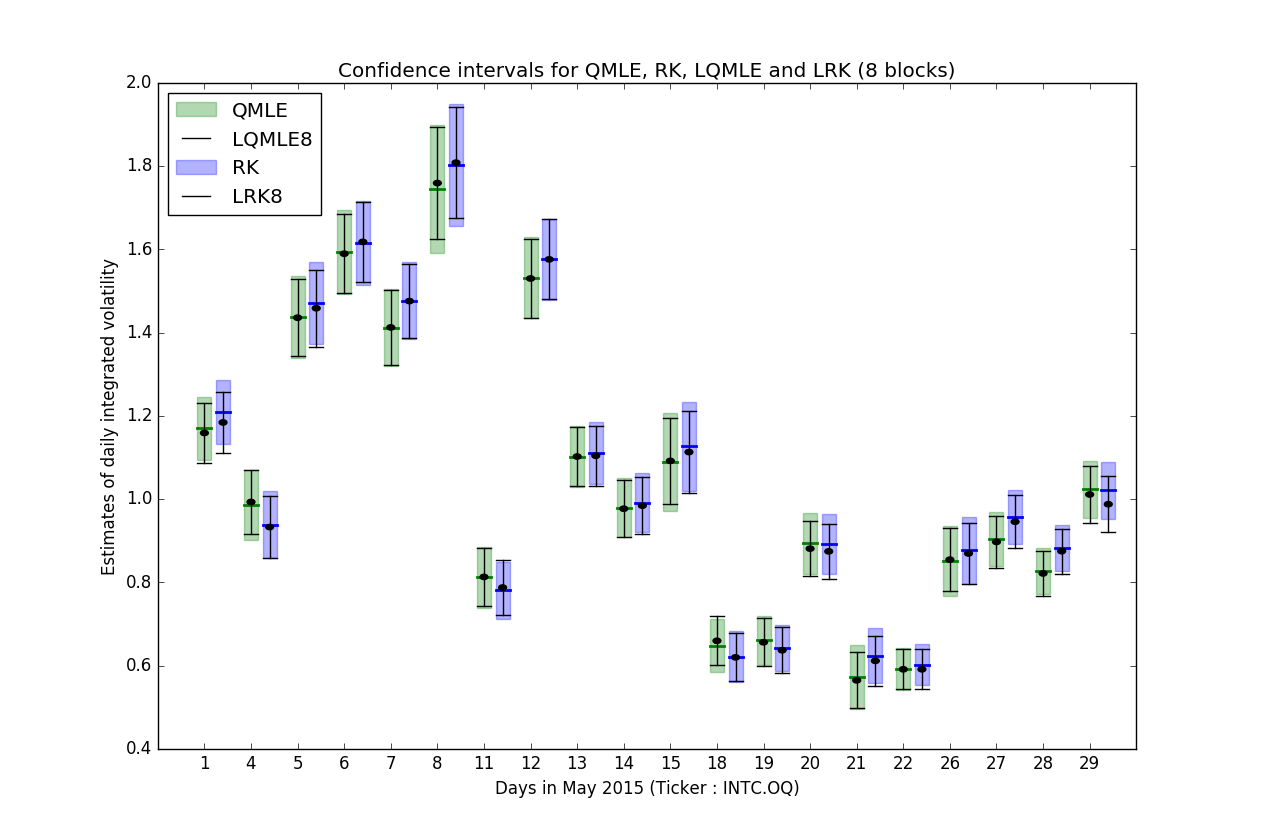}
\centering
\caption{95\% Confidence intervals for the four estimators $Q$ and $\tilde{Q}_8$ (green, left), $K$ and $\tilde{K}_8$ (blue, right) on INTC data in May 2015. The CIs are computed using the estimates of $AVAR_{B}^{(QMLE)}$, and $AVAR_{B}^{(RK)}$ for $B =1,8$ obtained as explained in Section \ref{empiricalIllustration}. The estimators are scaled by a factor $10^4$. }
\label{CI}
\end{figure}

\end{document}

%% file: nakamacro27.tex
\newtheorem{remark}{Remark}

\def\bd{\begin{description}}
\def\ed{\end{description}}

\def\D2{\bbD_{2,\infty-}}

\def\D{{\bf D}}

\def\I{{\bf I}}

\def\calf{{\cal F}}
\def\calg{{\cal G}}
\def\calh{{\cal H}}

\def\call{{\cal L}}
\def\calm{{\cal M}}
\def\caln{{\cal N}}

\def\calq{{\cal Q}}

\def\calu{{\cal U}}

\def\half{\frac{1}{2}}

\def\be{\begin{equation}}
\def\ee{\end{equation}}
\def\bea{\begin{eqnarray}}
\def\eea{\end{eqnarray}}
\def\beas{\begin{eqnarray*}}
\def\eeas{\end{eqnarray*}}
\def\bi{\begin{itemize}}
\def\ei{\end{itemize}}

\def\bd{\begin{description}}
\def\ed{\end{description}}
\def\l{\left}
\def\r{\right}

\newcommand{\bbD}{{\mathbb D}}

\newcommand{\reels}{\mathbb{R}}
\newcommand{\naturels}{\mathbb{N}}

\newcommand{\esp}{\mathbb{E}}
\newcommand{\proba}{\mathbb{P}}

\newcommand{\inv}[1]{\frac{1}{#1}}